\documentclass[11pt,a4paper]{article}
\usepackage{amsmath,amssymb,enumerate,graphicx,xcolor}
\usepackage{color}
\usepackage[cm]{fullpage}

\usepackage{subcaption}

\usepackage{cases}
\usepackage[hypertexnames=false]{hyperref}
\usepackage{amsthm} 

\newcommand\enc{\mathop{\mathrm{Enc}}}
\newcommand\enco{\mathop{\mathrm{Enc_0}}}

\newcommand \allp{\Sigma}
\newcommand \oddp{\Omega}
\newcommand \conp{\Gamma}
\newcommand \balp{{\mathbf B}}
\newcommand \dimp{\Upsilon}

\newtheorem{theorem}{Theorem}[section]
\newtheorem{proposition}[theorem]{Proposition}
\newtheorem{observation}[theorem]{Observation}
\newtheorem{lemma}[theorem]{Lemma}
\newtheorem*{lemma*}{Lemma}

\newtheorem{construction}[theorem]{Construction}
\newtheorem{corollary}[theorem]{Corollary}

\newtheorem*{remark}{Remark}
\newtheorem{claim}{Claim}

\newtheorem*{theoremMain}{Main Theorem}

\numberwithin{equation}{section}

\newcounter{tmpc}
\newcounter{saveenumi}

\makeatletter
\@addtoreset{footnote}{page}
\makeatother
\renewcommand{\thefootnote}{\ifcase\value{footnote}\maltese\or(*)\or(**)\or(***)\or(****)\or(\#)\or(\#\#)\or(\#\#\#)\or(\#\#\#\#)\or($\infty$)\fi}

\usepackage{authblk}

\title{\bf Generalized Gray codes with prescribed ends}

\author{Tom\'{a}\v{s} Dvo\v{r}\'ak, Petr Gregor and V\'aclav Koubek$^{\maltese}$}
\affil{Faculty of Mathematics and Physics\\
Charles University in Prague\\
Czech Republic}

\date{\vspace{-5ex}}

\begin{document}

\maketitle

\begin{abstract}
An $n$-bit Gray code is a sequence of all $n$-bit vectors such that consecutive vectors differ in a~single bit. It is well-known that given $\alpha,\beta\in\{0,1\}^n$, an $n$-bit Gray code between $\alpha$ and $\beta$ exists iff the Hamming distance $d(\alpha,\beta)$ of $\alpha$ and $\beta$ is odd. We generalize this classical result to $k$ pairwise disjoint pairs $\alpha_i, \beta_i\in\{0,1\}^n$:  if $d(\alpha_i,\beta_i)$ is odd for all $i$ and $k<n$, then the set of all $n$-bit vectors can be partitioned into $k$ sequences such that the $i$-th sequence leads from $\alpha_i$ to $\beta_i$ and consecutive vectors differ in a single bit. This holds for every $n>1$ with one exception in the case when $n = k + 1 = 4$. Our result is optimal in the sense that for every $n>2$ there are $n$ pairwise disjoint pairs $\alpha_i,\beta_i\in\{0,1\}^n$ with $d(\alpha_i,\beta_i)$ odd for which such sequences do not exist. 

\end{abstract}

\footnotetext{Deceased 12 February 2016.}


\section{Introduction}

An $n$-bit Gray code is a sequence of all $2^n$ binary vectors of length $n$ such that consecutive vectors differ in a~single coordinate. The code is named after Frank Gray who in 1953 patented a simple algorithm to generate such a code for every $n\ge1$. Gray codes have found applications in such diverse areas as information retrieval, signal encoding, image processing or data compression, and alternative constructions of Gray codes satisfying certain additional properties have been widely studied~\cite{Knuth,Sa}.

It is well-known that  given $\alpha,\beta\in\{0,1\}^n$, an $n$-bit Gray code between $\alpha$ and $\beta$ exists iff the Hamming distance $d(\alpha,\beta)$ of $\alpha$ and $\beta$ is odd \cite{H}. Caha and Koubek suggested the  following generalization: A set of $k$ pairwise disjoint pairs $\alpha_i ,\beta_i\in\{0,1\}^n$ is called \emph{connectable} if  there are $k$ sequences of vectors of $\{0,1\}^n$ such that
\begin{enumerate}[\upshape(1)]
\item \label{connectable:i}
the $i$-th sequence leads from $\alpha_i$ to $\beta_i$ for every $i=1,\dots,k$,
\item \label{connectable:ii}
consecutive vectors differ in a single bit,
\item \label{connectable:iii}
each vector of $\{0,1\}^n$ occurs in exactly one of these sequences.
\end{enumerate}
In \cite{CK}  they showed that $\{\alpha_i,\beta_i\}_{i=1}^k$ with $d(\alpha_i,\beta_i)$ odd is connectable for every $n\ge2$ provided that $k\le(n-1)/3$, and asked about the maximum value of $k$ for which this statement holds.

The connectability of $\{\alpha_i,\beta_i\}_{i=1}^k$ for $k=2$ follows from earlier works \cite{LW,D,CK}, while instances of the problem for $k$ bounded by a small constant were systematically studied by Casta\~neda and Gotchev  \cite{CG,CGL}. The authors of the present paper characterized the connectability of  $\{\alpha_i,\beta_i\}_{i=1}^k$ for $k\le(n-1)/2$ \cite{GD} and showed that when $k$ is unbounded, the problem becomes NP-hard or NP-complete, depending on the instance description \cite{DG, DK}.

As observed in \cite{DGK05}, for every $n>2$ there are $n$ pairwise disjoint pairs $\alpha_i,\beta_i\in\{0,1\}^n$ with $d(\alpha_i,\beta_i)$ odd for all $i$ which are not  connectable.  Indeed, let the $j$-th coordinate of $\alpha_i$ be 1 iff $i=j$, $\alpha$ be the zero vector, and $\{\beta_1,\dots,\beta_n\}\subseteq\{0,1\}^n\setminus\{\alpha_1,\dots,\alpha_n,\alpha\}$ be arbitrary. If $\{\alpha_i,\beta_i\}_{i=1}^n$ were connectable, then --- by condition \eqref{connectable:iii} in the definition of connectability --- $\alpha$ would be included in a sequence between $\alpha_i$ and $\beta_i$ for some $i$. But --- by condition \eqref{connectable:ii} --- this sequence has to pass through some $\alpha_j$, $j\ne i$, which means that $\alpha_j$ is included in both $i$-th and $j$-th sequences, contrary to  \eqref{connectable:iii}.

The goal of this paper is to derive the maximum $k$ for which  $\{\alpha_i,\beta_i\}_{i=1}^k$ with $d(\alpha_i,\beta_i)$ odd for all $i$ is connectable. The previous paragraph implies that $k<n$. Our main result shows this bound is also sufficient, thus providing the optimal solution to this problem.
\begin{theoremMain}
Let $A=\{\alpha_i,\beta_i\}_{i=1}^k\subseteq\{0,1\}^n$,  $d(\alpha_i,\beta_i)$ be odd for every $i=1,2,\dots,k$, and $k<n>0$.
Then $A$ is connectable except the case when $n=4$ and $A$ is  isomorphic to the set $C_2$  on Fig.~\ref{fig:Fig2}.
\end{theoremMain}
The rest of this paper is organized as follows. The next section introduces notation and terminology, surveys previous results, and deals with the problem for $n\le5$. Section~\ref{section:induction-step} describes a~non-deterministic algorithm (Construction~\ref{construction-simple}) which forms a cornerstone of our inductive construction, with induction step summarized as Proposition~\ref{prop-diminishable}. Various special cases not covered by the  general induction step are settled in Section~\ref{section:special-cases}.

Note that to derive the main theorem which deals with at most $n-1$ odd pairs, we need to prove a~rather stronger statement. Thus our auxiliary results consider also $n$ pairs at both odd and even distances.
\section{Preliminaries}
Each $n$-bit Gray code may be viewed as a Hamiltonian path in the graph of the $n$-dimensional hypercube $Q_n$. In the rest of this paper we therefore resort to this equivalent formulation and use standard graph-theoretic terminology \cite{Bo}.

Hypercubes are graphs on $(0,1)$-vectors.  We first recall basic definitions of hypercubes and describe some tools that may be useful while working with vertices of a hypercube.

Let $n\ge1$ be a natural number and $[n]=\{0,1,\dots,n-1\}$. Let $V_n$ be the set of all $n$-dimensional $(0,1)$-vectors, i.e.~the set of all functions from $[n]$ into $[2]=\{0,1\}$.  For $\alpha\in V_n$ and $i\in[n]$, let $\alpha (i)$ denote the $i$-th coordinate of $\alpha$.
For vectors $\alpha ,\beta
\in V_n$, define $(\alpha\oplus\beta )(i)=\alpha (i)\oplus\beta (i)$ for all $i\in[n]$ where for
$x,y\in[2]$ we have
$$x\oplus y=\begin{cases} 0&\text{ if }x=y,\\
1&\text{ if }x\ne y.\end{cases} $$
Then $\oplus$ is a group operation of order $2$ on the set $V_n$ with the neutral element $\epsilon$ where $\epsilon (i)=0$ for all $i\in[n]$.  Let $e^n_i\in V_n$ for $i\in[n]$ be vectors such that
$$e^n_i(j)=\begin{cases} 1&\text{ if }i=j,\\
0&\text{ if }i\ne j.\end{cases} $$
If $n$ is clear from the context then we omit $n$ and write only $e_i$.

A \emph{hypercube} $Q_n$ of dimension $n$ is a graph with vertex set $V_n$, $\{\alpha ,\beta \}$ being an edge  when $\alpha\oplus\beta =e_i$ for some $i\in[n]$. The set of all edges of $Q_n$ is denoted by $E(Q_n)$.  

In the following, we often need to transform $n$-dimensional vectors into  dimensions $n-1$ or $n+1$. For this purpose, we employ the following operations. For $i\in[n]$ and $\alpha\in V_n$ such that $\alpha (i)=k$ we define $\rho_{i=k}(\alpha )\in V_{n-1}$  as follows:
$$\rho_{i={k}}(\alpha )(j)=\begin{cases} \alpha (j)&\text{ if }j<i,\\
\alpha (j+1)&\text{ if }j\ge i\end{cases} $$
while $\rho_{i=k}(\alpha)$ is undefined if $\alpha(i)\ne k$.
Given a set $X\subseteq V_n$, $i\in[n]$ and $k\in[2]$, put $\rho_{i=k}(X)=\{\rho_{i=k}(\alpha )\mid\alpha\in X\}.$
Conversely, for a vector $\alpha\in V_{n-1}$, $i\in[n]$ and $k\in[2]$, let $\iota_{i=k}(\alpha )$ be a vector of $V_n$ such that
$$\iota_{i=k}(\alpha )(j)=\begin{cases} \alpha (j)&\text{ if }j<i,\\
k&\text{ if }j=i,\\
\alpha (j-1)&\text{ if }j>i.\end{cases} $$
For $\alpha\in V_n$ we define the \emph{parity} $\chi(\alpha)$ of $\alpha$ by $\chi (\alpha)=\prod_{i\in[n]}-1^{\alpha(i)}$.
\par
Let $P_2(V_n)=\{\{\alpha ,\beta \}\mid\alpha ,\beta\in V_n\}$ be the set of all multisets consisting of two elements of $V_n$.  Note that elements of $P_2(V_n)$ --- which we call \emph{pairs} --- may consist of either two different or two identical elements of  $V_n$, both cases are needed throughout the paper.
For $\{\alpha,\beta\}\in A\in P_2(V_n)$ let $\chi(\alpha,\beta)=\chi(\alpha)+\chi(\beta)$ and  $\chi(A)=\sum_{ \{\alpha ,\beta \}\in A}\chi(\alpha,\beta)$.
Note that $\chi(\alpha,\alpha)$ is either $2$ (if $\sum_{i\in[n]}\alpha(i)$ is even) or $-2$ (if $\sum_{i\in[n]}\alpha(i)$ is odd).
We say that a pair $\{\alpha ,\beta \}\in P_2(V_n)$ is
\begin{description}
\item \emph{even} if $\chi(\alpha)=\chi(\beta)$,
\item \emph{odd} if $\chi(\alpha)\ne\chi(\beta)$,
\item \emph{degenerated} if $\alpha =\beta$,
\item \emph{edge}-\emph{pair} if $\alpha\oplus\beta =e_i$ for some $i\in[n]$.
\end{description}
Note that  $\{\alpha ,\beta \}$ is an edge-pair exactly when $\{\alpha ,\beta \}\in E(Q_n)$.  A subset $A\subseteq P_2(V_n)$ is called a \emph{pair}-\emph{set} in $Q_n$ if
\begin{itemize}
\item $\{\alpha ,\beta \}\cap \{\alpha',\beta'\}=\emptyset$ for all distinct $
\{\alpha ,\beta \},\{\alpha',\beta'\}\in A$, and
\item if $A\ne\emptyset$ then there exists $\{\alpha ,\beta \}\in A$ with $\alpha\ne\beta$.
\end{itemize}
Let $\allp_n$ be the set of all pair-sets in $Q_n$.  For $A\in\allp_n$,
\begin{itemize}
\item $|A|$ is the size of $A$,
\item $\|A\|$ denotes the number of odd pairs in $A$, and
\item $\bigcup A=\bigcup_{\{\alpha ,\beta \}\in A}\{\alpha ,\beta \}$.
\end{itemize}

Thus
$|A|-\|A\|$ is the number of even pairs in $A$.  For a positive integer $k$, let $\allp_n^k$ be the set of all pair-sets $A\in\allp_n$ with $|A|\le k$.

We say that a vertex $\alpha\in V_n$ is \emph{encompassed} by a set
$X\subseteq V_n$ if for every edge $\{\alpha ,\beta \}\in E(Q_n)$ we have $
\beta\in X$. Let $\enco(X)$ be the set of all vertices encompassed by $X$ and $\enc(X)=\enco(X)\setminus X$. Given a pair-set $A$, we set $\enco(A)=\enco(\bigcup A)$, $\enc(A)=\enco(A)\setminus\bigcup A$ and say that a vertex $\alpha$ is  \emph{encompassed by a pair-set} $A$ if $\alpha$ is encompassed by $\bigcup A$.

For a pair-set $A\in\allp_n$ and $i\in[n]$, let $\sigma_i(A)=(n_0,n_1)$ where $n_k$ is the number of $\{\alpha ,\beta \}\in A$ with $\alpha (i)=\beta (i)=k$ for $k\in[2]$. Clearly, $|A|=n_0+n_1+|\{\{\alpha ,\beta \}\in A\mid\alpha
(i)\ne\beta (i)\}|$
where $\sigma_i(A)=(n_0,n_1)$. Further, for $k\in[2]$ define
$$\rho_{i=k}(A)=\{\{\rho_{i=k}(\alpha ),\rho_{i=k}(\beta )\}\mid
\{\alpha ,\beta \}\in A,\,\alpha (i)=\beta (i)=k\}.$$
Note that $\rho_{i=k}(A)$ need not be a pair-set.
Indeed, if $\rho_{i=k}(A)\ne\emptyset$ consists only of degenerate pairs, then --- by the
definition --- $\rho_{i=k}(A)$ is not a pair-set. Observe that $\rho_{i=k}(\bigcup A)$ can be distinct from $\bigcup(\rho_{i=k}(A))$.
If $\sigma_i(A)=(n_0,n_1)$ and $\rho_{i=k}(A)$ is a pair-set, then $|\rho_{i=k}(A)|=n_k$ for $k\in[2]$.  Conversely, if $A_0,A_1\in\allp_{n-1}$ and $k\in[2]$, we define
$$\begin{aligned}\iota_{i,k}(A_0,A_1)=&\{\{\iota_{i=k}(\alpha ),\iota_{i=k}
(\beta )\}\mid \{\alpha ,\beta \}\in A_0\}\,\cup\\
&\{\{\iota_{i=1-k}(\alpha ),\iota_{i=1-k}(\beta )\}\mid \{\alpha
,\beta \}\in A_1\}.\end{aligned}$$
Clearly, $\iota_{i,k}(A_0,A_1)\in\allp_n$ such that
$$\rho_{i=k}(\iota_{i,k}(A_0,A_1))=A_0\text{ and }\rho_{i=1-k}(\iota_{
i=k}(A_0,A_1))=A_1.$$
If $A$ is a pair-set such that $n_0+n_1=|A|$ where $\sigma_i(A)=(n_0,n_1)$ for some $i\in[n]$, then $\iota_{i,0}(\rho_{i=0}(A),\rho_{i=1}(A))=A$.

We say that a~pair-set
$A\in\allp_n$ is
\begin{description}
\item \emph{pure} if $\alpha\ne\beta$ for all $\{\alpha ,\beta \}\in A$;
\item \emph{balanced} if $\chi (A)=0$;
\item \emph{odd} if every pair in $A$ is odd.
\end{description}
Let $\balp_n$ and $\oddp_n$ be the sets of all balanced and odd pair-sets from $\allp_n$, respectively, and  put $\balp_n^k=\balp_n\cap\allp_n^k$ and $\oddp_n^k=\oddp_n\cap\allp_n^k$ for every $k$.
Observe that if $A\in\oddp_n$, then $\rho_{i=k}(A)\in\oddp_{n-1}$ for all $i\in[n]$ and  $k\in[2]$.

\begin{observation}
\label{obser-encom}
If $A\in\balp_n^{2n-3}$ then $|\enco(A)|\le2$ and if $\alpha,\beta\in\enco(A)$ are
distinct then $\chi(\alpha)\ne\chi(\beta)$. Moreover, if $A\in\balp_n^{2n-4}$, $\gamma\in \enco(A)$ and $\{\alpha,\beta\}\in A$ with $\{\beta,\gamma\}\in E(Q_n)$ and $\beta'\in V_n\setminus\bigcup A$
such that $\chi(\beta')=\chi(\beta)$, then
$$
\enco(A\setminus\{\{\alpha,\beta\}\}\cup\{\{\alpha,\beta'\}\})=\enco(A)\setminus\{\gamma\}.
$$
\end{observation}
\begin{proof}
Recall that a vertex of $Q_n$ has exactly $n$ neighbors, while two vertices of the same parity share at most two common neighbors. Hence if distinct $\alpha,\beta\in V_n$ with $\chi(\alpha)=\chi(\beta)$ are encompassed by a set $X\subseteq V_n$, then $|\{\eta\in X\mid\chi(\eta)\ne\chi(\alpha)\}|\ge2n-2$. If $\enco(A)$ contains at least three vertices, then two of them must be of the same parity. Since $|\{\alpha\in\bigcup A\mid\chi(\alpha)=\ell\}|\le k$ for $A\in\balp_n^k$ and $\ell\in\{-1,1\}$, it follows that then $|A|\ge2n-2$, contrary to the assumption that $A\in\balp_n^{2n-3}$. This settles the first statement.
\par
The second statement follows from the following facts. First, $\gamma\notin\enco(\bigcup A\setminus\{\beta\})$. Second,
$\{\gamma'\in\enco(A)\mid \chi(\gamma')\ne\chi(\gamma)\}=\{\gamma'\in\enco(\bigcup A\cup\{\beta'\})\mid\chi(\gamma')\ne\chi(\gamma)\}$ because $\chi(\gamma')\ne\chi(\gamma)$ implies $\chi(\gamma')=\chi(\beta')$ and therefore $\beta'$ cannot be a neighbor of $\gamma'$. And finally,
$$|\{\eta\in\bigcup A\mid\{\eta,\gamma'\}\in E(Q_n)\}|\le |A|-(n-2)<n-1$$
for every $\gamma'\in V_n\setminus\{\gamma\}$ with $\chi(\gamma')=\chi(\gamma)$ because $\gamma\in\enco(A)$, which means that such $\gamma'$ cannot be encompassed by $\bigcup A\cup\{\beta'\}$.
\end{proof}
An odd pair-set $A\in\Omega_n$ called \emph{diminishable} if
\begin{itemize}
\item
$|A|\le n-1$ and if $n=4$ then
\begin{itemize}
\item
either $A$ contains an edge pair
\item
or there exists no subset $\bigcup A\subseteq X\subseteq V_4$ such that the  subgraph of $Q_4$ induced by $X$ is isomorphic to $Q_3$;
\end{itemize}
\item
$|A|=n$, $n\ne 4$, $A$ contains at least two edge-pairs, and $\enc(A)=\emptyset$.
\end{itemize}

Let $\dimp_n$ be the set of all diminishable pair-sets in $\allp_n$.
Given $A\in\dimp_n$  with $|A|=n$, we say that $i\in[n]$ is \emph{separating} for $A$ if there exist edge-pairs $\{\alpha ,\beta \},\{\alpha',\beta'\}\in A$ with $\alpha (i)=\beta (i)\ne\alpha'(i)=\beta'(i)$.

A~pair-set $A\in\allp_n$ is called
 \emph{connectable} if there exists a family $\{P_{\alpha ,\beta}\mid
\{\alpha ,\beta \}\in A\}$ of
vertex-disjoint paths in $Q_n$ such that $P_{\alpha ,\beta}$ is a path between
$\alpha$ and $\beta$ and for every vertex $\gamma\in V_n$ there is exactly one $
\{\alpha ,\beta \}\in A$ such
that path $P_{\alpha ,\beta}$  passes through $\gamma$. Then the family $\{P_{\alpha ,\beta}\mid \{\alpha ,\beta \}\in A\}$
is called a \emph{connector} of $A$.
Observe that if $\{P_{\alpha ,\beta}\mid \{\alpha ,\beta \}\in A\}$ is a connector of a pair-set $A$ and $\{\alpha ,\alpha \}\in A$, then $P_{\alpha ,\alpha}$ is a singleton path consisting of $\alpha$ because  it is the unique path from $\alpha$ into $\alpha$. Let $\conp_n$ denote the set of all connectable pair-sets from $\allp_n$.

Note that if there exists a vertex $\alpha\in\enc(A)$ for an odd pair-set $A$, then $A$ is not connectable, because $\chi(\beta)=\chi(\beta')$ for all $\{\alpha,\beta\},\{\alpha,\beta'\}\in E(Q_n)$, and therefore any path of length $>1$ passing through $\alpha$ visits two distinct pairs of $A$.
Hence $\Omega_n^n\setminus \Gamma_n\ne\emptyset$ for every $n\ge 3$. This argument fails for diminishable pairs, because if $\alpha\in\enco(A)$ then $\alpha\in\bigcup A$. In fact, our main result says that every diminishable pair-set $A\in\Upsilon_n$ is connectable.

If $P$ is a path in $Q_n$ between $\alpha$ and $\beta$ and $\gamma ,\delta\in V_n$ are vertices that belong to $P$, then we say that $\gamma$ \emph{is} \emph{closer} \emph{to} $\alpha$ \emph{than} $\delta$ \emph{in} $P$ if the subpath of $P$ between $\alpha$ and $\gamma$ does not contain $\delta$.

For $A,B\in\allp_n$ we write $A\implies B$ if there exist
$\{\alpha ,\beta \},\{\alpha',\beta'\}\in A$ and $i\in[n]$ such that $\beta\oplus e_i=\beta'$  and $B=(A\setminus \{\{\alpha ,\beta \},\{\alpha',\beta'\}\})\cup \{\{\alpha ,\alpha'\}\}$.
The transitive and reflexive closure of $\implies$ is denoted by $\overset {*}{\implies}$.

We say that a pair-set $A$ is an $i$-\emph{completion} of a pair-set $B$
if $A\overset{*}{\implies}B$ and $|A|=n_0+n_1$ where
$\sigma_i(A)=(n_0,n_1)$. Note that then
$A=\iota_{i,0}(\rho_{i=0}(A),\rho_{i=1}(A))$.

The following easy lemma motivates these notions and forms the cornerstone of our proofs.

\begin{lemma}
\label{lemma-simple}
Let $n\ge 1$ be a natural number. Then
\begin{enumerate}[\upshape (1)]
\item
if $A\in\conp_n$ and $A\overset {*}{\implies}B$ then $B\in\conp_
n$;
\item
if $n>1$ and $A,B\in\conp_{n-1}$, then $\iota_{i,k}(A,B)\in\conp_n$ for every $i\in[n]$ and $k\in[2]$;
\item
if $A$ is balanced and $B\overset{*}{\implies}A$ then $B$ is balanced;
\item
if $A\in\conp_{n}$ then $A$ is balanced;
\item
if $A\in\oddp_n$ then $A$ is balanced;
\item
if $A\in\dimp_n$ and $|A|=n$, then there exists $i\in[n]$ such that $i$ is separating for $A$;
\item
If $B\in\allp_n$ is an $i$-completion of $A$ and $\rho_{i=0}(B),\rho_{i=1}(B)\in\conp_{n-1}$, then $A\in\conp_n$. 
\end{enumerate}
\end{lemma}

Observe that if $A\in\allp_n$ is balanced, then there exists a perfect matching $R$ on even pairs from $A$ such that  $\{\{\alpha ,\beta \},\{\alpha',\beta'\}\}\in R$ implies $\chi(\alpha,\beta)+\chi(\alpha',\beta')=0$. Note that then $|R|=\frac {|A|-\|A\|}2$. Such a perfect matching is called \emph{$(A,i)$-matching} for $i\in[n]$ if for every pair of degenerated pairs $\{\{\alpha ,\alpha \},\{\beta ,\beta \}\}\in R$ we have $\alpha (i)=\beta (i)$.

Finally, we recall several known results that shall be useful for our main theorem.

\begin{proposition}
\label{old-results-prop}
Let $n\ge 1$ be a natural number. Then
\begin{enumerate}[\upshape(1)]
\item{\rm\cite{H}}
\label{old-results-prop-part1}
a singleton pair-set $A$ belongs to $\conp_n$ if and only if $A$ is odd; \item{\rm\cite{LW}}
\label{old-results-prop-part2}
if $A=\{\{\alpha ,\beta \},\{\gamma ,\gamma \}\}\in\balp_n$, then $A\in\conp_n$;
\item{\rm\cite{D}}
\label{old-results-prop-part3}
$\oddp_n^2\subseteq\conp_n$;
\item{\rm\cite{CK}}
\label{old-results-prop-part5}
if $A\in\balp^2_n$ is pure with $||A||=0$ and $n\ge 4$, then $A\in\conp_n$;
\item{\rm\cite{CK}}
\label{old-results-prop-part6}
if $A\in\oddp^3_n$ and $n\ge 5$, then $A\in\conp_n$;
\item{\rm\cite{CG,CGL}}
\label{old-results-prop-part7}
if $A\in\balp_n^3$
with  $n\ge 4$
and $\|A\|=1$, then $A\in\conp_n$;
\item{\rm\cite{CG,CGL}}
\label{old-results-prop-part8}
if $A\in\balp_n^4$ is not pure and $n\ge 5$,
then $A\in\conp_n$;
\item{\rm\cite{CG}}
\label{old-results-prop-part9}
 if $A\in\balp_n^4$, $\|A\|=2$, $n\ge4$, and $A$ has two degenerated pairs then $A\in\conp_n$;
\item{\rm\cite{CG}}
\label{old-results-prop-part10}
 if $A\in\balp_n^5$, $n\ge5$ and $A$ has four degenerated pairs then $A\in\conp_n$;
\item{\rm\cite{GD}}
\label{old-results-prop-part11}
 $\balp_n^{\lfloor\frac{n-1}2\rfloor}\subseteq\conp_n$ for every $n\ge 1$;
\item{\rm\cite[Lemma~5.2]{CG}}
\label{old-results-prop-part12}
 if $A\subseteq V_n$ is a set of $k$ vertices of the same parity, $1\le k\le n$, then $|\{\alpha\in V_n\mid\{\alpha,\beta\}\in E(Q_n)\text{ for some } \beta\in A\}|\le1+kn-\binom{k+1}2$. 
\end{enumerate}
\end{proposition}

At the end of this section we state several folklore and simple results on connectable pair-sets in hypercubes of dimension at most $5$. Note that  these statements deal only with pair-sets of small dimensions, and therefore we were able to verify them by a computer search. Formal proofs are --- due to their technicality and length --- moved to another paper \cite{DGK1}.

We say that pair-sets $A,B\in\allp_n$ are \emph{isomorphic} if there
exists an automorphism $f$ of a hypercube $Q_n$ such that
$B=\{\{f(\alpha),f(\beta)\}\mid\{\alpha,\beta\}\in A\}$. The following Proposition~\ref{dimen-3} may be verified by inspection.

\begin{proposition}
\label{dimen-3}
Let $A\in\balp_3$. Then
\begin{enumerate}[\upshape(1)]
\item
\label{dimen-3-1}
if $|A|=2$ then $A$ is connectable if and only if $A$ is not isomorphic to the pair-sets $C_0$ and $C_1$
on Fig.~\ref{fig:Fig1};
\item
\label{dimen-3-2}
if $A$ is diminishable then $A$ is connectable.
\end{enumerate}
\end{proposition}

\begin{proposition}
\label{dimen-4}
Let $A\in\oddp_4$ be a pair-set. Then
\begin{enumerate}[\upshape(1)]
\item
\label{dimen-4-1}
if $|A|\le 3$ then $A$ is connectable if and only if $A$ is not isomorphic to the pair-set $C_2$ on Fig.~\ref{fig:Fig2};
\item
\label{dimen-4-2}
if $|A|=4$, $A$ contains at least three edge pairs and $\enc(A)=\emptyset$, then $A$ is connectable.
\end{enumerate}
\end{proposition}

\begin{figure}[!h]
\centering
\hspace{-1cm}%
\scalebox{0.8}{\includegraphics{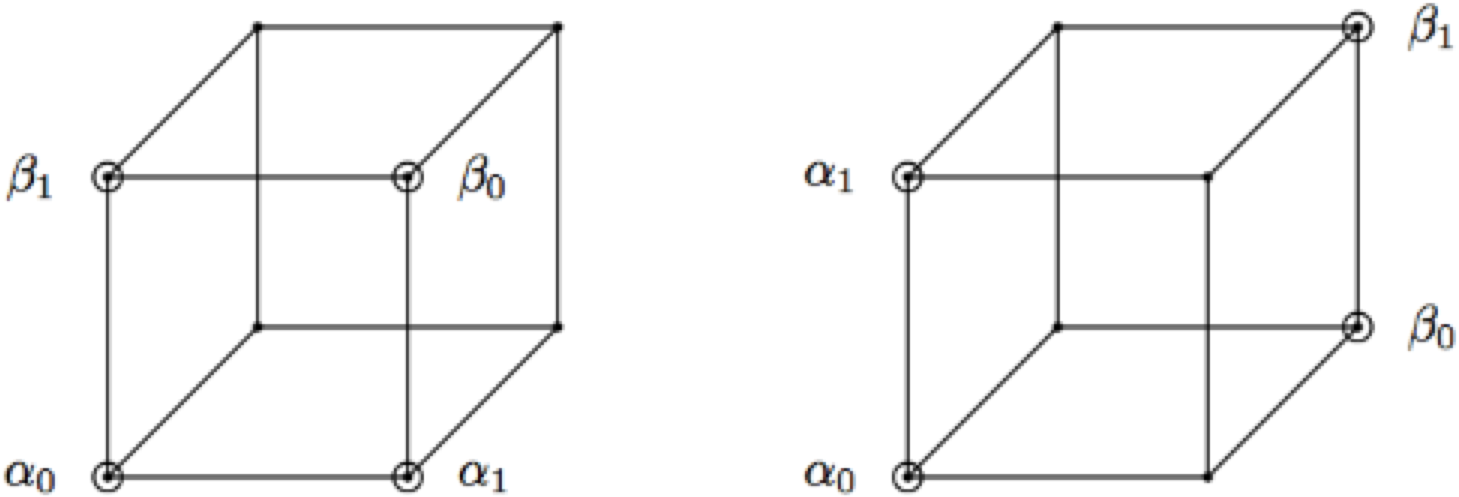}}
\begin{minipage}[t]{.23\linewidth}
\centering
\subcaption*{$C_0=\{\{\alpha_0,\beta_0\},\{\alpha_1,\beta_1\}\}$}\label{fig:C0}
\end{minipage}%
\begin{minipage}[t]{.38\linewidth}
\centering
\subcaption*{$C_1=\{\{\alpha_0,\beta_0\},\{\alpha_1,\beta_1\}\}$}\label{fig:C1}
\end{minipage}
\caption{The only non-connectable balanced pair-sets in $Q_3$\label{fig:Fig1}}
\end{figure}

\begin{figure}
\centering
\hspace{-1cm}%
\begin{minipage}[t]{0.5\linewidth}
\centering
\scalebox{0.8}{\includegraphics{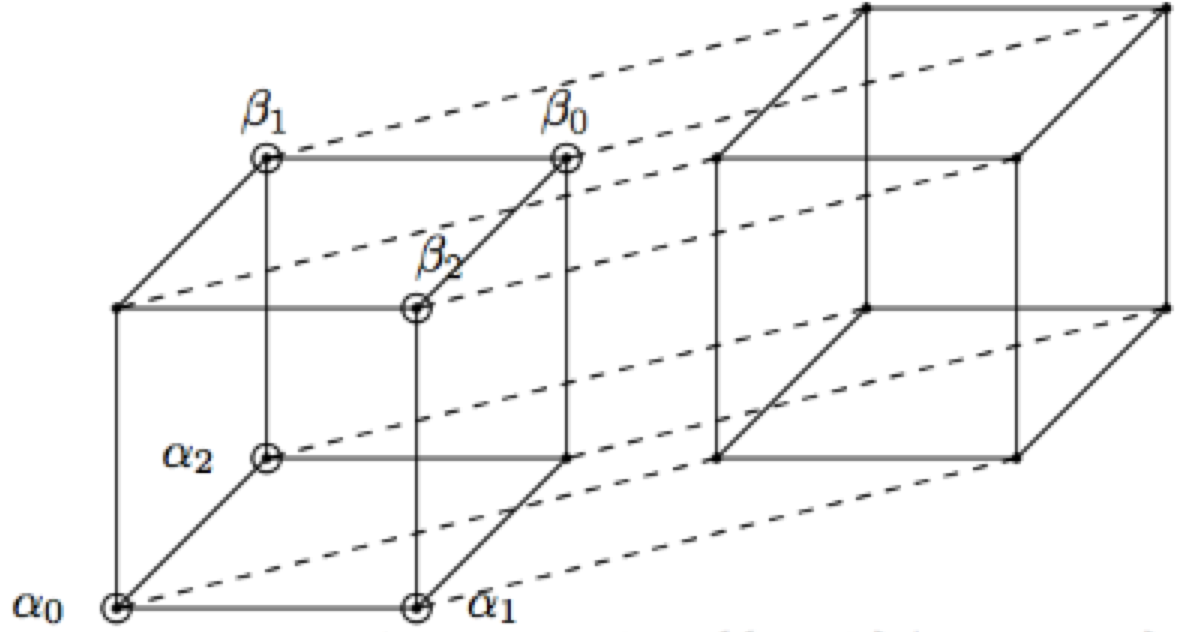}}
\subcaption*{$C_2=\{\{\alpha_i,\beta_i\}\mid i=0,1,2\}$}\label{fig:C2}
\end{minipage}
\caption{The only non-connectable odd pair-set in $Q_4$\label{fig:Fig2}}
\end{figure}

\begin{remark} Assumptions of Proposition~\ref{dimen-4}\,\eqref{dimen-4-2} are the best possible in the following sense: By a computer search we identified 53 non-connectable pair-sets $A\in\oddp_4$ such that $|A|=4$, $A$ contains two edge pairs and $\enc(A)=\emptyset$.
\end{remark}

In the rest of this text we often use the following consequence of Proposition~\ref{dimen-4}.

\begin{corollary}
\label{cor-dim-4}
If $A\in\oddp_4$ with $|A|=3$ and either $A$ contains an edge pair or for every $k\in[4]$ there exists a pair $\{\alpha,\beta\}\in A$ with $\alpha(k)\ne\beta(k)$, then $A$ is connectable.
\end{corollary}
The following theorem serves as an initial step for the inductive construction, described in the next section. 
Note that we were able to verify its statement by a computer search.
\begin{theorem}
\label{dimen-5}
$\dimp_5\subseteq\conp_5$.
\end{theorem}
\section{Induction step}
\label{section:induction-step}
First we present a general construction for a reduction of dimension.  Under some conditions imposed on a given pair-set $A\in\allp_n$, we construct two pair-sets in $\allp_{n-1}$ such that if these pair-sets are connectable, then $A$ is also connectable.  This construction exploits the concept of $i$-completion. The section is concluded when we prove that under some conditions, this construction may be applied to diminishable pairs.
\begin{construction}
\label{construction-simple}
Let $i\in[n]$ and $A\in\allp_n$ be a non-empty balanced pair-set.  Choose an $(A,i)$-matching $R$, set
$B=\emptyset$ and $A'=A$ and repeat one of the following steps whenever it is possible
until $A'=\emptyset$:
\begin{enumerate}[\upshape(i)]
\rm
\item\label{construction-simple:i}
if $\{\alpha ,\beta \}\in A'$ is an odd pair and $\alpha (i)=\beta (i)$, then add $\{\alpha ,\beta \}$ to $B$ and delete $\{\alpha ,\beta \}$ from $A'$;
\item \label{construction-simple:ii}
if $\{\alpha ,\beta \}\in A'$ is an odd pair and $\alpha (i)\ne\beta (i)$, then choose $\gamma\in V_n$ such that $\gamma (i)=\alpha (i)$, $\chi(\alpha)\ne\chi(\gamma)$ and $\gamma ,\gamma\oplus e_i\notin\bigcup A\cup\bigcup B$, add $\{\alpha ,\gamma \}$ and $\{\gamma\oplus e_i,\beta \}$ to $B$ and delete $\{\alpha ,\beta \}$ from $A'$;
\item\label{construction-simple:iii}
if $\{\alpha ,\beta \},\{\alpha',\beta'\}\in A'$ are even such that $\{\{\alpha ,\beta \},\{\alpha',\beta'\}\}\in R$ and $\alpha (i)=\beta (i)=\alpha'(i)=\beta'(i)$,
then add $\{\alpha ,\beta \}$ and $\{\alpha',\beta'\}$ to $B$ and delete $\{\alpha ,\beta \}$ and $\{\alpha',\beta'\}$ from $A'$;
\item\label{construction-simple:iv}
if $\{\alpha ,\beta \},\{\alpha',\beta'\}\in A'$ are even such that $\{\{\alpha ,\beta \},\{\alpha',\beta'\}\}\in R$ and $\beta'(i)=k\ne\alpha (i)=\beta (i)=\alpha'(i)$ for some $k\in[2]$ --- choose $\gamma\in V_n$ such that $\gamma (i)\ne k$, $\chi(\gamma)=\chi(\alpha')$ and $\gamma ,\gamma\oplus e_i\notin\bigcup A\cup\bigcup B$,
add $\{\alpha ,\beta \}$, $\{\alpha',\gamma \}$ and $\{\gamma\oplus e_i,\beta'\}$ to $B$ and delete $\{\alpha ,\beta \}$ and $\{\alpha',\beta'\}$ from $A'$;
\item\label{construction-simple:v}
if $\{\alpha ,\beta \},\{\alpha',\beta'\}\in A'$ are even such that $\{\{\alpha ,\beta \},\{\alpha',\beta'\}\}\in R$, $\alpha (i)=\beta (i)\ne\alpha'(i)=\beta'(i)$ and $\{\alpha',\beta'\}$ is not degenerated,
then choose distinct $\gamma ,\gamma'\in V_n$ such that $\gamma (i)=\gamma'(i)=\alpha'(i)$, $\chi(\gamma)=\chi(\gamma')\ne\chi(\alpha')$ and $\gamma ,\gamma',\gamma\oplus e_i,\gamma'\oplus e_i\notin\bigcup A\cup\bigcup B$,
add $\{\alpha ,\beta \}$, $\{\alpha',\gamma \}$, $\{\beta',\gamma'\}$, $\{\gamma\oplus e_i,\gamma'\oplus e_i\}$ to $B$ and delete $\{\alpha ,\beta \}$ and $\{\alpha',\beta'\}$ from $A'$;
\item\label{construction-simple:vi}
if $\{\alpha ,\beta \},\{\alpha',\beta'\}\in A'$ are even such that $\{\{\alpha ,\beta \},\{\alpha',\beta'\}\}\in R$, $\alpha (i)\ne\beta (i)$, $\alpha'(i)\ne\beta'(i)$, then --- assuming without loss of generality that $\alpha (i)=\alpha'(i)$ ---
choose distinct $\gamma,\gamma'\in V_n$ such that $\gamma (i)=\gamma'(i)=\alpha (i)$, $\chi(\gamma)=\chi(\alpha)$, $\chi(\gamma')=\chi(\alpha')$ and $\gamma ,\gamma',\gamma\oplus e_i,\gamma'\oplus e_i\notin\bigcup A\cup\bigcup B$, add $\{\alpha ,\gamma \}$, $\{\alpha',\gamma'\}$, $\{\gamma\oplus e_i,\beta \}$, $\{\gamma'\oplus e_i,\beta'\}$ to $B$ and delete $\{\alpha ,\beta \}$ and $\{\alpha',\beta'\}$ from $A'$.
\end{enumerate}
If the vertex $\gamma$ required for steps \eqref{construction-simple:ii} or \eqref{construction-simple:iv} does not exist, or the vertices $\gamma$ and $\gamma'$ required for steps \eqref{construction-simple:v} or \eqref{construction-simple:vi} do not exist, then Construction~\ref{construction-simple} ends unsuccessfully, i.e. with $A'\ne\emptyset$. If for each of the steps \eqref{construction-simple:ii},\eqref{construction-simple:iv},\eqref{construction-simple:v} or \eqref{construction-simple:vi} the required vertices exist, then Construction~\ref{construction-simple} ends successfully and creates a set $B$. The set of all sets $B$ that were successfully constructed from a set $A$ by Construction~\ref{construction-simple} is denoted by $\mathcal C_i(A)$.
 \end{construction}
Observe that every $B\in\mathcal C_i(A)$ is a pair-set. Indeed, if a step of Construction~\ref{construction-simple} treats a non-degenerate pair $\{\alpha,\beta\}$, then it adds a non-degenerate pair into $B$.  If the sets $B$ and $A'$ were constructed by execution of some sequence of successfull steps of Construction~\ref{construction-simple}, then $B\cup A'$ is a pair-set that is called a \emph{partial} $i$-\emph{completion} of a balanced pair-set $A$  for an
$(A,i)$-matching $R$.
\begin{proposition}
\label{observation-on-construction-simple}
Let $i\in[n]$ and $A\in\balp_n$ such that $A\ne\emptyset$. Then
\begin{enumerate}[\upshape(1)]
\item \label{observation-on-construction-simple:1}
Construction~\ref{construction-simple} is non-deterministic: by a selection of $(A,i)$-matching $R$ and by choices in steps \eqref{construction-simple:ii}, \eqref{construction-simple:iv}, \eqref{construction-simple:v}, and \eqref{construction-simple:vi}, more pair-sets may be obtained.
\item \label{observation-on-construction-simple:2}
every partial $i$-completion of $A$ has the same family of degenerated pairs as $A$;
\item \label{observation-on-construction-simple:3}
$\rho_{i=k}(A)\subseteq\rho_{i=k}(B)$ for every partial $i$-completion $B$ of $A$ and both $k\in[2]$;
\item\label{observation-on-construction-simple:4}
every partial $i$-completion $B$ of $A$ is balanced and $B\overset{*}{\implies}A$;
\item\label{observation-on-construction-simple:5}
every $B\in\mathcal C_i(A)$ is an $i$-completion of $A$;
\item\label{observation-on-construction-simple:6}
if $B'\in\mathcal C_i(B)$ for a partial $i$-completion $B$ of $A$, then $B'\in\mathcal C_i(A)$.
\end{enumerate}
\end{proposition}

\begin{proof}
Statements \eqref{observation-on-construction-simple:1} -- \eqref{observation-on-construction-simple:3} and \eqref{observation-on-construction-simple:6} follow from the definitions of  Construction~\ref{construction-simple} and partial $i$-completion. Referring to the notation of Construction~\ref{construction-simple}, if a step adds an even pair into $B$, then it actually adds two even pairs of distinct parities. Because $A$ is balanced, every partial $i$-completion of $A$ must be balanced as well. 
Moreover, if a step deletes some $\{\alpha,\beta\}$ from $A'$, then it adds into $B$ either $\{\alpha,\beta\}$, or $\{\alpha,\gamma\},\{\gamma\oplus e_i,\beta\}$, or $\{\alpha,\gamma \},\{\beta,\gamma'\},\{\gamma\oplus e_i,\gamma'\oplus e_i\}$ for suitable $\gamma, \gamma'$. Therefore if $B\cup A'\overset{*}{\implies}A$ before this step, then also $B\cup A'\overset{*}{\implies}A$ for the new values of $A'$ and $B$ after this step. 
It follows that after every successful sequence of steps of Construction~\ref{construction-simple} we have $B\cup A'\overset{*}{\implies}A$ 
and \eqref{observation-on-construction-simple:4} is proved.
Since every $\{\alpha,\beta\}\in B\in\mathcal C_i(A)$ satisfies $\alpha(i)=\beta(i)$, the statement \eqref{observation-on-construction-simple:5} is proved.
\end{proof}

A pair-set $B\in \mathcal C_i(A)$ is called a \emph{simple} $i$-\emph{completion} of $A$.
Our next important task is to find a condition under which every partial $i$-completion of $A$ can be extended to a simple $i$-completion of $A$.

\begin{lemma}
\label{lemma-on-construction-simple}
Let $i\in[n]$, $A\in\allp_n$ be a balanced pair-set, and $R$ be an
$(A,i)$-matching such that
$$2|A|-k_0-2k_1-k_2\le 2^{n-2}\quad\text{or}\quad A\text { is odd and }2|A|-k_0-1\le2^{n-2}$$
where
\begin{itemize}
\item
$k_0$ is the number of odd pairs $\{\alpha,\beta\}\in A$ with $\alpha(i)=\beta(i)$,
\item
$k_1$ is the number of $\{\{\alpha,\beta\},\{\alpha',\beta'\}\}\in R$ with $\alpha(i)=\beta(i)=\alpha'(i)=\beta'(i)$,
\item
$k_2$ is the number of $\{\{\alpha,\beta\},\{\alpha',\beta'\}\}\in R$ with $\alpha(i)=\beta(i)=\alpha'(i)\ne\beta'(i)$.
\end{itemize}
Then, for every partial $i$-completion $B$ of $A$ and either for every odd pair $\{\alpha,\beta\}\in A'$ or for every even pairs $\{\alpha,\beta\},\{\alpha',\beta'\}\in A'$ with $\{\{\alpha,\beta\}\{\alpha',\beta'\}\}\in R$
we can successfully execute the required step of Construction~\ref{construction-simple}.
Hence $\mathcal C_i(B)\ne \emptyset$ and, in particular, $\mathcal C_i(A)\ne\emptyset$. Moreover, for every $B\in\mathcal C_i(A)$ we have $|B|=2|A|-k_0-2k_1-k_2$ and $|\rho_{i=\ell}(B)|\le |A|$ for both $\ell\in[2]$.
\end{lemma}

\begin{remark}
Lemma~\ref{lemma-on-construction-simple} states that if $A\in\allp_n$, $i\in[n]$ and $(A,i)$-matching $R$ satisfy the assumptions of Lemma~\ref{lemma-on-construction-simple},
then Construction~\ref{construction-simple} always successfully ends,
no matter which $\gamma$ or $\gamma'$ were possibly chosen in steps \eqref{construction-simple:ii}, \eqref{construction-simple:iv}, \eqref{construction-simple:v}, or \eqref{construction-simple:vi} of Construction~\ref{construction-simple}.
\end{remark}

\begin{proof}[Proof~of~Lemma~\ref{lemma-on-construction-simple}]
First note that we need not consider steps \eqref{construction-simple:i} and \eqref{construction-simple:iii} of Construction~\ref{construction-simple}, as they do not require choosing any additional vertices and therefore can be always performed. We have to prove that in steps \eqref{construction-simple:ii} and \eqref{construction-simple:iv}  we can choose a vertex $\gamma$, and in steps \eqref{construction-simple:v} and \eqref{construction-simple:vi} we can choose vertices $\gamma$, $\gamma'$ satisfying the requirements.  The requirements are of three types:
\begin{description}
\item[1st] determines the $i$-th coordinate of $\gamma$ or $\gamma'$,
\item [2nd] determines $\chi(\gamma)$ or $\chi(\gamma')$,
\item [3rd] forbids some vertices.
\end{description}
Assume that one of the steps \eqref{construction-simple:ii}, \eqref{construction-simple:iv} --
\eqref{construction-simple:vi} of Construction~\ref{construction-simple} is to be executed.
Let $X\subseteq V_n$ be a~set satisfying the first two requirements, then $|X|=2^{n-2}$.
To deal with the third requirement, we analyze the reasons why a vertex is forbidden.
Observe that a vertex $\gamma$ is forbidden if $\gamma\in\bigcup A'\cup\bigcup B$ or
$\gamma\oplus e_i\in\bigcup A'\cup\bigcup B$.
Since the $i$-th coordinate of $\gamma$ or $\gamma'$ is fixed by the 1st requirement,
any pair of vertices $\{\alpha,\alpha\oplus e_i\}$ such that $\{\alpha,\alpha\oplus e_i\}\cap(\bigcup A'\cup\bigcup B)\ne\emptyset$
forbids at most one vertex of $X$. Moreover, $\{\alpha,\alpha\oplus e_i\}\cap(\bigcup A'\cup\bigcup B)\ne\emptyset$
if and only if either $\{\alpha,\alpha\oplus e_i\}\cap\bigcup A'\ne\emptyset$ or one
vertex from $\{\alpha,\alpha\oplus e_i\}$ was chosen in a preceding step of
Construction~\ref{construction-simple} and in this case
$\{\alpha,\alpha\oplus e_i\}\subseteq\bigcup B$.  Thus 
\begin{itemize}
\item either $\gamma\in\{\alpha,\beta\}\in A'$, 
\item or there exists an odd $\{\alpha,\beta\}\in A'\setminus A$ such that $\gamma\in\{\alpha,\beta\}$ or $\gamma$ was chosen by the execution of a step of Construction~\ref{construction-simple} on $\{\alpha,\beta\}$ (if $\{\alpha,\beta\}$ was deleted from $A'$),
\item or there exists a pair $\{\{\alpha,\beta\},\{\alpha',\beta'\}\}\in R$ of even pairs and $\gamma\in\{\alpha,\beta,\alpha',\beta'\}$ or $\gamma$ was chosen by the execution of a step of Construction~\ref{construction-simple} on $\{\{\alpha,\beta\},\{\alpha',\beta'\}\}$ (if $\{\alpha,\beta\},\{\alpha',\beta'\}$ were deleted from $A'$).
\end{itemize}
 In the second case we say that $\{\alpha,\beta\}$ forbids $\gamma$, in the third case we say that  $\{\{\alpha,\beta\},\{\alpha',\beta'\}\}\in R$ forbids $\gamma$. In the first case, if $\{\alpha,\beta\}$ is even, then there exists $\{\alpha',\beta'\}\in A'$ such that  $\{\{\alpha,\beta\},\{\alpha',\beta'\}\}\in R$ and in this case we say that $\{\{\alpha,\beta\},\{\alpha',\beta'\}\}\in R$ forbids $\gamma$; if $\{\alpha,\beta\}$ is odd then $\{\alpha,\beta\}$ forbids $\gamma$. To perform a more detailed analysis of this process,
we investigate odd pairs of $A$ and elements of $(A,i)$-matching $R$.

 Let $\{\alpha,\beta\}\in A$ be an odd pair. If $\alpha(i)=\beta(i)$, then this pair forbids only one vertex regardless whether $\{\alpha,\beta\}\in A'$ or $\{\alpha,\beta\}\not\in A'$, because $\chi(\alpha)\ne\chi(\beta)$ and $X$ consists -- by the 2nd requirement -- of vertices of the same parity. If $\alpha(i)\ne\beta(i)$ and $\{\alpha,\beta\}\in A'$, then vertices $\alpha$ and $\beta$ forbid either two vertices of $X$ or no vertex of $X$; the latter case occurs, in particular, if the step to be executed deals with $\{\alpha,\beta\}$. If $\{\alpha,\beta\}\notin A'$, then $\{\alpha,\beta\}$ was treated previously by step \eqref{construction-simple:ii}: recall that this step adds into $B$ two odd pairs with the same $i$-th coordinate, so these pairs forbid two vertices of $X$. In all cases, $\{\alpha,\beta\}$ forbids at most two vertices of $X$.

Let $\{\{\alpha,\beta\},\{\alpha',\beta'\}\}\in R$. If $\alpha(i)=\alpha'(i)=\beta(i)=\beta'(i)$, then these vertices forbid two vertices in $X$ regardless whether $\{\alpha,\beta\}\in A'$ or $\{\alpha,\beta\}\not\in A'$. Next assume that $\alpha(i)=\beta(i)=\alpha'(i)\ne\beta'(i)$.
If $\{\alpha,\beta\},\{\alpha',\beta'\}\notin A'$, then $\{\alpha,\beta\}$ and $\{\alpha',\beta'\}$ were treated previously by step \eqref{construction-simple:iv}: recall that this step adds into $B$
two even pairs with distinct parities and the same $i$-th coordinate and one odd pair with the same $i$-th coordinate, which forbids three vertices of $X$.
If $\{\alpha,\beta\},\{\alpha',\beta'\}\in A'$,
then these vertices forbid either three vertices or one vertex of $X$; the latter case occurs in particular if  $\{\alpha,\beta\}$ and $\{\alpha',\beta'\}$ are treated by the step to be executed. Assume that $\alpha(i)=\alpha'(i)\ne\beta(i)=\beta'(i)$.  If $\{\alpha,\beta\},\{\alpha',\beta'\}\notin A'$, then $\{\alpha,\beta\}$ and $\{\alpha',\beta'\}$ were treated previously by step \eqref{construction-simple:vi}: recall that this step adds  into $B$ two even pairs with distinct parities and the same $i$-th coordinate and two odd pairs with the same $i$-th coordinate, which forbids four vertices of $X$.
If $\{\alpha,\beta\},\{\alpha',\beta'\}\in A'$, then these vertices forbid two vertices of $X$; note that this happens in particular if  $\{\alpha,\beta\}$ and $\{\alpha',\beta'\}$ are treated by the step to be  the executed.
Finally assume that $\alpha(i)=\beta(i)\ne\alpha'(i)=\beta'(i)$.
If $\{\alpha,\beta\},\{\alpha',\beta'\}\notin A'$, then $\{\alpha,\beta\}$ and $\{\alpha',\beta'\}$ were treated previously by step \eqref{construction-simple:v}: recall that this step adds into $B$ two even pairs with distinct parities and the same $i$-th coordinate and two odd pairs with the same $i$-th coordinate, which forbids four vertices of $X$. If $\{\alpha,\beta\},\{\alpha',\beta'\}\in A'$, then these vertices forbid either four vertices of $X$ or no vertex of $X$; the latter case occurs in particular if $\{\alpha,\beta\}$ and $\{\alpha',\beta'\}$ are treated by the step to be executed.

To sum up these contributions, let
\begin{itemize}
\item $\ell_0$ be the number of odd pairs $\{\alpha,\beta\}\in A$ with $\alpha(i)=\beta(i)$;
\item $\ell_1$ be the number of odd pairs $\{\alpha,\beta\}\in A$ with $\alpha(i)\ne\beta(i)$;
\item $\ell_2$ be the number of elements $\{\{\alpha,\beta\},\{\alpha',\beta'\}\}\in R$ with $\alpha(i)=\beta(i)=\alpha'(i)=\beta'(i)$;
\item $\ell_3$ be the number of elements $\{\{\alpha,\beta\},\{\alpha',\beta'\}\}\in R$ such that exactly three vertices from the set $\{\alpha,\beta,\alpha',\beta'\}$ have the same $i$-th coordinate;
\item $\ell_4$ be the number of elements $\{\{\alpha,\beta\},\{\alpha',\beta'\}\}\in R$ with $\alpha(i)=\alpha'(i)\ne\beta(i)=\beta'(i)$;
\item $\ell_5$ be the number of elements $\{\{\alpha,\beta\},\{\alpha',\beta'\}\}\in R$ with $\alpha(i)=\beta(i)\ne\alpha'(i)=\beta'(i)$.
\end{itemize}
Then the number of forbidden vertices of $X$ is at most $\ell_0+2(\ell_1+\ell_2)+3\ell_3+4(\ell_4+\ell_5)-2$. Since $|A|=\ell_0+\ell_1+2(\ell_2+\ell_3+\ell_4+\ell_5)$, $k_0=\ell_0$, $k_1=\ell_2$, $k_2=\ell_3$ and $2|A|-k_0-2k_1-k_2\le 2^{n-2}=|X|$, we obtain that
$$
|X|-\left(\ell_0+2(\ell_1+\ell_2)+3\ell_3+4(\ell_4+\ell_5)-2\right)\ge2
$$
(or if $A$ is odd then $|X|-\left(\ell_0+2\ell_1-2\right)\ge1$)
and therefore $X$ always contains a vertex $\gamma$ in steps \eqref{construction-simple:ii} and \eqref{construction-simple:iv} or vertices $\gamma,\gamma'$ in steps \eqref{construction-simple:v} and \eqref{construction-simple:vi} that are not forbidden. Clearly, for every $B\in\mathcal C_i(A)$ we have $|B|=\ell_0+2(\ell_1+\ell_2)+3\ell_3+4(\ell_4+\ell_5)=2|A|-k_0-2k_1-k_2$. For $\l\in[2]$, each of steps
\begin{itemize}
\item \eqref{construction-simple:i} -- \eqref{construction-simple:ii} deletes one pair from $A'$ and adds  into $B$ at most one pair $\{\alpha,\beta\}$ with $\alpha(i)=\beta(i)=\l$,
\item \eqref{construction-simple:iii} -- \eqref{construction-simple:vi} deletes two pairs from $A'$ and adds into $B$ at most two pairs $\{\alpha,\beta\}$ with $\alpha(i)=\beta(i)=\l$,
\end{itemize}
which means that $|\rho_{i=\ell}(B)|\le|A|$ and the proof is complete.
\end{proof}

\begin{lemma}
\label{cor-on-construction-simple}
Let $A\in\balp_n$.
\begin{enumerate}[\upshape(1)]
\item
\label{cor-on-construction-simple:1}
If $2|A|\le2^{n-2}$, then $\mathcal C_i(A)\ne\emptyset$ for every $i\in[n]$;
\item
\label{cor-on-construction-simple:2}
if $|A|\le n-1$ and $n\ge 5$, then $\mathcal C_i(A)\ne\emptyset$ for every $i\in[n]$;
\item
\label{cor-on-construction-simple:3}
if $A$ is diminishable and $|A|=n\ge 5$, then there exists a separating $i\in[n]$ and
for every separating $i\in[n]$ we have $\mathcal C_i(A)\ne\emptyset$;
\item
\label{cor-on-construction-simple:4}
if $A$ is odd, $|A|=n-1$ and $n\ge4$, then there exists $i\in[4]$ such that $\mathcal{C}_i(A)\ne\emptyset$;
\item
\label{cor-on-construction-simple:5}
if $i\in[n]$ and $B$ is a partial $i$-completion of $A$, then
$$|A|-k_1^A-k_2^A\ge|B|-k_1^B-k_2^B;$$
where
\begin{itemize}
\item
$k_1^A$ (or $k_1^B$) is the number of odd pairs $\{\alpha,\beta\}\in A$ (or $\{\alpha,\beta\}\in B$)
with $\alpha(i)=\beta(i)$,
\item
$k_2^A$ (or $k_2^B$) is the number of even pairs $\{\alpha,\beta\}\in A$
(or $\{\alpha,\beta\}\in B$) with $\alpha(i)=\beta(i)$.
\end{itemize}
\item
\label{cor-on-construction-simple:6}
if $B$ is a simple $i$-completion of $A$ where $i\in[n]$ and $\{\alpha,\beta\}\in A$ is an odd pair
with $\alpha(i)=\beta(i)$, then $\{\alpha,\beta\}\in B$;
\item
\label{cor-on-construction-simple:7}
if $A$ is odd such that $2|A|-n_0-n_1<2^{n-2}$, $|A|-n_j\le n-1$ for both $j\in[2]$ where $\sigma_i(A)=(n_0,n_1)$
and $\enc(\rho_{i=\ell}\bigcup(A))=\emptyset$
for both $\ell\in[2]$, then there exists a simple $i$-completion $B$ of $A$ such that
$\enc(\rho_{i=\ell}(B))=\emptyset$ for both $\ell\in[2]$.
\end{enumerate}
\end{lemma}

\begin{proof}
\eqref{cor-on-construction-simple:1} follows from Lemma~\ref{lemma-on-construction-simple}.
Since $2(n-1)\le2^{n-2}$ for $n\ge 5$, part \eqref{cor-on-construction-simple:2}
follows from Lemma~\ref{lemma-on-construction-simple} as well. To verify \eqref{cor-on-construction-simple:3},
observe that the conditions imposed on $A$ in this case guarantee that there exists a separating $i\in[n]$. Moreover, for every separating $i\in[n]$ we have $n_0,n_1\ge 1$ where $\sigma_i(A)=(n_0,n_1)$ and therefore
$k_0=|\{\{\alpha,\beta\}\in A\mid\alpha(i)=\beta(i)\}|=n_0+n_1\ge2$.
Hence $2|A|-k_0\le2n-2\le2^{n-2}$ and $\mathcal C_i(A)\ne\emptyset$ by Lemma~\ref{lemma-on-construction-simple}, which completes the proof of \eqref{cor-on-construction-simple:3}.
If $A\in\oddp_n$ and $|A|=n-1$ for $n\ge4$ then, by Lemma~\ref{lemma-on-construction-simple},
$\mathcal{C}_i(A)=\emptyset$ implies that $n=4$ and $\sigma_i(A)=(0,0)$,  but since
$A$ is odd while $n=4$ is even, there exists $i\in[4]$ such that $\sigma_i(A)\ne(0,0)$ and
\eqref{cor-on-construction-simple:4} is proved.

To prove \eqref{cor-on-construction-simple:5}, consider pair-sets $B_1$  and $B_2$ such
that $B_2$ arises from a partial $i$-completion $B_1$ of $A$ by one step of Construction~\ref{construction-simple}.
Let $k_1$ (or $k_1'$) be the number of odd pairs $\{\alpha,\beta\}\in B_1$ (or $\{\alpha,\beta\}\in B_2$)
with $\alpha(i)=\beta(i)$ and let $k_2$ (or $k_2'$) be the number of even pairs $\{\alpha,\beta\}\in B_1$ (or $\{\alpha,\beta\}\in B_2$)
with $\alpha(i)=\beta(i)$. If $B_2$ arises by step \eqref{construction-simple:i} or  \eqref{construction-simple:iii} of Construction~\ref{construction-simple},
then $|B_1|=|B_2|$, $k_1=k'_1$ and $k_2=k_2'$, if $B_2$ arises by step \eqref{construction-simple:ii} then
$|B_1|+1=|B_2|$, $k_1+2=k'_1$ and $k_2=k'_2$, if $B_2$ arises by step \eqref{construction-simple:iv} then
$|B_1|+1=|B_2|$, $k_1+1=k'_1$ and $k_2+1=k'_2$, and  if $B_2$ arises by step \eqref{construction-simple:v} or
\eqref{construction-simple:vi} then
$|B_1|+2=|B_2|$, $k_1+2=k'_1$ and $k_2+2=k'_2$. Thus $|B_1|-k_1-k_2\ge|B_2|-k'_1-k'_2$ and hence by an easy induction we obtain
\eqref{cor-on-construction-simple:5}.

\eqref{cor-on-construction-simple:6} follows from the fact that in this case, step \eqref{construction-simple:i} of Construction~\ref{construction-simple} must be applied to $\{\alpha,\beta\}$.
To prove \eqref{cor-on-construction-simple:7}, consider a partial $i$-completion $B_1$ of $A$. Let $\sigma_i(B_1)=(m_0,m_1)$ and let
$B_2$ arise from $B_1$ by one step of Construction~\ref{construction-simple}. If $B_2$ is created from $B_1$ by
step \eqref{construction-simple:i} then $|B_2|=|B_1|$,
$\sigma_i(B_2)=(m_0,m_1)$ and $\enc(B_1)=\enc(B_2)$. If $B_2$ is created from $B_1$ by step \eqref{construction-simple:ii} then $|B_2|=|B_1|+1$ and $\sigma_i(B_2)=(m_0+1,m_1+1)$. Hence if
$\sigma_i(A)=(n_0,n_1)$ then, by induction, we obtain that
$m_0=n_0+|B_1|-|A|$, $m_1=n_1+|B_1|-|A|$ and thus $2|B_1|-m_0-m_1=2|A|-n_0-n_1$.
In the following we assume that $B_1$ is a partial $i$-completion of $A$ which is not an $i$-completion of $A$ and $\enc(\rho_{i=\ell}(\bigcup B_1))=\emptyset$ for both
$\ell\in[2]$. Hence $|B_1|<2|A|-n_0-n_1$ since $k_0=n_0+n_1$ and, by Lemma~\ref{lemma-on-construction-simple}, $2|A|-k_0\ge|B_2|>|B_1|$. Then in Construction~\ref{construction-simple} we have $A'\ne\emptyset$ and assume that the next step of Construction~\ref{construction-simple} is applied on $\{\alpha,\beta\}\in A'$. If $\alpha(i)=\beta(i)$, then necessarily step \eqref{construction-simple:i} is applied, which means that  $B_2=B_1$ and $\enc(\rho_{i=\ell}(\bigcup B_2))=\emptyset$
for both $\ell\in[2]$. Thus we can assume that $\alpha(i)=0\ne\beta(i)$. Our aim is to prove that
there exists $\gamma\in V_n\setminus\bigcup B_1$ such that $\gamma(i)=0$,
$\chi(\gamma)\ne\chi(\alpha)$, $\gamma\oplus e_i\notin\bigcup B_1$ and
$\enc(\rho_{i=\ell}(\bigcup B_1\cup\{\gamma,\gamma\oplus e_i\}))=\emptyset$ for both $\ell\in[2]$.
If step \eqref{construction-simple:ii} of Construction~\ref{construction-simple} chooses $\gamma$ with these properties,
then $B_2$ arising from $B_1$ satisfies $\enc(\rho_{i=\ell}(\bigcup B_2))=\emptyset$
for both $\ell\in[2]$. To prove the existence of such a $\gamma$, let $X=\{\eta\in V_n\mid\eta(i)=0,\,\chi(\eta)\ne\chi(\alpha)\}$, note that then $|X|=2^{n-2}$. In the proof of Lemma~\ref{lemma-on-construction-simple} we prove that the number of $\gamma\in X$ that do not forbid is at least $2|A|-n_0-n_1-2$ and, by the foregoing discussion,  $2|B_1|-m_0-m_1-2=2|A|-n_0-n_1-2$.
If $\rho_{i=0}(\eta)\in\enc(\rho_{i=0}(\bigcup B_1\cup\{\gamma\}))$ for some $\gamma\in X$, then $\eta(i)=0$ and $\{\gamma,\eta\}\in E(Q_n)$ and $$|\{\nu\in\bigcup B_1\mid\nu(i)=0,\,\{\nu,\eta\}\in E(Q_n)\}|=n-2.$$
Analogously if $\rho_{i=1}(\eta)\in\enc(\rho_{i=1}(\bigcup B_1\cup\{\gamma\oplus e_i\}))$ for some $\gamma\in X$ then $\eta(i)=1$ and $\{\gamma\oplus e_i,\eta\}\in E(Q_n)$ and 
$$|\{\nu\in\bigcup B_1\mid\nu(i)=1,\,\{\nu,\eta\}\in E(Q_n)\}|=n-2 $$ 
because $\enc(\rho_{i=\ell}(\bigcup B_1))=\emptyset$ for both $\ell\in[2]$.
By a straightforward computation $$|\{\nu\in\bigcup B_1\mid\nu(i)=0,\,\chi(\nu)\ne\chi(\alpha)\}|\le |A|-n_1-1\le n-2$$ and $$|\{\nu\in\bigcup B_1\mid\nu(i)=1,\,\chi(\nu)=\chi(\alpha)\}|\le |A|-n_0-1\le n-2$$ because $\chi(\alpha)\ne\chi(\gamma)=\chi(\beta)$ and $\beta(i)=1\ne\gamma(i)$. By Observation~\ref{obser-encom} there exists at most one $\gamma\in X$ such that $\enc(\rho_{i=0}(\bigcup B_1)\cup\{\gamma\}))\ne\emptyset$ and there exists at most one $\gamma'\in X$ such that $\enc(\rho_{i=1}(\bigcup B_1)\cup\{\gamma'\oplus e_i\}))\ne\emptyset$. Including these encompassed vertices among forbidden vertices and using the fact that $2|A|-n_0-n_1<2^{n-2}$, we conclude that there exists a vertex $\gamma\in X$ that is not forbidden. Thus if we select such a $\gamma$, the step \eqref{construction-simple:ii} of  Construction~\ref{construction-simple} creates a partial $i$-completion $B_2$ from $B_1$ such that $\enc(\rho_{i=\ell}(\bigcup B_2))=\emptyset$ for both $\ell\in[2]$.
An easy induction on steps of Construction~\ref{construction-simple}
completes the proof of \eqref{cor-on-construction-simple:7}.
\end{proof}
Next we give properties of pair-sets in $\mathcal C_i(A)$.
\begin{lemma}
\label{lemma-on-construction-simple-properties}
Let $A\in\balp_n$, $i\in[n]$ and let $R$ be an $(A,i)$-matching.
For $k\in[2]$ let
\begin{description}
\item $m_{o,k}$ be the number of odd $\{\alpha,\beta\}\in A$ with $\alpha(i)=\beta(i)=k$;
\item $m_{e,k}$ be the number of pairs $\{\{\alpha ,\beta \},\{\alpha',\beta'\}\}\in R$ such that $\alpha(i)=\beta(i)=\alpha'(i)=\beta'(i)=k$;
\item $m_{3,k}$ be the number of $\{\{\alpha ,\beta \},\{\alpha',\beta'\}\}\in R$ such that exactly three vertices from  $\{\alpha,\beta,\alpha',\beta'\}$ have its the $i$-th coordinate equal  to $k$.
\end{description}
then
\begin{enumerate}[\upshape(1)]
\setcounter{enumi}{\thetmpc}
\item
\label{lemma-on-construction-simple-properties:4}
for every $B\in\mathcal C_i(A)$ we have $|B|-\|B\|=|A|-\|A\|$;
\item
\label{lemma-on-construction-simple-properties:5}
for every $B\in \mathcal C_i(A)$ and both $k\in[2]$
we have $|\rho_{i=1-k}(B)|= |A|-(m_{o,k}+2m_{e,k}+m_{3,k})$;
\item
\label{lemma-on-construction-simple-properties:6}
 if $2|A|\le 2^{n-2}$ and $\sigma_i(A)=(n_0,n_1)$, then for
each $k\in[2]$ there exists $B\in \mathcal C_i(A)$ such that
\end{enumerate}
\vspace{-0.5cm}
\begin{align*}
\|A\|-m_{o,1-k}+m_{3,1-k}\le&\|\rho_{i=k}(B)\|\le\|A\|-m_{o,1-k}+2(|R|-m_{e,0}-m_{e,1}-m_{3,0})-m_{3,1},\\
2(m_{e,k}+m_{3,k})\le&|\rho_{i=k}(B)|-\|\rho_{i=k}(B)\|\le2(|R|-m_{e,1-k}-m_{3,1-k}),\\
&|\rho_{i=k}(B)|=2\|A\|-m_{o,0}-m_{o,1}+2|R|-2m_{e,1-k}-m_{3,1-k}.
\end{align*}
\end{lemma}

\begin{proof}
If a step of Construction~\ref{construction-simple} deletes some even pairs from $A'$, then the same number of even pairs is inserted into $B$.
Hence $|A|-\|A\|=|B|-\|B\|$ and \eqref{lemma-on-construction-simple-properties:4} is proved.

To prove \eqref{lemma-on-construction-simple-properties:5}, let $k\in[2]$ and observe the step of Construction~\ref{construction-simple} omitting an edge $\{\alpha,\beta\}$ from $A'$.
If this step omits odd pair from $A'$, then it adds one odd pair to $\rho_{i=1-k}(B)$ except $\alpha(i)=\beta(i)=k$, The number of these steps is $m_{o,k}$.
If this step omits two pairs $\{\alpha,\beta\},\{\alpha',\beta'\}$ from $A'$, then $\{\alpha,\beta\}$ and $\{\alpha',\beta'\}$ are even pairs. This step adds two pairs to $\rho_{i=1-k}(B)$ except that
\begin{itemize}
\item
either $\alpha(i)=\beta(i)=\alpha'(i)=\beta'(i)=k$ --- this step adds no pair to $\rho_{i=1-k}(B)$ and the number of such steps is $m_{e,k}$
\item
or $|\{\gamma\in\{\alpha,\beta,\alpha',\beta'\mid\gamma(i)=k\}|=3$ --- this step adds one pair to $\rho_{i=1-k}(B)$ and the number of such steps is $m_{3,k}$.
\end{itemize}
Therefore $|\rho_{i=1-k}(B)|=|A|-(m_{o,k}+2m_{e,k}+m_{3,k})$ and \eqref{lemma-on-construction-simple-properties:5} is proved.

\par
To prove \eqref{lemma-on-construction-simple-properties:6}, we have to analyse steps more carefully.  Step (i) (or (iii)) only shifts pairs from $A'$ to $B$, and it is repeated  $m_{o,0}+m_{o,1}$-times (or $m_{e,0}+m_{e,1}$-times).  Thus step (i) (or (iii)) adds $m_{o,k}$ (or $2m_{e,k}$) odd (or even) pairs $
\{\alpha ,\beta \}$  with $\alpha (i)=\beta (i)=k$ to $B$.
 Step (ii) is repeated $\|A\|-(m_{o,0}+m_{o,1})$-times and it adds two odd pairs $\{\alpha ,\beta \}$ and $\{\alpha',\beta'\}$  with $\alpha (i)=\beta (i)\ne\alpha'(i)=\beta'(i)$ to $B$.
Step (iv) is repeated $m_{3,0}+m_{3,1}$-times, and it adds to $B$
\begin{description}
\item $2m_{3,0}$ even pairs $\{\alpha,\beta \}$  with $\alpha (i)=\beta (i)=0$,
 \item $m_{3,1}$ odd pairs $\{\alpha ,\beta \}$ to  with $\alpha (i)=\beta (i)=0$,
\item $2m_{3,1}$ even pairs $\{\alpha ,\beta \}$ with $\alpha (i)=\beta (i)=1$,
\item $m_{3,0}$ odd pairs $\{\alpha ,\beta \}$ with $\alpha (i)=\beta (i)=1$.
\end{description}
Steps (v) and (vi) are repeated $|R|-(m_{e,0}+m_{e,1}+m_{3,0}+m_{3,1})$-times altogether. Note that they non-deterministically select $k\in[2]$, and then they add two even pairs to $\rho_k(B)$ as well as two odd pairs to $\rho_{1-k}(B)$.  Assuming that these steps select the same $k$, we conclude that they add
\begin{description}
\item $2(|R|-(m_{e,0}+m_{e,1}+m_{3,0}+m_{3,1})$ even pairs to $\rho_k(B)$, and
\item $2(|R|-(m_{e,0}+m_{e,1}+m_{3,0}+m_{3,1})$ odd pairs to $\rho_{1-k}(B)$.
\end{description}
If we sum all these numbers and use \eqref{lemma-on-construction-simple-properties:4},
then we obtain the equalities of \eqref{lemma-on-construction-simple-properties:6}.
To complete the proof of \eqref{lemma-on-construction-simple-properties:6}, note that the existence of  $B\in \mathcal C_i(A)$ satisfying these equalities follows from Lemma~\ref{lemma-on-construction-simple}.
\end{proof}

\begin{corollary}
\label{cor-induction-step-easy}
Let $n$ and $m$ be natural numbers with $m<n\ge 5$.  If every balanced  pair-set (or odd pair-set) $A\in\allp_n$ of size $|A|=m$ belongs to $\conp_n$, then $\balp^m_{n'}\subseteq\conp_{n'}$ (or $\oddp^m_{n'}\subseteq\conp_{n'})$ for every $n'\ge n$.
\end{corollary}
\begin{proof}Let $m'<m$ and $B\in\balp_n^{m'}$  (or $B\in\oddp_n^{m'}$). Select a non-degenerated pair $\{\alpha,\beta\}\in B$. Since $n\ge5$,
there exist $k=m-m'$ pairwise disjoint edge-pairs
$\{\alpha_i,\beta_i\}\subseteq V_n\setminus\bigcup B$ such that
$\chi(\alpha)\ne\chi(\alpha_i)$ for all $i\in[k]$. It follows that
\begin{itemize}
\item $\{\alpha,\alpha_0\}$ and $\{\beta_i,\alpha_{i+1}\}$  are odd pairs for all
$i\in[k-2]$;
\item $\chi(\beta_{k-1},\beta)=\chi(\alpha,\beta)$.
\end{itemize}
Replacing $\{\alpha,\beta\}$ in $B$ with
$\{\alpha,\alpha_0\},\{\beta_0,\alpha_1\},\dots,\{\beta_{k-2},\alpha_{k-1}\},\{\beta_{k-1},\beta\}$,
we obtain a balanced  (or odd) pair-set $A\in\allp_n$  such that
$A\overset {*}{\implies}B$ and $|A|=m$. By Lemma~\ref{lemma-simple}, if
$A\in\conp_n$, then $B\in\conp_n$ as well.
\par
Now let $n'>n$ and $m'\le m$ and $A\in\balp_{n'}^{m'}$  (or $A\in\oddp_{n'}^{m'}$) pair-set
with $|A|=m'$.  By Lemmas~\ref{cor-on-construction-simple} and~\ref{lemma-on-construction-simple-properties}, for every $i\in[n']$ there exists
an $i$-completion $B$  of $A$ such that $|B|-\|B\|=|A|-\|A\|$ and
$|\rho_{i=0}(B)|,|\rho_{i=1}(B)|\le |A|$.  Further if $A$ is odd, then
$\rho_{i=0}(B)$ and $\rho_{i=1}(B)$ are odd pair-sets.  If
$\rho_{i=0}(B),\rho_{i=1}(B)\in\conp_{n-1}$ then  by
Lemma~\ref{lemma-simple} we have $A\in\conp_n$. To complete the proof of the corollary, we argue by induction on $n'$, where the initial step is for $n'=n$.
\end{proof}
Another easy consequence that gives the first step of reduction is the following corollary.
\begin{corollary}
\label{cor-simple}
Let $A\in\oddp_n^{n-1}$ be an odd pair-set for $n\ge 5$ such that for some
$i\in[n]$ we have $\sigma_i(A)=(n_0,n_1)$ where $n_0\ne 0\ne n_1$. Then
for every simple $i$-completion $B$ of $A$ we have $\rho_{i=0}(B),\rho_{i=1}(B)\in\oddp^{n-2}_{n-1}$.
Thus if $\oddp^{n-2}_{n-1}\subseteq\conp_{n-1}$, then $A\in\conp_n$.

Let $A\in\dimp_n$ be a diminishable pair-set for $n\ge 5$ such that for some
$i\in[n]$ we have $\sigma_i(A)=(n_0,n_1)$ where $n_0,n_1>1$. Then 
for every simple $i$-completion $B$ of $A$ we have $\rho_{i=0}(B),\rho_{i=1}(B)\in\oddp^{n-2}_{n-1}$.
Thus if $\oddp^{n-2}_{n-1}\subseteq\conp_{n-1}$, then $A\in\conp_n$.
\end{corollary}
The rest of this section is devoted to diminishable pair-sets.
\begin{lemma}
\label{lemma-diminishable}
Let $n>4$, $A\in\Upsilon_n$, $\sigma_i(A)=(n_
0,n_1)$ for some 
$i\in n$, and if  $|A|=n$ then $i$ is separating for $A$. Then
\begin{enumerate}[\upshape(1)]
\item\label{lemma-diminishable:1}
if $|A|=n-1$ and $n_1=0$, or $|A|=n$ and $n_1=1$, then 
$\gamma\in\enc(\rho_{i=1}(\bigcup A)$ implies $|A|=n-1$ and $n_0=
0$ or $|A|=n$ and 
$n_0\le1$;
\item\label{lemma-diminishable:2}
there are at most two distinct $\gamma ,\gamma'\in V_n$ with 
$\rho_{i=\gamma (i)}(\gamma )\in\enc(\rho_{i=\gamma (i)}(\bigcup 
A))$ and $\rho_{i=\gamma'(i)}(\gamma')\in$ \linebreak$\in\enc(\rho_{i=\gamma'(i)}
(\bigcup A))$;
\item\label{lemma-diminishable:3}
if  $\gamma ,\gamma'\in V_n$ are distinct such that $\rho_{i=\gamma 
(i)}(\gamma )\in\enc(\rho_{i=\gamma (i)}(\bigcup A))$ 
and $\rho_{i=\gamma'(i)}(\gamma')\in\enc(\rho_{i=\gamma'(i)}(\bigcup A))$, then $|A|\ge n-1$ and $\chi(\gamma)\ne\chi(\gamma')$. Moreover,
\begin{itemize}
\item if $\gamma(i)=\gamma'(i)$ then either  $|A|=n-1$ and $n_{\gamma(i)}=n-1$, or $|A|=n$ and $n_{\gamma(i)}\ge n-1$; 
\item if
$\gamma(i)\ne\gamma'(i)$ then either $|A|=n-1$ and $n_0=n_1=0$, or $|A|=n$ and $n_0,n_1\le1$.
\end{itemize}
\end{enumerate}
\end{lemma}
\begin{proof}
Assume that $\gamma\in\enc(\rho_{i=1}(\bigcup A)$ and either $
|A|=n-1$, 
$n_1=0$ or $|A|=n$ and $n_1=1$. Then every neighbor $\alpha$ of $
\gamma$ with 
$\alpha (i)=1$ belongs to $\rho_{i=1}(\bigcup A)$. Therefore there exist $
n-1$ vertices 
$\alpha\in\bigcup A$ with $\alpha (i)=1$ and $\chi (\gamma )\ne\chi 
(\rho_{i=1}(\alpha ))$. If $|A|=n-1$ then there are
$|A|-n_0$ such vertices, if $|A|=n$ then there are  
$|A|-n_0-1$ such vertices. This completes the proof of \eqref{lemma-diminishable:1}.

To prove \eqref{lemma-diminishable:3}, first note that the that the existence of a $\gamma\in V_n$ with  
$\rho_{i=\gamma (i)}(\gamma )\in\enc(\rho_{i=\gamma (i)}(\bigcup 
A))$ immediately implies that $|A|\ge n-1$. 
Set 
$$X=\{\alpha\in V_n\mid(\alpha (i)=\gamma(i)\wedge\{\alpha ,\gamma 
\}\in E(Q_n)) \vee (\alpha(i)=\gamma'(i)\wedge\{\alpha ,\gamma'\}\in E(Q_n))\},$$
then $X\subseteq\bigcup A$. First assume that $\chi (\gamma )=\chi (\gamma')$. Then $|X|\ge 2n-4$, but as the vertices of X share the same parity, we also have $|X|\le|A|\le n$. This would imply that $n\le4$, contrary to our assumption. It follows that $\chi (\gamma )\ne\chi (\gamma')$.

If $\gamma (i)=\gamma'(i)$ then our assumptions on $\gamma$ and $\gamma'$ imply that $\{\gamma,\gamma'\}\not\in E(Q_n)$. 
Hence $|X|=2n-2$, $X$ is balanced and all vertices of $X$ share the same $i$-th coordinate. Thus either $|A|=n-1$ and $n_{\gamma(i)}=n-1$, or $|A|=n$ and  --- considering that $i$ is separating for $A$ --- $n_{\gamma(i)}\ge n-1$. 
If $\gamma (i)\ne\gamma'(i)$ then $|X|=2n-2$, $X$ consists of $n-1$ vertices with the $i$-th coordinate $0$ and of $n-1$ vertices with the $i$-th 
coordinate $1$. Moreover, the vertices of $X$ with the same $i$-th coordinate share the same parity that is opposite for vertices of distinct $i$-th coordinate. Therefore either $|A|=n-1$ and $n_0=n_1=0$, or $|A|=n$ and $n_0,n_1\le1$.
The proof of \eqref{lemma-diminishable:3} is complete.

To complete the proof of the lemma, note that \eqref{lemma-diminishable:2} follows from \eqref{lemma-diminishable:3}. Indeed, if there are three pairwise distinct vertices of $V_n$ satisfying the assumptions of \eqref{lemma-diminishable:2}, then two of them must share the same parity, contrary to \eqref{lemma-diminishable:3}. The proof is complete. 
\end{proof}
\begin{lemma}
\label{lemma-diminishableII}
Let $n\ge 6$, $i\in[n]$ and $A\in\Omega_n$ such that $n-1\le|A|\le n+3$,
$|A|-n+1\le n_0$, $n_1=|A|-n+1$ where $\sigma_i(A)=(n_0,n_1)$, and $\enc(\rho_{i=0}(\bigcup A))=\emptyset$.
Then
\begin{enumerate}[\upshape(1)]
\item\label{lemma-diminishableII:1}
if there exist two distinct pairs $\{\alpha ,\beta \},\{\alpha',\beta'
\}\in A$ with
$\alpha (i)=\alpha'(i)=0\ne\beta (i)=\beta'(i)$ and $\chi (\alpha
)=\chi (\alpha')$, then there exists
$\gamma\in V_n$   such that $\gamma(i)=0$, $\gamma ,\gamma\oplus e_i\notin\bigcup
A$, $\{\{\alpha ,\gamma \},\{\alpha',\gamma \}\}\cap E(Q_n)\ne\emptyset$,
and for both $\ell\in[2]$ we have
$$\enc(\rho_{i=\ell}(\bigcup A\cup\{\gamma,\gamma\oplus e_i\}))=\emptyset;$$
\item\label{lemma-diminishableII:2}
if there exist three pairwise distinct pairs
$\{\alpha ,\beta \},\{\alpha',\beta'\},\{\alpha^{\prime\prime},\beta^{
\prime\prime}\}\in A$ with
$\alpha (i)=\alpha'(i)=\alpha^{\prime\prime}(i)=0\ne\beta (i)=\beta'
(i)=\beta^{\prime\prime}(i)$ and
$\chi (\alpha )=\chi (\alpha')=\chi (\alpha^{\prime\prime})$, then there exist
distinct $\gamma,\gamma'\in V_n$ such
that $\gamma(i)=\gamma'(i)=0$, $\gamma,\gamma\oplus e_i,\gamma',\gamma'\oplus e_i
\notin\bigcup A$, $\{\gamma,\eta\},\{\gamma',\eta'\}\in E(Q_n)$ for some distinct
$\eta,\eta'\in\{\alpha ,\alpha',\alpha^{\prime\prime}\}$ and for both $\ell\in[2]$
we have $$\enc(\rho_{i=\ell}(\bigcup A\cup\{\gamma,\gamma\oplus e_i,\gamma',\gamma'\oplus e_i\}))=\emptyset.$$
\end{enumerate}
\end{lemma}
\begin{proof}
Note that our assumptions imply that $\enc(\rho_{i=\ell}(\bigcup A))=\emptyset$ for both $\ell\in[2]$ (the case $\ell=1$ follows from the fact that $n_1\le4<n-1$).
First assume that $A\in\oddp_n$, $A$ satisfies the assumptions of the lemma, and $\{\alpha,\beta\},\{\alpha',\beta'\}\in A$ are pairs such that $\chi(\alpha)=\chi(\alpha')$ and $\alpha(i)=\alpha'(i)=0\ne\beta(i)=\beta'(i)$.
Then
\begin{align}\label{lemma-diminishableII:a}\tag{$*$}
\begin{split}
|\{\eta\in\bigcup A\mid\chi(\alpha)\ne\chi(\eta),\,\eta(i)=0\}|&\le n_0+|A|-n_0-n_1-2=|A|-n_1-2=n-3\\
|\{\eta\in\bigcup A\mid\chi(\alpha)=\chi(\eta),\,\eta(i)=1\}|&\le n_1+|A|-n_0-n_1-2=|A|-n_0-2\le n-3
\end{split}
\end{align}
because every pair $\{\kappa,\kappa'\}\in A$ with $\kappa(i)=\kappa'(i)$ contributes by $1$ to the size of one of these sets (if $\kappa(i)=0$ then it contributes to the first set, otherwise it contributes to the second set) and a pair $\{\kappa,\kappa'\}\in A$ with $\kappa(i)=0\ne\kappa'(i)$ and $\chi(\kappa)\ne\chi(\alpha)$ contributes to both sets by at most one. Therefore for every $\gamma\in V_n$ such that $\gamma(i)=0$, $\chi(\gamma)\ne\chi(\alpha)$ and
$\gamma\notin\bigcup A$ we have that $\enc(\rho_{i=0}(\bigcup A\cup\{\gamma\}))=\emptyset$
because of \eqref{lemma-diminishableII:a}, $n-3+1<n-1$ and $\enc(\rho_{i=0}(\bigcup A))=\emptyset$. Analogously for every $\gamma\in V_n$ such that $\nu(i)=1$, $\chi(\gamma)=\chi(\alpha)$ and
$\gamma\notin\bigcup A$ we have that $\enc(\rho_{i=1}(\bigcup A\cup\{\gamma\}))=\emptyset$. Thus if there are two distinct pairs
$\{\alpha ,\beta \},\{\alpha',\beta'\}\in A'$ such that $0=\alpha
(i)=\alpha'(i)\ne\beta (i)=\beta'(i)=1$ and
$\chi (\alpha )=\chi (\alpha')$ and there exists $\gamma\in V_n$ with $\gamma(i)=0$,
$\gamma,\gamma\oplus e_i\notin\bigcup A$, and $\{\alpha,\gamma\}\in E(Q_n)$, then
$\enc(\rho_{i=0}(\bigcup A\cup\{\gamma,\gamma\oplus e_i\}))=\enc(\rho_{i=1}(\bigcup A\cup\{\gamma,\gamma\oplus e_i\}))=\emptyset$.

If there exist two distinct pairs $\{\alpha ,\beta \},\{\alpha',\beta'\}\in A$ such that $0=\alpha
(i)=\alpha'(i)\ne\beta (i)=\beta'(i)=1$ and
$\chi (\alpha )=\chi (\alpha')$, then there exist at least $2n-4$ vertices $
\gamma$ such
that $\gamma (i)=0$ and $\{\alpha ,\gamma \}\in E(Q_n)$ or $\{\alpha'
,\gamma \}\in E(Q_n)$.
We claim that the number of vertices $\gamma\in V_n$ such that $\gamma(i)=0$ and $\{\alpha,\gamma\}\in E(Q_n)$ or $\{\alpha',\gamma\}\in E(Q_n)$ and $\gamma\in\bigcup A$ or $\gamma\oplus e_i\in\bigcup A$ is at most $2|A|-n_0-n_1-4$. To verify the claim, set $X=\{\gamma\in V_n\mid \gamma(i)=0,\{\{\alpha ,\gamma \},\{\alpha',\gamma \}\}\cap E(Q_n)\ne\emptyset\}$ and note that the cardinalities of the sets
\begin{gather*}
\{\gamma\in X\mid\text{$\gamma$ belongs to some $\{\kappa,\kappa'\}\in A$ with $\kappa(i)=\kappa'(i)=0$}\},\\
\{\gamma\in X\mid\text{$\gamma\oplus e_i$ belongs to some $\{\kappa,\kappa'\}\in A$ with $\kappa(i)=\kappa'(i)=1$}\},\\
\{\gamma\in X\mid\text{$\gamma$ or $\gamma\oplus e_i$ belongs to some $\{\kappa,\kappa'\}\in A$ with $\kappa(i)\ne\kappa'(i)$}\},
\end{gather*}
do not excedd $n_0$, $n_1$ and $2(|A|-n_0-n_1-2)$, respectively, which sums up to $2|A|-n_0-n_1-4$. Since
$$
2n-4-(2|A|-n_0-n_1-4)=(n-|A|+n_0)+(n-|A|+n_1)\ge2,
$$
there are at least two vertices $\gamma$ satisfying the conclusion of part \eqref{lemma-diminishableII:1}.

 If there exist three pairwise distinct pairs $\{\alpha ,\beta \},\{\alpha',\beta'\},\{\alpha'',\beta''\}\in A$ such that $0=\alpha
(i)=\alpha'(i)=\alpha''(i)\ne\beta (i)=\beta'(i)=\beta''(i)=1$, then there exist at least $3n-8$ vertices $
\gamma\in V_n$ such
that $\gamma (i)=0$ and $\{\alpha ,\gamma \}$ or $\{\alpha'
,\gamma \}$ or $\{\alpha'',\gamma\}$ belongs to $E(Q_n)$. Analogously, there exist at most $2|A|-n_0-n_1-6$ vertices $
\gamma\in V_n$ such
that $\gamma (i)=0$, $\{\gamma,\gamma\oplus e_i\}\cap \bigcup A \neq\emptyset$ and $\{\alpha ,\gamma \}$ or $\{\alpha'
,\gamma \}$ or $\{\alpha'',\gamma\}$ belongs to $E(Q_n)$. Thus there exist at least
$$
3n-8-2|A|+n_0+n_1+6=3n-2|A|+n_0+n_1-2\ge n
$$
vertices $\gamma\in V_n$ such that 
\begin{align}\label{lemma-diminishableII:b}\tag{$**$}
\text{$\gamma (i)=0$, $\gamma,\gamma\oplus e_i\not\in\bigcup A$  and 
$\{\gamma ,\delta \}\in E(Q_n)$ for some $\delta\in \{\alpha ,\alpha',\alpha^{\prime\prime}\}$.}
\end{align}
Note that such vertices $\gamma$ share the same parity, different from $\chi(\alpha)$. Using \eqref{lemma-diminishableII:a} and the fact that two vertices of the same parity have at most $2<n-3$ neighbors in common, we infer that for each $\ell\in[2]$ there are at most two distinct vertices $\gamma_\ell$ and $\gamma_\ell'$ such that $\enc(\rho_{l=0}(\bigcup A\cup \{\gamma_\ell,\gamma_\ell'\}))\ne\emptyset$.
Since $n\ge 6$, there are distinct $\gamma ,\gamma'$ 
satisfying \eqref{lemma-diminishableII:b} such that $\enc(\rho_{i=\ell}(\bigcup A\cup\{\gamma,\gamma\oplus e_i,\gamma',\gamma'\oplus e_i\}))=\emptyset$ for both $\ell\in[2]$.
This completes the proof of \eqref{lemma-diminishableII:2}.
\end{proof}
We say that $i\in[n]$ is \emph{bad} for a diminishable pair-set $A\in\dimp_n$
if $|A|=n$, $i$ is separating for $A$ with
$\sigma_i(A)=(n_0,n_1)$ and there exist $k\in[2]$ such that $n_k=n-3$,
$n_{1-k}=1$ and distinct pairs $\{\alpha_0,\beta_0\},\{\alpha_1,\beta_1\},\{\alpha_2,\beta_2\}\in A$
with $\alpha_0(i)=\alpha_1(i)=k\ne\beta_0(i)=\beta_1(i)=\alpha_2(i)=\beta_2(i)$,
$\chi(\alpha_0)\ne\chi(\alpha_1)$, and for every $\kappa\in V_n$ with $\kappa(i)=k$ and $\{\kappa,\alpha_j\}\in E(Q_n)$ for some $j\in\{0,1\}$ we have
$\kappa\in \bigcup (A\setminus \{\{\alpha_l,\beta_l\}\mid l\in[3]\})$ or
$\kappa =\alpha_{1-j}$ or $\kappa\oplus e_i\in \{\alpha_2,\beta_0,\beta_1,\beta_2\}$.
Observe that for every $\{\alpha,\beta\}\in(A\setminus\{\{\alpha_j,\beta_j\}\mid j=0,1,2\})$
we have $\alpha(i)=\beta(i)=k$.
\begin{proposition}
\label{prop-diminishable}
Let $n\ge 6$ and $i\in[n]$.  Let $A\in\dimp_n$ such that $\sigma_i(A)=(n_0,n_1)$ and
\begin{itemize}
\item
either $|A|\le n-1$, $n_0,n_1\le |A|-3$ and if $n=6$ and $|A|=5$ then $\sigma_i(A)\ne(0,0)$;
\item
or $|A|=n$, $1\le n_0,n_1\le n-3$, $i$ is separating for $A$ and $i$ is not bad for $A$.
\end{itemize}
Then there exists a simple $i$-completion $B\in \mathcal C_i(A)$ with
$\rho_{i=0}(B),\rho_{i=1}(B)\in\dimp_{n-1}$. If $|A|=5$ and $n=6$ then for some $i\in[6]$
either $\max\{n_0,n_1\}\ge3$ or there exists a simple $i$-completion $B$ of $A$ with $\rho_{i=0}(B),\rho_{i=1}(B)\in\dimp_5$.
\end{proposition}
\begin{proof}
Assume that $A\in\Upsilon_n$ is a diminishable pair-set satisfying the assumptions of
the proposition. By Lemma~\ref{cor-on-construction-simple}\;\eqref{cor-on-construction-simple:1},
there exists a simple
$i$-completion $B$ of $A$. If $|A|<n-1$, then by Lemma~\ref{lemma-on-construction-simple},
$|\rho_{i=k}(B)|\le|A|<n-1$ for every simple $i$-completion $B$ of $A$ and both
$k\in[2]$. Since $B$ is odd, we conclude that $\rho_{i=k}(B)$ is diminishable for both
$k\in[2]$. Therefore we can assume that $|A|\ge n-1$.

Let $\sigma_i(A)=(n_0,n_1)$. If $|A|=n-1$ and $n_0,n_1\ne 0$ or $|A|=n$ and $n_0,n_1>1$ then, by
Corollary~\ref{cor-simple} and by Lemma~\ref{lemma-on-construction-simple-properties}\,\eqref{lemma-on-construction-simple-properties:5},
for every simple $i$-completion $B$ of $A$ we have that $\rho_{i=\ell}(B)$ are
diminishable for both $\ell\in[2]$. Therefore we can assume that
\begin{enumerate}[(i)]
\item \label{prop-diminishable:i} 
either $|A|=n-1$, $0\le n_0\le n-4$, $n_1=0$, and if $n=6$, $|A|=5$ then $n_0>0$;
\item \label{prop-diminishable:ii}
or  $|A|=n$, $1\le n_0\le n-3$, $n_1=1$ and there are edge-pairs $\{\alpha_0,\beta_0\},\{\alpha_1,\beta_1\}\in A$ with $\alpha_0(i)=\beta_0(i)=0$ and $\alpha_1(i)=\beta_1(i)=1$.
\end{enumerate}
Assume that we have constructed a partial $i$-completion $B$ of $A$ satisfying
\begin{enumerate}[(a)]
\item
\label{prop-diminishable:a}
there exist two edge pairs $\{\alpha ,\beta\}\in B$ with $\alpha (
i)=\beta (i)=0$;
\item
\label{prop-diminishable:b}
if $\gamma\in V_{n-1}$ is encompassed by $\bigcup\rho_{i=k}(B)$ for some $k\in[2]$, then $
\gamma\in\bigcup\rho_{i=k}(B)$;
\item
\label{prop-diminishable:c}
if either $|A|=n-1$ and $n_0=0$, or $|A|=n$ and $n_0=1$, then there exist two edge pairs $\{\alpha ,\beta\}\in B$ with $\alpha(i)=\beta (i)=1$;
\item
\label{prop-diminishable:d}
$|B|\le n+3$ and $\sigma_i(B)=(n_0+|B|-|A|,n_1+|B|-|A|)$.
\setcounter{saveenumi}{3}
\end{enumerate}
Since $2|B|-n_0-|B|+|A|-n_1-|B|+|A|=2|A|-n_0-n_1$, $A$ is diminishable and $n\ge 6$,
we conclude that $2|A|-n_0-n_1\le 2n-2<2^{n-2}$. For $j\in[2]$, $|B|-n_j-|B|+|A|=|A|-n_j\le n-1$.
Thus we can apply Lemma~\ref{cor-on-construction-simple}\;\eqref{cor-on-construction-simple:6}
and \eqref{cor-on-construction-simple:7} and we obtain the existence of simple $i$-completion $B'$
of $B$ such that both $\rho_{i=\ell}(B')$ are diminishable for
both $\ell\in[2]$ (by Lemma~\ref{lemma-on-construction-simple}, $|\rho_{i=\ell}(B')|\le n-1$ for both $\ell\in[2]$
and, by Lemma~\ref{cor-on-construction-simple}\;\eqref{cor-on-construction-simple:6}
and \eqref{cor-on-construction-simple:7}, $B'$ satisfies the conditions \eqref{prop-diminishable:a},
\eqref{prop-diminishable:b} and \eqref{prop-diminishable:c}). Hence the proof will be complete.

We will construct a partial $i$-completion $B$ satisfying \eqref{prop-diminishable:a},
\eqref{prop-diminishable:b}, \eqref{prop-diminishable:c}, and \eqref{prop-diminishable:d} in two steps.  In {\sc Step~1},
we construct a partial $i$-completion $A'$ of $A$ such that $\enc(\rho_{i=\ell}(\bigcup A'))=\emptyset$ for both $\ell\in[2]$
which settles condition \eqref{prop-diminishable:b}. In {\sc Step~2}, we add suitable
edge pairs to construct the required partial $i$-completion $B$, satisfying also
\eqref{prop-diminishable:a} and \eqref{prop-diminishable:c}. The verification that $B$ satisfies
the condition \eqref{prop-diminishable:d} will be straightforward.
\par
{\sc Step~1}. If $\enc(\rho_{i=\ell}(A))=\emptyset$ for both $\ell\in[2]$
then $A'=A$ satisfies \eqref{prop-diminishable:b}.  Thus we can assume that there is
$\gamma\in V_n$ such that
$\rho_{i=k}(\gamma )\in\enc(\rho_{i=k}(\bigcup A))$ for some $k\in[2]$. Then $k=\gamma(i)$.

Since $A\in\Upsilon_n$, we conclude that $\gamma\notin\enc(A)$ and
therefore $\gamma\oplus e_i\notin\bigcup A$ and, by Lemma~\ref{lemma-diminishable}\;\eqref{lemma-diminishable:2},
if $\gamma'\in\enc(\rho_{i=\gamma'(i)}(A))$ then either $\gamma=\gamma'$ or
$\gamma(i)\ne\gamma'(i)$ and $\chi(\gamma)\ne\chi(\gamma')$. Thus if
$\rho_{i=\gamma(i)}(\nu)\in\enc(\rho_{i=\gamma(i)}(\bigcup A\cup\{\gamma\}))$
then $\{\nu,\gamma\}\in E(Q_n)$ and $\nu(i)=\gamma(i)$. This is a contradiction with
$\rho_{i=\gamma(i)}(\gamma)\in\enc(\rho_{i=\gamma(i)}(A))$. Therefore
$\enc(\rho_{i=\gamma(i)}(\bigcup A\cup\{\gamma\}))=\emptyset$.

Since $n_0,n_1\le n-3$, there is a pair $\{\beta ,\beta'\}\in A$ with
$\beta (i)\ne\beta'(i)$, and without loss of generality we can assume that
$\beta (i)=\gamma (i)$.
Put
$$A_1=(A\setminus \{\{\beta ,\beta'\}\})\cup \{\{\beta ,\gamma \},\{\gamma\oplus e_i,\beta'\}\},$$
then $A_1$ is a partial $i$-completion of $A$, $|A_1|=|A|+1$, $\sigma_i(A_1)=(n_0+1,n_1+1)$,
and $$\enc(\rho_{i=\gamma(i)}(A_1))=\emptyset.$$

If $\enc(\rho_{i=1-\gamma(i)}(A_1)=\emptyset$, then $A'=A_1$ satisfies \eqref{prop-diminishable:b} and we are done in this step.
Otherwise there exists  $\gamma'\in V_n\setminus\bigcup A\cup\{\gamma\}$ with $\rho_{i=l}(\gamma')\in\enc(\rho_{i=l}(\bigcup A_1))$ for some $l\in[2]$. 
Then, by Lemma~\ref{lemma-diminishable}\;\eqref{lemma-diminishable:2},
we have $\gamma(i)\ne\gamma'(i)$ and
$\chi (\gamma )\ne\chi (\gamma')$, while either $|A|=n-1$ and $n_0=0$, or $|A|=n$ and
$n_0=1$. Applying the same construction as above, put
$$A_2=(A_1\setminus \{\{\beta_1 ,\beta_1'\}\})\cup \{\{\beta_1,\gamma'\},\{\gamma'\oplus e_i,\beta_1'\}\},$$
for some $\{\beta_1,\beta_1'\}\in A$ with $\beta_1(i)\ne\beta_1'(i)$. By
Lemma~\ref{lemma-diminishable}\;\eqref{lemma-diminishable:2}, $A'=A_2$ satisfies \eqref{prop-diminishable:b}.

{\sc Step~2}.
First observe that if $|A|=n-1$ and $n_0>0$ or $|A|=n$ and $n_0>1$ then,
by Lemma~\ref{lemma-on-construction-simple-properties}\;\eqref{lemma-on-construction-simple-properties:6}, $\rho_{i=1}(B)$ is diminishable for every simple $i$-completion $B$ of $A$ and therefore the
condition \eqref{prop-diminishable:c} is satisfied for them. 

Further observe that the application of {\sc Step 1} creates one edge pair instead of an encompassed vertex and deletes no edge pair of $A$. More precisely, if there exists a vertex encompassed by $\rho_{i=\ell}(\bigcup A)$, then the added edge pair $\{\alpha,\beta\}$ satisfies  $\alpha(i)=\beta(i)=\ell$ for $\ell\in[2]$. Note that if $n_0>1$, then such an encompassed $\gamma$ satisfies $\gamma(i)=0$. Moreover, by Lemma~\ref{lemma-diminishable}\,\eqref{lemma-diminishable:3}, if there are two such encompassed vertices $\gamma,\gamma'$, then $\gamma(i)\ne\gamma(i)$. Hence if
$|A'|=|A|+2$ and $|A|=n$ or $|A'|=|A|+1$, $|A|=n$ and $n_0>1$, then $B=A'$ satisfies
the conditions \eqref{prop-diminishable:a}--\eqref{prop-diminishable:d}. We need  to deal with the rest of the cases.

If $|A|=n-1$ and $|A'|>|A|$, then there exists $\gamma\in V_n$ and we can assume that  $\gamma(i)=0$ and $\rho_{i=0}(\gamma)\in\enc(\rho_{i=0}(\bigcup A))$. Therefore
$A'$ contains an edge pair $\{\alpha',\beta'\}$ with $\alpha'(i)=\beta'(i)=0$. 
Since $n_0\le n-4$ we conclude that $A'$ contains at least two pairs $\{\alpha,\beta\}$ satisfying the assumptions of Lemma~\ref{lemma-diminishableII}\;\eqref{lemma-diminishableII:1}
and if $|A'|=|A|+2$ then $n_0=0$ by Lemma 3.8 (3) and $|A'|$ contains at least three pairs $\{\alpha,\beta\}$
satisfying the assumptions of Lemma~\ref{lemma-diminishableII}\;\eqref{lemma-diminishableII:2}.
Apply Lemma~\ref{lemma-diminishableII}\;\eqref{lemma-diminishableII:1}
to obtain a pair $\{\alpha,\beta\}\in A'$ with $\alpha(i)=0\ne\beta(i)$ and $\nu\in V_n$
such that $\{\nu,\alpha\}\in E(Q_n)$, $\nu(i)=0$ and $\nu,\nu\oplus e_i\notin\bigcup A'$.
If $|A'|=|A|+1$, set
$$B=(A'\setminus\{\{\alpha,\beta\}\})\cup\{\{\alpha,\nu\},\{\nu\oplus e_i,\beta\}\}.$$
If $|A'|=|A|+2$ then apply Lemma~\ref{lemma-diminishableII}\;\eqref{lemma-diminishableII:1}
to obtain $\{\alpha',\beta'\}\in A'\setminus\{\{\alpha,\beta\}\}$ such that $\alpha'(i)=1\ne\beta'(i)$
and $\nu'\in V_n$ such that $\nu'\ne\nu\oplus e_i$, $\{\nu',\alpha'\}\in E(Q_n)$, $\nu'(i)=1$,
and $\nu',\nu'\oplus e_i\notin\bigcup A'$. Then set
$$B=(A'\setminus\{\{\alpha,\beta\},\{\alpha',\beta'\}\})\cup\{\{\alpha,\nu\},\{\nu\oplus e_i,\beta\},\{\alpha',\nu'\},\{\nu'\oplus e_i,\beta'\}\}.$$
In both cases $B$ satisfies the conditions \eqref{prop-diminishable:a}--\eqref{prop-diminishable:d}
(Lemma~\ref{lemma-diminishable}\;\eqref{lemma-diminishable:1} implies that $B$
satisfies \eqref{prop-diminishable:b}).

If $|A|=n$, $n_0=1$ and $|A'|=|A|+1$, apply
Lemma~\ref{lemma-diminishableII}\;\eqref{lemma-diminishableII:1} to obtain
$\{\alpha,\beta\}\in A'$ such that $\alpha(i)=1\ne\beta(i)$ and $\nu'\in V_n$ such
that $\{\nu',\alpha'\}\in E(Q_n)$, $\nu'(i)=1$ and $\nu',\nu'\oplus e_i\notin\bigcup A'$.
Then set
$$B=(A'\setminus\{\{\alpha,\beta\}\})\cup\{\{\alpha,\nu\},\{\nu\oplus e_i,\beta\}\}$$
to obtain $B$ satisfying the conditions \eqref{prop-diminishable:a}--\eqref{prop-diminishable:d}
(Lemma~\ref{lemma-diminishable}\;\eqref{lemma-diminishable:1} implies that $B$
satisfies \eqref{prop-diminishable:b}).

It remains to settle the case that $A=A'$. First assume that $|A|=n$. If $n_0=1$ then there exist
four pairwise distinct $\{\alpha_j,\beta_j\}\in A$ for $j=0,1,2,3$ with
$\alpha_j(i)\ne\beta_j(i)$.  Thus we can assume $\alpha_0(i)=\alpha_
1(i)=0\ne\beta_2(i)=\beta_3(i)$
and $\chi (\alpha_0)=\chi (\alpha_1)$.  By Lemma~\ref{lemma-diminishableII}\,\eqref{lemma-diminishableII:1}, we can assume there
exists $\gamma\in V_n$ such that $\{\alpha_0,\gamma \}\in E(Q_n)$ and $
\gamma ,\gamma\oplus e_i\notin\bigcup A$ and $\enc(\rho_{i=\ell}(\bigcup A\cup\{\gamma,\gamma\oplus e_i\}))=\emptyset$ for both $\ell\in[2]$.  Set
$$A_3=(A\setminus \{\{\alpha_0,\beta_0\}\})\cup \{\{\alpha_0,\gamma
\},\{\gamma\oplus e_i,\beta_0\}\}.$$
Since $\beta_1(i)=\beta_2(i)=\beta_3(i)$ we can assume that $\chi
(\beta_1)=\chi (\beta_2)$.
By Lemma~\ref{lemma-diminishableII}\,\eqref{lemma-diminishableII:1}, there exists $\gamma'\in V_n$ such
that $\{\beta_2,\gamma'\}\in E(Q_n)$ and $\gamma',\gamma'\oplus e_
i\notin\bigcup A_3$ and $\enc(\rho_{i=\ell}(\bigcup A_3\cup\{\gamma',\gamma'\oplus e_i\}))=\emptyset$ for both $\ell\in[2]$.  Set
$$B=(A_3\setminus \{\{\alpha_2,\beta_2\}\})\cup \{\{\beta_2,\gamma'
\},\{\alpha_2,\gamma'\oplus e_i\}\}$$
then $B$ satisfies \eqref{prop-diminishable:a}--\eqref{prop-diminishable:d}.
If $|A|=n$, $n_0>1$ and if there exist two distinct $\{\alpha,\beta\},\{\alpha',\beta'\}\in A$
with $\alpha(i)=\alpha'(i)=0\ne\beta(i)=\beta'(i)$ and $\chi(\alpha)=\chi(\alpha')$,
then, by Lemma~\ref{lemma-diminishableII}\,\eqref{lemma-diminishableII:1},
we can assume that there exists $\gamma\in V_n$ such that $\gamma,\gamma\oplus e_i\notin\bigcup A$,
$\{\alpha,\gamma\}\in E(Q_n)$ and $\enc(\rho_{i=\ell}(A\cup\{\gamma,\gamma\oplus e_i\}))=\emptyset$ for both $\ell\in[2]$.
Set $$B=(A\setminus\{\{\alpha,\beta\}\})\cup\{\{\alpha,\gamma\},\{\gamma\oplus e_i,\beta\}\}$$
then $B$ satisfies \eqref{prop-diminishable:a}--\eqref{prop-diminishable:d}.
If $n_0<n-3$ then there exist at least
three pairwise distinct pairs $\{\alpha,\beta\}\in A$ with $\alpha(i)\ne\beta(i)$. Hence
there exist two distinct $\{\alpha,\beta\},\{\alpha',\beta'\}\in A$
with $\alpha(i)=\alpha'(i)=0\ne\beta(i)=\beta'(i)$ and $\chi(\alpha)=\chi(\alpha')$ and the previous case applies.
Thus we can assume that $n_0=n-3$ and for $\{\alpha,\beta\},\{\alpha',\beta'\}\in A$
with $\alpha(i)=\alpha'(i)=0\ne\beta(i)=\beta'(i)$ we have $\chi(\alpha)\ne\chi(\alpha')$.
Then $|\rho_{i=1}(\bigcup A)|=4$ and for both $\ell\in\{-1,1\}$
we have $|\{\eta\in\bigcup \rho_{i=0}(A)\mid\chi(\eta)=\ell\}|=n-3$ while
$|(\eta\in\bigcup \rho_{i=1}(A)\mid\chi(\eta)=\ell\})|=2$. Since $i$ is not bad,
there exist $\{\alpha,\beta\}\in A$ with $\alpha(i)=0\ne\beta(i)$ and
$\gamma\in V_n$ such that $\{\gamma,\alpha\}\in E(Q_n)$, $\gamma(i)=0$ and
$\gamma,\gamma\oplus e_i\notin\bigcup A$. Then
$\enc(\bigcup \rho_{i=0}(A)\cup\rho_{i=0}(\gamma))=\emptyset$ and
$$B=(A\setminus\{\{\alpha,\beta\}\})\cup\{\{\alpha,\gamma\},\{\gamma\oplus e_i,\beta\}\}$$
satisfies \eqref{prop-diminishable:a}--\eqref{prop-diminishable:d}.

Now assume that $A=A'$ and $|A|=n-1$.  If
$n-4>n_0>0$ then there exist four pairwise distinct
$\{\alpha_j,\beta_j\}\in |A|$ for $j=0,1,2,3$ with $\alpha_j(i)=0\ne
\beta_j(i)$. Then we can assume that either $\alpha_j$ for $j\in[3]$
have the same parity or $\chi (\alpha_0)=\chi (\alpha_1)\ne\chi (\alpha_2)=
\chi (\alpha_3)$.  
In both cases, by
Lemma~\ref{lemma-diminishableII}\,\eqref{lemma-diminishableII:1} and \eqref{lemma-diminishableII:2},
we can assume that there exist distinct
$\gamma ,\gamma'\in V_n$ such that $\gamma(i)=\gamma'(i)=0$, $\gamma ,\gamma\oplus e_i,\gamma'
,\gamma'\oplus e_i\notin\bigcup A$ and
$\{\gamma ,\alpha_0\},\{\gamma',\alpha_2\}\in E(Q_n)$ and
$\enc(\rho_{i=0}(A\cup\{\gamma,\gamma'\}))=\emptyset$. 
Then
$$B=(A\setminus \{\{\alpha_0,\beta_0\},\{\alpha_2,\beta_2\}\})\cup
\{\{\alpha_0,\gamma \},\{\gamma\oplus e_i,\beta_0\},\{\alpha_2,\gamma'
\},\{\gamma'\oplus e_i,\beta_2\}\}$$
satisfies \eqref{prop-diminishable:a}--\eqref{prop-diminishable:d}.

If $|A|=n-1$ and $n_0=n-4$, then
there exist $\{\alpha ,\beta \},\{\alpha',\beta'\},\{\alpha^{\prime
\prime},\beta^{\prime\prime}\}\in A$ with
$\alpha (i)=\alpha'(i)=\alpha^{\prime\prime}(i)=0\ne\beta (i)=\beta'
(i)=\beta^{\prime\prime}(i)$.  If
$\chi (\alpha )=\chi (\alpha')=\chi (\alpha^{\prime\prime})$ then, by
Lemma~\ref{lemma-diminishableII}\,\eqref{lemma-diminishableII:2}, there are distinct $\gamma ,\gamma'\in V_n$ such that $\{\alpha
,\gamma \},\{\alpha',\gamma'\}\in E(Q_n)$,
$\gamma ,\gamma',\gamma\oplus e_i,\gamma'\oplus e_i\notin\bigcup
A$ and $\enc(\rho_{i=0}(A\cup\{\gamma,\gamma'\}))=\emptyset$.  Then
$$B=(A\setminus \{\{\alpha ,\beta \},\{\alpha',\beta'\}\})\cup \{
\{\alpha ,\gamma \},\{\gamma\oplus e_i\},\{\gamma',\alpha'\},\{\gamma'
\oplus e_i,\beta'\}\}$$
satisfies \eqref{prop-diminishable:a}--\eqref{prop-diminishable:d}.
If $\alpha$, $\alpha'$ and $\alpha^{\prime\prime}$ do not share the same
parity,  we can assume that $\chi (\alpha )=\chi (\alpha')\ne\chi(\alpha^{\prime\prime})$. 
Then there exist $n-4$ pairs $\{\kappa,\lambda\}\in A$ with $\kappa(i)=\lambda(i)=\alpha(i)$ and therefore there exist $\{\gamma,\gamma'\}$ such that $\gamma(i)=\gamma'(i)=\alpha(i)$, 
$\{\gamma,\alpha\},\{\gamma',\alpha'\}\in E(Q_n)$ and $\gamma,\gamma',\gamma\oplus e_i,\gamma'\oplus e_i\notin\bigcup A$.
Say, if there exist $n-3$ vertices $\kappa$ such that $\kappa(i)=\alpha(i)$, $\{\kappa,\alpha\}\in E(Q_n)$, $\{\kappa,\kappa\oplus e_i\}\cap\bigcup A\ne\emptyset$ then there exist $n-3$ vertices $\gamma'$ such that $\gamma'(i)=\alpha(i)$, $\{\gamma',\alpha'\}\in E(Q_n)$, $\gamma',\gamma'\oplus e_i\notin\bigcup A$.
 
Since $A$ is odd, it follows that there is exactly one $\nu\in\bigcup A$ with $\nu (i)=1$ and $\chi (\nu )=\chi (\alpha)$.
Furthermore, there are at least $2n-4$ vertices $\eta\in V_n$ such that $\eta(i)=0$ and
$\{\eta,\alpha\}\in E(Q_n)$ or $\{\eta,\alpha'\}\in E(Q_n)$, and  since
$|\{\nu\in\bigcup A\mid\chi (\nu )\ne\chi (\alpha ),\,\nu (i)=0\}|=n-3$,  there are
at most $n-3$ vertices $\eta\in\bigcup A$ with $\eta(i)=0$ and $\{\eta,\alpha\}\in E(Q_n)$
or $\{\eta,\alpha'\}\in E(Q_n)$. Consequently, we conclude that
there are at least $n-2$ distinct vertices $\eta\in V_n$ such that $\eta(i)=0$,
$\eta,\eta\oplus e_i\notin\bigcup A$ and $\{\eta,\alpha\}\in E(Q_n)$ or $\{\eta,\alpha'\}\in E(Q_n)$. Moreover,
if for some $\kappa\in\{\alpha,\alpha'\}$ we have $|\{\xi\in V_n\mid\{\xi,\kappa\}\in E(Q_{n}), \xi\in\bigcup A \text{ or } \xi\oplus e_i\in\bigcup A\}|=n-2$, then there exists $\gamma\in V_n$ such that $\gamma(i)=0$, $\{\gamma,\kappa\}\in E(Q_n)$, $\gamma,\gamma\oplus e_i\not\in\bigcup A$ and there exist at least $n-3$ vertices $\gamma'\in V_n$ such that 
$\gamma'(i)=0$, $\gamma',\gamma'\oplus e_i\not\in\bigcup A$ and $\{\gamma',\kappa'\}\in E(Q_n)$ where $\{\kappa,\kappa'\}=\{\alpha,\alpha'\}$. Hence there exist distinct $\gamma,\gamma'\in V_n$ 
such that $\gamma(i)=\gamma(i)=0$, $\{\alpha,\gamma\},\{\alpha',\gamma'\}\in E(Q_n)$,
$\gamma,\gamma\oplus e_i,\gamma',\gamma'\oplus e_i\notin\bigcup A$ and
$\enc(\bigcup \rho_{i=0}(A)\cup\rho_{i=0}(\{\gamma,\gamma'\}))=\emptyset$. Thus
$$B=(A\setminus \{\{\alpha ,\beta \},\{\alpha',\beta'\}\})\cup \{
\{\alpha ,\gamma \},\{\gamma\oplus e_i,\beta\},\{\gamma',\alpha'\},\{\gamma'
\oplus e_i,\beta'\}\}$$
satisfies \eqref{prop-diminishable:a}--\eqref{prop-diminishable:d}. 

If $|A|=n-1$ and $n_0=0$, then $n>6$ and for every $\{\alpha,\beta\}\in A$ we have $\alpha(i)\ne\beta(i)$.
In this case either 
\begin{itemize}
\item there are five pairs $\{\alpha_j,\beta_j\}\in A$ for $j\in[5]$ such that $\alpha_j(i)=0\ne\beta_j(i)$ for all $j\in[5]$ and $\chi(\alpha_j)=\chi(\alpha_k)$ for all $j,k\in[5]$, 
\end{itemize}
or there are six pairs $\{\alpha_j,\beta_j\}\in A$ for $j\in[6]$ such that $\alpha_j(i)=0\ne\beta_j(i)$
for all $j\in[6]$ and $\chi(\alpha_j)\ne\chi(\alpha_k)$ for all $j,k$
such that 
\begin{itemize}
\item either $j\in[3]$ and $k\in[6]\setminus[3]$, 
\item or $j\in[4]$ and $k\in[6]\setminus[4]$.
\end{itemize} 
By Lemma~\ref{lemma-diminishableII}\;\eqref{lemma-diminishableII:2} applied to 
$\{\alpha_j,\beta_j\}$ for $j\in[3]$, we can assume that there exist distinct
$\gamma,\gamma'\in V_n$ such that $\{\alpha_0,\gamma\},\{\alpha_1,\gamma'\}\in E(Q_n)$,
$\gamma(i)=\gamma'(i)=0$, $\gamma,\gamma\oplus e_i,\gamma',\gamma'\oplus e_i\notin\bigcup A$,
and  $\enc(\rho_{i=\ell}(\bigcup A\cup\{\gamma,\gamma',\gamma\oplus e_i,\gamma'\oplus e_i\}))=\emptyset$ for both $\ell\in[2]$. 
Set
$$B_1=(A\setminus\{\{\alpha_0,\beta_0\},\{\alpha_1,\beta_1\}\})\cup\{\{\alpha_0,\gamma\},\{\gamma\oplus e_i,\beta_0\},\{\alpha_1,\gamma'\},\{\gamma'\oplus e_i,\beta_1\}\}.$$
Then in the first case $\{\alpha_j,\beta_j\}$ for $j=2,3,4$  and in the second case  
$\{\alpha_j,\beta_j\}$ for $j=3,4,5$ satisfy the assumptions of
Lemma~\ref{lemma-diminishableII}\;\eqref{lemma-diminishableII:2}, while in the third
case $\{\alpha_j,\beta_j\}$ for $j=2,3$ as well as for $j=4,5$ satisfy the assumptions of
Lemma~\ref{lemma-diminishableII}\;\eqref{lemma-diminishableII:1}. 
Thus we can assume that
there exist distinct $\nu,\nu'\in V_n$ such that $\nu(i)=\nu'(i)=1$, $\{\beta_3,\nu\},\{\beta_4,\nu'\}\in E(Q_n)$,
$\nu,\nu\oplus e_i,\nu',\nu'\oplus e_i\notin\bigcup A$, and
$\enc(\rho_{i=\ell}(\bigcup B_1\cup\{\nu,\nu',\nu\oplus e_i,\nu'\oplus e_i\}))=\emptyset$ for both $\ell\in[2]$. Set
$$B=(B_1\setminus\{\{\alpha_3,\beta_3\},\{\alpha_4,\beta_4\}\})\cup\{\{\beta_3,\nu\},\{\nu\oplus e_i,\alpha_3\},\{\beta_4,\nu'\},\{\nu'\oplus e_i,\alpha_4\}\}.$$
Then $B$ satisfies the conditions \eqref{prop-diminishable:a}--\eqref{prop-diminishable:d}.

It remains to settle the case that $n=6$ and $|A|=5$.
Note that for every $\{\alpha,\beta\}\in A$ there exists $i\in[n]$ such that $\alpha(i)=\beta(i)$ (otherwise we have $\alpha(i)\ne\beta(i)$ for every $i\in[6]$ which means that $\{\alpha,\beta\}$ is even pair,  contrary to our assumption that $A$ is odd) and we can assume that $n_0>0$ and $n_1=0$ for $\sigma_i(A)=(n_0,n_1)$. Thus if $n_0\le3$ then $\sigma_i(A)\ne (0,0)$ and therefore --- by the previous part of the proof --- there exists a simple  $i$-completion $B$ of $A$ such that $\rho_{i=l}(B)\in\dimp_5$ for both $l\in[2]$. The proof is complete.
\end{proof}

As a consequence, we obtain the following corollary that forms a cornerstone of the induction step.
\begin{corollary}
\label{cor:main-section3}
Let $A\in\Upsilon_n$ for $n\ge 6$. Then there exists $i\in[n]$ such that one of the following conditions holds
\begin{enumerate}[\upshape(1)]
\item\label{cor:main-section3:1}
$|A|=n-1$, $\min\{n_0,n_1\}=0$, $\max\{n_0,n_1\}\ge|A|-2$ where $\sigma_i(A)=(n_0,n_1)$;
\item\label{cor:main-section3:2}
$|A|=n$, $i$ is separating for $A$, $\min\{n_0,n_1\}=1$ and $\max\{n_0,n_1\}\ge n-3$ where $\sigma_i (A)=(n_0,n_1)$;
\item\label{cor:main-section3:3}
$|A|=n$ and $i$ is bad for $A$;
\item\label{cor:main-section3:4}
there exists a simple $i$-completion $B$ of $A$ such that $\rho_{i=0}(B),\rho_{i=1}(B)\in\dimp_{n-1}$.
\end{enumerate}
\end{corollary}
\section{Special cases}
\label{section:special-cases}
The aim of this section is to deal with the cases that were not resolved by Proposition~\ref{prop-diminishable}, i.e. cases \eqref{cor:main-section3:1}-\eqref{cor:main-section3:3} of Corollary~\ref{cor:main-section3}. First we investigate the case that $A\in\oddp_n^{n-1}$ and there exists $i\in [n]$ such that $\sigma_i(A)=(n_0,n_1)$ where $n_0\ge |A|-2$ and $n_1=0$.

\begin{lemma}
\label{lem-4-1}
Let $n\ge 5$ and let $A\in\oddp_n$ such that $|A|\le n-1$ and
$\sigma_i(A)=(|A|,0)$ for some $i\in [n]$. If $\oddp_{n-1}^{n-2}\subseteq\Gamma_{n-1}$
then $A\in\Gamma_n$.
\end{lemma}

\begin{proof}
Choose $\{\hat{\alpha},\hat{\beta}\}\in A$ and set  $\alpha=\rho_{i=0}(\hat{\alpha })$, $\beta =\rho_{i=0}(\hat{\beta })$,  and $A'=\rho_{i=0}(A\setminus \{\{\hat{\alpha},\hat{\beta}\}\})$.  By the assumption, $A'\in\Gamma_{n-1}$ and hence there exists a connector $\{P_{\gamma ,\gamma'}\mid \{\gamma,\gamma'\}\in A'\}$  of $A'$.  One of the following possibilities occurs:
\begin{enumerate}[\upshape(A)]
\item
\label{lem-4-1:A}
there exist distinct $\{\kappa ,\kappa'\},\{\xi ,\xi'\}\in A'$ such that $
\alpha$ belongs to $P_{\kappa ,\kappa'}$
and $\beta$ belongs to $P_{\xi ,\xi'}$;
\item
\label{lem-4-1:B}
there exists $\{\kappa ,\kappa'\}\in A'$ such that $\alpha$ and $
\beta$ belong to the path
$P_{\kappa ,\kappa'}$ and $\{\alpha ,\beta \}$ is not an edge of $P_{
\kappa ,\kappa'}$;
\item
\label{lem-4-1:C}
there exists $\{\kappa ,\kappa'\}\in A'$ such that $\{\alpha ,\beta
\}$ is an edge of the path
$P_{\kappa ,\kappa'}$.
\end{enumerate}
If \eqref{lem-4-1:A} occurs then there exist subpaths $(\zeta ,\alpha ,\zeta')$ and $(\nu ,\beta ,\nu')$ of $P_{\kappa ,\kappa'}$  and $P_{\xi ,\xi'}$, respectively such
that $\zeta$ is closer to $\kappa$ than $\zeta'$ in $P_{\kappa ,\kappa}$ and $\nu$ is closer to $\xi$ than $\nu'$ in $
P_{\xi ,\xi'}$. Set
$$A^{\prime\prime}=(A'\setminus \{\{\kappa ,\kappa'\},\{\xi ,\xi'
\}\})\cup \{\{\kappa ,\zeta \},\{\alpha ,\alpha \},\{\zeta',\kappa'
\},\{\xi ,\nu \},\{\beta ,\beta \},\{\nu',\xi'\}\}.$$
If \eqref{lem-4-1:B} occurs then there exist vertex-disjoint subpaths $(\zeta ,
\alpha ,\zeta')$
and $(\nu ,\beta ,\nu')$ of $P_{\kappa ,\kappa'}$ such that $\zeta$ is closer to $\kappa$ than $\zeta'$ and $\nu$ in $P_{\kappa ,\kappa'}$ and $\nu$ is closer to $\kappa$ than $\nu'$ in $P_{\kappa ,\kappa'}$. Set
$$A^{\prime\prime}=(A'\setminus \{\{\kappa ,\kappa'\}\})\cup \{\{
\kappa ,\zeta \},\{\alpha ,\alpha \},\{\zeta',\nu \},\{\beta ,\beta
\},\{\nu',\kappa'\}\}.$$
Since $\{P_{\gamma ,\gamma'}\mid \{\gamma,\gamma'\}\in A'\}$ is a connector of $A'$, we conclude that in both cases $A^{\prime\prime}\in\Gamma_{n-1}$. Set
$$A^{\prime\prime\prime}=\{\{\zeta ,\zeta'\},\{\alpha ,\beta \},\{
\nu ,\nu'\}\}.$$
Since $\chi (\zeta )=\chi (\zeta')=\chi (\beta )\ne\chi (\alpha )
=\chi (\nu )=\chi (\nu')$ we conclude that $A^{\prime\prime\prime}$
is balanced and, by Proposition~\ref{old-results-prop}\,\eqref{old-results-prop-part7}, $A^{\prime\prime\prime}$ is connectable.  Since
$\iota_{i,0}(A^{\prime\prime},A^{\prime\prime\prime})\overset{*}{\implies}A$, by Lemma~\ref{lemma-simple} we can conclude that $A\in\Gamma_n$.

If \eqref{lem-4-1:C} occurs then there exists a subpath $(\zeta ,\alpha ,\beta ,\zeta')$ of $P_{\kappa ,\kappa'}$ such that $\zeta$ is closer to $\kappa$ than $
\zeta'$ in $P_{\kappa ,\kappa'}$.  Set
$$A^{\prime\prime}=(A'\setminus \{\{\kappa ,\kappa'\}\})\cup \{\{
\kappa ,\zeta \},\{\alpha ,\beta \},\{\zeta',\kappa'\}\}.$$
Since $\{P_{\nu ,\nu'}\mid \{\nu ,\nu'\}\in A'\}$ is a connector of $A'$, we conclude that $A^{\prime\prime}\in\Gamma_{n-1}$. Since $\chi (\zeta )\ne\chi (\zeta')$, $A^{\prime\prime\prime}=\{\{\zeta ,\zeta'\}\}$ is connectable by Proposition~\ref{old-results-prop}\,\eqref{old-results-prop-part1}, and since $\iota_{i,0}(A^{\prime\prime},A^{\prime\prime\prime})\overset{*}{\implies}A$, we have $A\in\Gamma_n$ as required.
\end{proof}

\begin{lemma}
\label{lem-4-2}
Let $n\ge 5$ and let $A\in\Omega_n$ be such that $|A|\le n-1$ and
$\sigma_i(A)=(|A|-1,0)$ for some $i\in [n]$. If $\oddp_{n-1}^{n-2}\subseteq\Gamma_{n-1}$
then $A\in\Gamma_n$.
\end{lemma}
\begin{proof}
Since $\sigma_i(A)=(|A|-1)$, there exists $\{\alpha',\beta'\}\in A$ with $\alpha'(i)\ne\beta'
(i)$ while $\alpha (i)=\beta (i)=0$ for all $\{\alpha ,\beta \}\in A\setminus \{\{\alpha',\beta'\}\}$. We can assume that $\alpha'(i)=0$. Set
$\hat{\alpha}=\rho_{i=0}(\alpha')$, $\hat{\beta}=\rho_{i=1}(\beta')$ and $A'=\rho_{i=0}(A)$.
Since $A'\in\conp_{n-1}$, there exists a connector $\{P_{\gamma,\gamma'}\mid \{\gamma ,\gamma'\}\in A'\}$ of $A'$.  Hence there exists $\{\kappa ,\kappa'\}\in A'$ such that $\hat{\alpha}$ belongs to the path $P_{\kappa ,\kappa'}$.  Let $(\zeta ,\hat{\alpha},\zeta')$ be a subpath of $P_{\kappa ,\kappa'}$.  Assuming without loss of generality that $\zeta$ is
closer to $\kappa$ than $\zeta'$, set
$$A^{\prime\prime}=(A'\setminus \{\{\kappa ,\kappa'\}\})\cup \{\{
\kappa ,\zeta \},\{\hat{\alpha},\hat{\alpha}\},\{\zeta'
,\kappa'\}\}$$
and note that $A^{\prime\prime}\in\conp_n$ because $\{P_{\gamma,\gamma'}\mid \{\gamma ,\gamma'\}\in A'\}$ is a connector
of $A'$. Let  $$
A^{\prime\prime\prime}=\{\{\zeta ,\zeta'\},\{\hat\alpha,\hat\beta\}\}.
$$
Since $\chi(\zeta)=\chi(\zeta')\ne\chi(\hat\alpha)=\chi(\hat\beta)$ and $\zeta\ne\zeta'$, we conclude that $A^{\prime\prime\prime}\in\balp^2_{n-1}$ with  $\|A^{\prime\prime\prime}\|=0$, and  therefore by parts \eqref{old-results-prop-part2} and \eqref{old-results-prop-part5} of Proposition~\ref{old-results-prop}, $A^{\prime\prime\prime}\in\conp_{n-1}$. Since $\iota_{i,0}(A^{\prime\prime},A^{\prime\prime\prime})\overset {*}{\implies}A$, it follows that $A\in\Gamma_n$.
\end{proof}

\begin{lemma}
\label{lem-4-3}
Let $n\ge 5$ and let $A\in\Omega_n^{n-1}$ be such that $\sigma_i(A)=(|A|-2,0)$ for some $i\in[n]$.
If $\Upsilon_{n-1}\subseteq\Gamma_{n-1}$, then $A\in\Gamma_n$.
\end{lemma}
\begin{proof}
If the assumptions of the lemma are satisfied, then there exist distinct
$\{\alpha_0,\beta_0\},\{\alpha_1,\beta_1\}\in A$ with $\alpha_l(i)=0\ne\beta_l(i)$
for $l=0,1$ and $\alpha (i)=\beta (i)=0$ for all
$\{\alpha ,\beta \}\in A\setminus\{\{\alpha_0,\beta_0\},\{\alpha_1,\beta_1\}\}$.

(1) First assume that $\chi (\alpha_0)=\chi (\alpha_1)$.  Select a vertex
$\gamma\in V_n\setminus\bigcup A$ with $\{\alpha_1,\gamma \}\in E(Q_n)$ and
$\gamma (i)=0$. Since $n\ge5$ and $\chi (\alpha_0)=\chi (\alpha_1)$, such a~vertex $\gamma$ exists.  Let
$$B=(A\setminus \{\{\alpha_0,\beta_0\},\{\alpha_1,\beta_1\}\})\cup
\{\{\alpha_1,\gamma \}\}.$$
Since $A$ is an odd pair-set, we infer that $A'=\rho_{i=0}(B)$ is an odd pair-set
containing an edge pair and $|A'|=|B|=|A|-1$. Therefore
$A'\in\Upsilon_{n-1}\subseteq\Gamma_{n-1}$. Let $
\{P_{\nu,\nu'}\mid \{\nu,\nu'\}\in A'\}$ be a
connector of $A'$.  Then there exist $\{\kappa ,\kappa'\}\in A'$ and a subpath
$(\zeta ,\rho_{i=0}(\alpha_0),\zeta')$ of
$P_{\kappa ,\kappa'}$ such that $\zeta$ is closer to $\kappa$ than $
\zeta'$ in $P_{\kappa ,\kappa'}$.  Set
$$A^{\prime\prime}=(A'\setminus \{\{\kappa ,\kappa'\}\})\cup \{\{
\kappa ,\zeta \},\{\rho_{i=0}(\alpha_0),\rho_{i=0}(\alpha_0)\},\{
\kappa',\zeta'\}\}.$$
Then $A^{\prime\prime}$ is connectable because
$\{P_{\nu ,\nu'}\mid \{\nu ,\nu'\}\in A'\}$ is a connector of $A'$.
Observe that
\[
\begin{split}
\chi(\rho_{i=0}(\gamma))\ne&\chi(\rho_{i=0}(\alpha_1))=\chi (\rho_{i=1}(\beta_1))=\chi(\rho_{i=0}(\alpha_0))\\
\chi (\zeta )=&\chi (\zeta')\ne\chi (\rho_{i=0}(\alpha_0))=\chi (\rho_{i=1}(\beta_0)).
\end{split}
\]
If $\rho_{i=0}(\gamma)\notin\{\zeta,\zeta'\}$
we infer that $\{\rho_{i=0}(\alpha_0),\rho_{i=1}(\beta_0)\}$, $\{\zeta,\zeta'\}$ and $\{\rho_{i=0}(\gamma),\rho_{i=1}(\beta_1)\}$
are pairwise disjoint. Hence $$A^{\prime\prime\prime}=\{\{\rho_{i=0}(\alpha_0),\rho_{i=1}(\beta_
0)\},\{\zeta ,\zeta'\},\{\rho_{i=0}(\gamma),\rho_{i=1}(\beta_1)\}\}$$
is balanced. By Proposition~\ref{old-results-prop}\;\eqref{old-results-prop-part7},
$A^{\prime\prime\prime}\in\Gamma_{n-1}$ and since
$\iota_{i,0}(A^{\prime\prime},A^{\prime\prime\prime})\overset {*}{\implies}A$, we conclude that $A\in\Gamma_n$.
If $\rho_{i=0}(\gamma)\in\{\zeta,\zeta'\}$ then $\{\kappa,\kappa'\}=\{\rho_{i=0}(\alpha_1),\rho_{i=0}(\gamma)\}$
and we can assume that $\rho_{i=0}(\gamma)=\zeta'$. Then
$$A''=(A'\setminus\{\{\rho_{i=0}(\alpha_1),\rho_{i=0}(\gamma)\}\}\cup\{\{\rho_{i=0}(\alpha_1,\zeta)\},\{\rho_{i=0}(\alpha_0),\rho_{i=0}(\gamma)\}\}$$
is a connectable pair-set. Then
$$A'''=\{\{\rho_{i=0}(\gamma),\rho_{i=1}(\beta_0)\},\{\rho_{i=1}(\beta_1),\zeta\}\}$$
is an odd pair-set and, by Proposition~\ref{old-results-prop}\;\eqref{old-results-prop-part3},
$A'''$ is connectable and from $\iota_{i,0}(A'',A''')\overset{*}{\implies}A$ it
follows that $A$ is connectable.

(2) Therefore it remains to settle the case that $\chi (\alpha_0)\ne\chi (\alpha_1)$.
Then
$$B=\rho_{i=0}((A\setminus \{\{\alpha_0,\beta_0\},\{\alpha_1,\beta_1\}\})\cup \{\{\alpha_0,\alpha_1\}\})$$
is an odd pair-set in $Q_{n-1}$ with $|B|=|A|-1$. To simplify the notation, denote $\rho_{i=0}(\alpha_0)$, $\rho_{i=0}(\alpha_1)$,
$\rho_{i=1}(\beta_0)$, $\rho_{i=1}(\beta_1)$ by $\hat\alpha_0$, $\hat\alpha_1$,
$\hat\beta_0$, $\hat\beta_1$.

(2.1)  First assume that $B\in\Upsilon_{n-1}\subseteq\Gamma_{n-1}$.  Note that then there exists a
connector $\{P_{\nu ,\nu'}\mid \{\nu ,\nu'\}\in B\}$ of $B$ and consider three subcases.

(2.1.1) The length of the path $P_{\hat\alpha_0,\hat\alpha_1}$ is at least $5$. Then
there exists an edge $\{\gamma ,\gamma'\}$ of $P_{\hat\alpha_0,\hat\alpha_1}$ with
$\{\gamma ,\gamma'\}\cap \{\hat\beta_0,\hat\beta_1\}=\emptyset$.
Without loss of generality we assume that $\gamma$ is closer to $\hat\alpha_0$ than
$\gamma'$ in $P_{\hat\alpha_0,\hat\alpha_1}$.  Let
$$A^{\prime}=(B\setminus \{\{\hat\alpha_0,\hat\alpha_1\}\})\cup \{\{\hat\alpha_0,\gamma \},\{\gamma',\hat\alpha_1\}\}$$
Then $A^{\prime}$ is connectable because $\{P_{\nu ,\nu'}\mid
\{\nu ,\nu'\}\in B\}$ is a connector
of $B$.  Set
$$A''=\{\{\hat\beta_0,\gamma \},\{\hat\beta_1,\gamma'\}\}.$$
Then $A''$ is balanced because $\chi (\beta_0)\ne\chi (\beta_1)$,
$\beta_0(i)=\beta_1(i)$ and  $\chi (\gamma )\ne\chi (\gamma')$.  Since $n-1\ge4$, by
Proposition~\ref{old-results-prop}\;\eqref{old-results-prop-part3} and \eqref{old-results-prop-part5}
we conclude that $A''\in\Gamma_{n-1}$.  Since
$\iota_{i,0}(A^{\prime},A'')\overset {*}{\implies}A$, it follows that $A\in\Gamma_n$.

(2.1.2)  $P_{\hat\alpha_0,\hat\alpha_1}=(\hat\alpha_0,\kappa ,\kappa',\hat\alpha_1)$.
If there exists an edge $\{\gamma ,\gamma'\}$ on $P_{\hat\alpha_0,\hat\alpha_1}$ with
$$\{\gamma ,\gamma'\}\cap \{\hat\beta_0,\hat\beta_1\}=\emptyset,$$
then the construction of the previous case (2.1.1) applies.
As $\chi (\hat\beta_0)\ne\chi (\hat\beta_1)$, otherwise it must be the case that  $\{\hat\beta_0,\hat\beta_1\}=\{\kappa ,\kappa'\}$.
Hence  $P_{\hat\alpha_0,\hat\alpha_1}=(\hat\alpha_0,\hat\beta_1 ,\hat\beta_0,\hat\alpha_1)$
because $\chi (\hat\alpha_0)=\chi (\hat\beta_0)$. Then set
$$A^{\prime\prime}=(B\setminus \{\{\hat\alpha_0,\hat\alpha_1\}\})\cup \{\{\hat\alpha_0,\hat\beta_0\},\{\hat\alpha_1,\hat\alpha_1\}\}$$
Since $\{P_{\nu ,\nu'}\mid \{\nu ,\nu'\}\in B\}$ is a connector of $B$, it follows that $A^{\prime\prime}\in\Gamma_{n-1}$.  Since
$$A^{\prime\prime\prime}=\{\{\hat\alpha_1,\hat\beta_1\},\{\hat\beta_0,\hat\beta_0\}\}$$
is balanced, by Proposition~\ref{old-results-prop}\,\eqref{old-results-prop-part2} we obtain that $A^{\prime\prime\prime}\in\Gamma_{n-1}$.
Since $\iota_{i,0}(A^{\prime\prime},A^{\prime\prime\prime})\overset {
*}{\implies}A$, we conclude that $A\in\Gamma_n$.

(2.1.3) $P_{\hat\alpha_0,\hat\alpha_1}=(\hat\alpha_0,\hat\alpha_1)$.
Then we can assume that $\{\hat\alpha_0,\hat\alpha_1\}\cap \{\hat\beta_0,\hat\beta_1\}\ne\emptyset$,
for otherwise the construction (2.1.1) applies.  If
$\{\hat\alpha_0,\hat\alpha_1\}\ne \{\hat\beta_0,\hat\beta_1\}$, then set
$$A^{\prime\prime}=(B\setminus \{\{\hat\alpha_0,\hat\alpha_1\}\})\cup \{\{\hat\alpha_0,\hat\alpha_0
\},\{\hat\alpha_1,\hat\alpha_1\}\}.$$
Clearly $A^{\prime\prime}$ is connectable because
$\{P_{\nu ,\nu'}\mid \{\nu ,\nu'\}\in B\}$ is a connector of $B$. Set
$$A^{\prime\prime\prime}=\{\{\hat\alpha_0,\hat\beta_0\},\{\hat\alpha_1,\hat\beta_1\}\}.$$
Since $\chi (\alpha_0)=\chi (\beta_1)\ne\chi (\alpha_1)=\chi (\beta_0)$ and
$\alpha_0(i)=\alpha_1(i)\ne\beta_0(i)=\beta_1(i)$, it follows that
$A^{\prime\prime\prime}$ is balanced and therefore, by
Proposition~\ref{old-results-prop}\,\eqref{old-results-prop-part2}, $A'''$ is connectable.
Since $\iota_{i,0}(A^{\prime\prime},A^{\prime\prime\prime})\overset {*}{\implies}A$,
we conclude that $A\in\Gamma_n$.

If, on the other hand, $\{\hat\alpha_0,\hat\alpha_1\}=\{\hat\beta_0,\hat\beta_1\}$,
select $\{\alpha',\beta'\}\in B\setminus\{\{\hat\alpha_0,
\hat\alpha_1\}\}$ and  choose an edge $\{\gamma ,\gamma'\}$ on the path $P_{\alpha',\beta'}$ such that $
\gamma$ is closer to $\alpha'$ than
$\gamma'$.  Set
$$A^{\prime\prime}=(B\setminus \{\{\hat\alpha_0,\hat\alpha_1\},\{\alpha',\beta'\}\})\cup
\{\{\hat\alpha_0,\hat\alpha_0\},\{\hat\alpha_1,\hat\alpha_1\},\{\alpha',\gamma \},\{\gamma',\beta'\}\}
,$$
then $A^{\prime\prime}$ is connectable because $\{P_{\nu ,\nu'}\mid
\{\nu ,\nu'\}\in B\}$ is a connector
of $B$.  Let
$$A^{\prime\prime\prime}=\{\{\hat\beta_0,\hat\beta_0\},\{\hat\beta_1,\hat\beta_1\},\{\gamma ,\gamma'
\}\}.$$
Since $\chi (\beta_0)\ne\chi (\beta_1)$ and $\chi (\gamma)\ne\chi (\gamma')$, we infer
that $A^{\prime\prime\prime}$ is balanced and therefore connectable by Proposition~\ref{old-results-prop}\,\eqref{old-results-prop-part7}.
Since $\{\alpha_0,\beta_0\}$ and $\{\alpha_1,\beta_1\}$ are edges of $Q_n$, we infer
that $\iota_{i,0}(A^{\prime\prime},A''')\overset{*}{\implies}A$ and thus $A\in\Gamma_n$.

(2.2) It remains to settle the case that $\chi (\alpha_0)\ne\chi (\alpha_1)$ and  $B\notin\Upsilon_{n-1}$. Then $n=5$ and $B=\{\{\kappa_j,\nu_j\}\mid j\in[3]\}$ such that $\{\kappa_j,\nu_j\}\notin E(Q_4)$ for all $j\in[3]$ and there exists a set $X\subseteq V_4$ such that $\kappa_j,\nu_j
\in X$ for all $j\in[3]$ and the subgraph of $Q_4$ induced by the set $X$ is isomorphic to $Q_3$. Note that then $d(\kappa_j,\nu_j)=3$ for all $i\in[3]$. Assume that
$$A\setminus \{\{\alpha_0,\beta_0\},\{\alpha_1,\beta_1\}\})\cup \{\{\alpha_0,\alpha_1\}\}=\{\{\kappa_j',\nu'_j\}\mid j\in[3]\}$$
where $\rho_{i=0}(\{\kappa'_j,\nu'_j\})=\{\kappa_j,\nu_j\}$ and that $\{\alpha_0,\alpha_1\}=\{\kappa'_2,\nu'_2\}$.
By Proposition~\ref{old-results-prop}\,\eqref{old-results-prop-part9}, we obtain that
$$A'=\{\{\kappa_0,\nu_0\},\{\kappa_1,\nu_1\},\{\kappa_2,\kappa_2\}
,\{\nu_2,\nu_2\}\}$$
is connectable. From $\chi (\alpha_0)\ne\chi (\alpha_1)$ it follows  that if $\gamma\in \{\kappa_2,\nu_2\}\cap \{\hat\beta_0,\hat\beta_1\}$, then either $\gamma =\hat\alpha_0=\hat\beta_0$, or $\gamma =\hat\alpha_1=\hat\beta_1$. If  $\{\hat\beta_0,\hat\beta_1\}\ne \{\kappa_2,\nu_2\}$, then
$$A^{\prime\prime}=\{\{\hat\alpha_0,\hat\beta_0\},\{\hat\alpha_1,\hat\beta_1\}$$
is a balanced pair-set and, by Proposition~\ref{old-results-prop}\,\eqref{old-results-prop-part2} and \eqref{old-results-prop-part5}, $A^{\prime\prime}\in\Gamma_4$. Then $\iota_{i,0}(A',A^{\prime\prime})\overset {*}{\implies}A$ completes the proof that $A\in\Gamma_
5$.
Assume that $\{\hat\beta_0,\hat\beta_1\}=\{\kappa_2
,\nu_2\}$. Consider a connector
$\{P_{\zeta ,\zeta'}\mid \{\zeta ,\zeta'\}\in A'\}$, then there exist $
\{\zeta ,\zeta'\}\in A'$ and an edge $\{\gamma ,\gamma'\}$ of
$P_{\zeta ,\zeta^{\prime\prime}}$ such that $\gamma$ is closer to $
\zeta$ than $\gamma'$ in $P_{\zeta ,\zeta'}$. Set
$$A^{\prime\prime}=(A'\setminus \{\{\zeta ,\zeta'\}\})\cup \{\{\zeta
,\gamma \},\{\gamma',\zeta'\}\}.$$
Since $\{P_{\zeta ,\zeta'}\mid \{\zeta ,\zeta'\}\in A'\}$ is a connector of $A'$, we conclude that $A^{\prime\prime}\in\Gamma_4$. Set
$$A^{\prime\prime\prime}=\{\{\gamma ,\gamma'\},\{\hat\alpha_0,\hat\beta_0\},\{\hat\alpha_1,\hat\beta_1\}\}.$$
From $\hat\alpha_0=\hat\beta_0$ and $\hat\alpha_1=\hat\beta_1$ it follows that $A^{\prime\prime\prime}$ is balanced and, by Proposition~\ref{old-results-prop}\,\eqref{old-results-prop-part7}, $A^{\prime\prime\prime}\in\Gamma_4$. Then $\iota_{i,0}(A',A''')\overset{*}{\implies}A$ completes the proof that $A\in\Gamma_5$.
\end{proof}
\begin{lemma}
\label{2-even2-edgepairs}
Let $n\ge 5$ and $A=\{\{\alpha_i,\beta_i\}\mid i\in[4]\}
\in \bold B^4_n$ such that $\{\alpha_0,\beta_0\}$ is an edge-pair, $\{\alpha_1,\beta_1\}$ is an odd pair, $\{\alpha_2,\beta_2\}$ and 
$\{\alpha_3,\beta_3\}$ are even pairs and there is $i\in [n]$ such that one of the following conditions holds:
\begin{enumerate}[\upshape(1)]
\item \label{2-even2-edgepairs:1}
$\alpha_0(i)=\beta_0(i)\ne\alpha_1(i)=\beta_1(i)$,
\item \label{2-even2-edgepairs:2} 
$\alpha_0(i)=\beta_0(i)=\alpha_1(i)=\alpha_2(i)\ne\beta_1(i)=\beta_2(i)=\alpha_3(i)=\beta_3(i)$,
\item \label{2-even2-edgepairs:3}
$\alpha_0(i)=\alpha_1(i)=\alpha_2(i)=\beta_2(i)\ne\beta_1(i)=\alpha_3(i)=\beta_3(i)$ and $\chi(\alpha_1)\ne\chi(\alpha_2)$.
\end{enumerate}
Then $A\in\Gamma_n$.
\end{lemma}
\begin{proof}
Without loss of generality we can assume that $\alpha_0(i)=0$. First consider the case \eqref{2-even2-edgepairs:1}.
By Lemma~\ref{cor-on-construction-simple}\,\eqref{cor-on-construction-simple:2}, there exists a simple  $i$-completion $B$ of $A$ and, by Lemma~\ref{lemma-on-construction-simple-properties}, there exists $\ell\in [2]$ such that $|\rho_{i=\ell}(B)|=\|\rho_{i=\ell}(B)\|\in\{1,2,3\}$ while $\rho_{i=1-\ell}(B)$ consists of one odd and two even pairs of opposite parities. Note that we can assume that if $\rho_{i=\ell}(B)$ consists of three odd pairs then $\{\alpha_0,\beta_0\}$ is one of them. Thus, by  Corollary~\ref{cor-dim-4} and parts \eqref{old-results-prop-part6} and \eqref{old-results-prop-part7} of Proposition~\ref{old-results-prop}, $\rho_{i=\ell}(B)\in\Gamma_{n-1}$ for both $\ell\in [2]$ and therefore $A\in\Gamma_
n$.
\par
If condition~\eqref{2-even2-edgepairs:2} holds then, by Lemma~\ref{cor-on-construction-simple}\,\eqref{cor-on-construction-simple:2}, there exists a simple $i$-completion $B$ of $A$, and by Lemma~\ref{lemma-on-construction-simple-properties} we have $|\rho_{i=1}(B)|=3$, $\|\rho_{i=1}(B)\|=1$, $|\rho_{i=0}(B)|
=\|\rho_{i=0}(B)\|=3$ and 
$\rho_{i=0}(B)$ contains an edge-pair. Thus $\rho_{i=0}(B)$ and $\rho_{i=1}(B)$ are connectable by Corollary~\ref{cor-dim-4} (if $n=5$) or Proposition~\ref{old-results-prop}\,\eqref{old-results-prop-part6} (if $n>5$), and by Proposition~\ref{old-results-prop}\,\eqref{old-results-prop-part7}, respectively, and therefore $A\in\Gamma_n$. 

It remains to settle the case \eqref{2-even2-edgepairs:3}.
Since $V_{n-1}$ contains at least eight but $\rho_{i=0}(\bigcup A)\cup\rho_{i=1}(\bigcup A)$ contains at most six vertices of parity $p$ for both $p\in\{-1,1\}$, there exist $\gamma,\gamma'\in V_n$ such that $\gamma (i)=\gamma'(i)=0$, $\chi(\gamma)\ne\chi(\alpha_0)$, $\chi(\gamma')=\chi(\alpha_1)$ and $\gamma ,\gamma\oplus e_i, \gamma' ,\gamma'\oplus e_i\notin\bigcup A$.  If $\beta_0(i)=1$, set 
$$
B=(A\setminus\{\{\alpha_0,\beta_0\},\{\alpha_1,\beta_1\}\})\cup\{\{\alpha_0,\gamma \},\{\gamma\oplus e_i,\beta_0\},\{\alpha_1,\gamma' \},\{\gamma'\oplus e_i,\beta_1\}\}.
$$
Then for both $\ell\in[2]$, $\rho_{i=\ell}(B)$ consists of one odd pair and two even pairs of opposite parities and therefore is connectable by Proposition~\ref{old-results-prop}\,\eqref{old-results-prop-part7}. If $\beta_0(i)=0$, set
$$
B=(A\setminus\{\{\alpha_1,\beta_1\}\})\cup\{\{\alpha_1,\gamma' \},\{\gamma'\oplus e_i,\beta_1\}\}.
$$
Then $\rho_{i=0}(B)$ consists of one edge-pair and two even pairs of opposite parities while $\rho_{i=1}(B)$ consists of two pure even pairs of opposite parities. Hence $\rho_{i=\ell}(B)\in\Gamma_{n-1}$ for both $\ell\in[2]$ by parts \eqref{old-results-prop-part7} and \eqref{old-results-prop-part5} of Proposition~\ref{old-results-prop}, respectively. In both cases we have $B\implies A$ and therefore $A\in\Gamma_n$.
The proof is complete.
\end{proof}
The second part of this section investigates diminishable pair-sets $A\in\Upsilon_n$ such that $|A|=n$ and $\sigma_i(A)=(n_0,1)$ where $n_0>n-3$ for a separating $i\in[n]$.
\begin{lemma}
\label{lemma-jeden-vpravo}
Let $A\in\Upsilon_n$ for $n\ge 6$ such that
$|A|=n$ and $\sigma_i(A)=(n-1,1)$ for some separating $i\in[n]$.
If $\Upsilon_{n-1}\subseteq\Gamma_{n-1}$ then $A\in\Gamma_n$.
\end{lemma}
\begin{proof}Assume that $A\in\Upsilon_n$ satisfies the assumption of the lemma. Let $\{\alpha_0,\beta_0\},\{\alpha_1,\beta_1\}\in A$ be edge pairs
such that $\alpha_0(i)=\beta_0(i)=0$ and $\alpha_1(i)=\beta_1(i)=1$.  Then $\alpha (i)=\beta (i)=0$
for all $\{\alpha ,\beta \}\in A_1=A\setminus \{\{\alpha_0,\beta_0\},\{\alpha_1,\beta_1\}\}$.
Let us denote
$$X=\{\gamma\in V_n\mid\gamma(i)=0,\,\{\{\gamma\oplus e_i,\alpha_1\},\{\gamma\oplus e_i,\beta_1\}\}\cap E(Q_n)\ne\emptyset\}.$$
Recall that by Observation~\ref{obser-encom} there exist no two distinct
$\gamma ,\gamma'\in V_{n-1}$ such that $\chi (\gamma )=\chi (\gamma')$ and
$\gamma,\gamma'\in\enco(\rho_{i=0}(A))$.
\setcounter{claim}{0}\begin{claim}
\label{lemma-jeden-vpravo:claim-A}
Assume that there 
\begin{enumerate}[\upshape(i)]
\item\label{lemma-jeden-vpravo:claim-A-(1)}
exists a pair $\{\alpha ,\beta \}\in A_1$ such that $\alpha,\beta\notin X$,
\item
\label{lemma-jeden-vpravo:claim-A-(2)}
or $\{\alpha_0,\beta_0\}=\{\alpha_1\oplus e_i,\beta_1\oplus e_i\}$ and $X=\bigcup A\setminus\{\alpha_1,\beta_1\}$.
\end{enumerate}
Then $A\in\Gamma_n$.
\end{claim}
\begin{proof}[Proof of Claim~\ref{lemma-jeden-vpravo:claim-A}]
Let $A'=\rho_{i=0}(A\setminus \{\{\alpha ,\beta \}\})$ 
where $\{\alpha ,\beta \}$ either satisfies \eqref{lemma-jeden-vpravo:claim-A-(1)}, or in case \eqref{lemma-jeden-vpravo:claim-A-(2)} is chosen arbitrarily from $A_1$.  Then $A'\in\oddp_{n-1}$ and $|A'|=n-2$, thus $A'$ is diminishable and, by the assumption, $A'\in\Gamma_{n-1}$.  The proof that $A\in
\Gamma_n$ is
similar to the proof of Lemma~\ref{lem-4-1}.  Since $A'\in\Gamma_{n-1}$, there exists
a connector $\{P_{\kappa ,\kappa'}\mid \{\kappa ,\kappa'\}\in A'\}$ of $
A'$.  Note that in case \eqref{lemma-jeden-vpravo:claim-A-(2)} we have
\begin{align}\tag{ii'}
\label{lemma-jeden-vpravo:claim-A-(2')}
\text{if $\{\kappa,\kappa'\}\ne\{\alpha_0,\beta_0\}$ then each vertex $\gamma$ of $P_{\kappa ,\kappa'}$ satisfies $\gamma\oplus e_i\not\in\bigcup A$.}
\end{align}
Then one of the following possibilities occurs:
\begin{enumerate}[\upshape(A)]
\item
\label{lemma-jeden-vpravo:claim-A:1}
there exist distinct $\{\kappa ,\kappa'\},\{\xi ,\xi'\}\in A'$ such that $
\rho_{i=0}(\alpha )$ belongs to $P_{\kappa ,\kappa'}$
and $\rho_{i=0}(\beta )$ belongs to $P_{\xi ,\xi'}$;
\item
\label{lemma-jeden-vpravo:claim-A:2}
there exists $\{\kappa ,\kappa'\}\in A'$ such that $\{\rho_{i=0}(\alpha ),\rho_{i
=0}(\beta )\}$ is not an edge of $P_{\kappa ,\kappa'}$ and $\rho_{i=0}(\alpha
)$ and $\rho_{i=0}(\beta )$ belong to
the path $P_{\kappa ,\kappa'}$;
\item
\label{lemma-jeden-vpravo:claim-A:3}
there exists $\{\kappa ,\kappa'\}\in A'$ such that $\{\rho_{i=0}(
\alpha ),\rho_{i=0}(\beta )\}$ is an edge of
the path $P_{\kappa ,\kappa'}$.
\end{enumerate}
If \eqref{lemma-jeden-vpravo:claim-A:1} occurs then there exist a subpath $(\zeta ,\rho_{i=0}(\alpha),\zeta')$ of $P_{\kappa ,\kappa'}$
such that $\zeta$ is closer to $\kappa$ than $\zeta'$ in $P_{\kappa,\kappa'}$ and a subpath
$(\nu ,\rho_{i=0}(\beta ),\nu')$ of $P_{\xi ,\xi'}$ such that $\nu$ is closer to $\xi$ than
$\nu'$ in $P_{\xi,\xi'}$. Using \eqref{lemma-jeden-vpravo:claim-A-(1)} or \eqref{lemma-jeden-vpravo:claim-A-(2')},
we conclude that
$$\{\zeta ,\rho_{i=0}(\alpha ),\zeta',\nu ,\rho_{i=0}(\beta ),\nu'
\}\cap \{\rho_{i=1}(\alpha_1),\rho_{i=1}(\beta_1)\}=\emptyset.$$
Then
\[
\begin{split}
A^{\prime\prime}=&(A'\setminus \{\{\kappa ,\kappa'\},\{\xi
,\xi'\}\})\; \cup \\&\quad\cup \{\{\kappa ,\zeta \},\{\rho_{i=0}(\alpha ),\rho_{i=0}(\alpha
)\},\{\zeta',\kappa'\},\{\xi ,\nu \},\{\rho_{i=0}(\beta ),\rho_{i
=0}(\beta )\},\{\nu',\xi'\}\}
\end{split}
\]
is a connectable pair-set.
If \eqref{lemma-jeden-vpravo:claim-A:2} occurs then there exist vertex-disjoint subpaths $(\zeta ,\rho_{i=0}(\alpha ),\zeta')$ and
$(\nu ,\rho_{i=0}(\beta),\nu')$ of $P_{\kappa ,\kappa'}$ such that
$\zeta$ is closer to $\kappa$ than $\zeta'$ and $\nu$ in $P_{\kappa ,\kappa'}$ and $\nu$ is
closer to $\kappa$ than $\nu'$ in $P_{\kappa ,\kappa'}$. From \eqref{lemma-jeden-vpravo:claim-A-(1)} or \eqref{lemma-jeden-vpravo:claim-A-(2')} it follows that
$\{\zeta ,\rho_{i=0}(\alpha ),\zeta',\nu ,\rho_{i=0}(\beta ),\nu'
\}\cap \{\rho_{i=1}(\alpha_1),\rho_{i=1}(\beta_1)\}=\emptyset$ and
$$A^{\prime\prime}=(A'\setminus \{\{\kappa ,\kappa'\}\})\cup \{\{
\kappa ,\zeta \},\{\rho_{i=0}(\alpha ),\rho_{i=0}(\alpha )\},\{\zeta'
,\nu \},\{\rho_{i=0}(\beta ),\rho_{i=0}(\beta )\},\{\nu',\kappa'\}
\}$$
is a connectable pair-set. 
In both cases set
$$A^{\prime\prime\prime}=\{\{\rho_{i=1}(\alpha_1),\rho_{i=1}(\beta_
1)\},\{\rho_{i=0}(\alpha ),\rho_{i=0}(\beta )\},\{\zeta ,\zeta'\}
,\{\nu ,\nu'\}\}.$$
Our goal is to prove that $A^{\prime\prime\prime}$ satisfies the assumptions of 
Lemma~\ref{2-even2-edgepairs}.  Let $i\in [n]$ such that 
$\rho_{i=1}(\alpha_1)(i)=\rho_{i=1}(\beta_1)(i)\ne\rho_{i=0}(\alpha 
)(i)=\rho_{i=0}(\beta )(i)$.  
Then the assumptions of Lemma~\ref{2-even2-edgepairs}\,\eqref{2-even2-edgepairs:1} are satisfied and thus $A^{\prime\prime\prime}\in\Gamma_n$.  Thus we can assume that if 
$\rho_{i=1}(\alpha_1)(i)=\rho_{i=1}(\beta_1)(i)\ne\rho_{i=0}(\alpha 
)(i)$ then 
$\rho_{i=0}(\beta )(i)=\rho_{i=1}(\alpha_1)(i)$ and by the symmetry, if 
$\rho_{i=1}(\alpha_1)(i)=\rho_{i=1}(\beta_1)(i)\ne\rho_{i=0}(\beta 
)(i)$ then 
$\rho_{i=0}(\alpha )(i)=\rho_{i=1}(\alpha_1)(i)$.  Let us denote $
i_1=\zeta\Delta\rho_{i=0}(\alpha )$, $i_2=\zeta'\Delta\rho_{i=0}(
\alpha )$, 
$i_3=\nu\Delta\rho_{i=0}(\beta )$, $i_4=\nu'\Delta\rho_{i=0}(\beta 
)$.  If 
$\rho_{i=1}(\alpha_1)(i)=\rho_{i=1}(\beta_1)(i)\ne\rho_{i=0}(\alpha)
(i)$ and  
$i\ne i_1,i_2$, then the assumptions of Lemma~\ref{2-even2-edgepairs}\,\eqref{2-even2-edgepairs:2} or \eqref{2-even2-edgepairs:3} are satisfied and thus $A^{\prime\prime\prime}\in\Gamma_n$.  Thus we can assume if 
$\rho_{i=1}(\alpha_1)(i)=\rho_{i=1}(\beta_1)(i)\ne\rho_{i=0}(\alpha 
)(i)$ then 
$i\in \{i_1,i_2\}$, and, by symmetry, if $\rho_{i=1}(\alpha_1)(i)=
\rho_{i=1}(\beta_1)(i)\ne\rho_{i=0}(\beta )(i)$ then 
$i\in \{i_3,i_4\}$.  Thus 
$\rho_{i=0}(\alpha )\Delta\rho_{i=0}(\beta )\subseteq \{i_1,i_2,i_
3,i_4,\rho_{i=1}(\alpha_1)\Delta\rho_{i=1}(\beta_1)\}$ and 
$$\rho_{i=1}(\alpha_1)\Delta\zeta ,\rho_{i=1}(\beta_1)\Delta\zeta\subseteq 
\{i_2,\rho_{i=1}(\alpha_1)\Delta\rho_{i=1}(\beta_1)\}.$$ Without loss of 
generality we assume that $\rho_{i=1}(\alpha_1)\Delta\zeta =\{
i_2\}$ and 
$\rho_{i=1}(\alpha_1)\Delta\zeta'=\{i_1\}$. By symmetry, 
$$\rho_{i=1}(\alpha_1)\Delta\nu ,\rho_{i=1}(\beta_1)\Delta\nu\subseteq 
\{i_4,\rho_{i=1}(\alpha_1)\Delta\rho_{i=1}(\beta_1)\}$$ and 
$\rho_{i=1}(\alpha_1)\Delta\nu',\rho_{i=1}(\beta_1)\Delta\nu'\subseteq 
\{i_3,\rho_{i=1}(\alpha_1)\Delta\rho_{i=1}(\beta_1)\}$. If $i_j=i_
k$ for some 
$j=1,2$ and $k=3,4$, then for $i=i_j$ we have 
$\alpha_1(i)=\beta_1(i)\ne\alpha (i)=\beta (i)$ and $A$ therefore satisfies the assumptions of Lemma~\ref{2-even2-edgepairs}\,\eqref{2-even2-edgepairs:1}, which means that $A^{\prime\prime\prime}\in\Gamma_n$. Thus we 
can assume that $i_1,i_2,i_3,i_4$ are pairwise distinct and since 
$\{\alpha ,\beta \}$ is an odd pair, we conclude that 
$$\rho_{i=0}(\alpha )\Delta\rho_{i=0}(\beta )=\{i_1,i_2,i_3,i_4,\rho_{i=1}(\alpha_1)\Delta\rho_{i=1}(\beta_1)\}.$$ 
But then $i=\rho_{i=1}(\alpha_1)\Delta\rho_{i=1}(\beta_1)$ satisfies the assumptions of Lemma~\ref{2-even2-edgepairs}\,\eqref{2-even2-edgepairs:3}. In all cases, $A^{\prime\prime\prime}\in\Gamma_{n-1}$ by Lemma~\ref{2-even2-edgepairs}.
Since $\iota_{i,0}(A^{\prime\prime},A^{\prime\prime\prime})\overset{*}\implies A$, we conclude that $A\in\Gamma_n.$ 

If \eqref{lemma-jeden-vpravo:claim-A:3} occurs then there exists a subpath
$(\zeta ,\rho_{i=0}(\alpha ),\rho_{i=0}(\beta ),\zeta')$
of $P_{\kappa ,\kappa'}$ such that $\zeta$ is closer to $\kappa$ than $
\zeta'$ in $P_{\kappa ,\kappa'}$.  Set
$$A^{\prime\prime}=(A'\setminus \{\{\kappa ,\kappa'\}\})\cup \{\{
\kappa ,\zeta \},\{\rho_{i=0}(\alpha ),\rho_{i=0}(\beta )\},\{\zeta'
,\kappa'\}\}.$$
Since $\{P_{\nu ,\nu'}\mid \{\nu ,\nu'\}\in A'\}$ is a connector of $A'$, we conclude that
$A^{\prime\prime}\in\Gamma_{n-1}$. Since $\chi (\zeta )\ne\chi (\zeta')$ and \eqref{lemma-jeden-vpravo:claim-A-(1)} or \eqref{lemma-jeden-vpravo:claim-A-(2')} holds, by Proposition~\ref{old-results-prop}\;\eqref{old-results-prop-part3}, $A^{\prime\prime\prime}=\{\{\zeta ,\zeta'\},\{\rho_{i=1}(\alpha_1),\rho_{i=1}(\beta_1)\}\}$ is a connectable pair-set. Since $\iota_{i,0}(A^{\prime\prime},A^{\prime\prime\prime})\overset {*}{\implies}A$, we obtain that $A\in\Gamma_n$ and the proof of Claim~\ref{lemma-jeden-vpravo:claim-A} is complete.
\end{proof}
Therefore we can assume that $\{\alpha,\beta\}\cap X\ne\emptyset$ for every pair $\{\alpha ,\beta \}\in A\setminus \{\{\alpha_1,\beta_1\}\}$. 
\begin{claim}
\label{lemma-jeden-vpravo:claim-B}
Let $\{\alpha ,\beta \}\in A_1$ such that $\beta\notin X$. If there exists $\{\alpha,\kappa\}\in E(Q_n)$
such that $\kappa(i)=0$, $\kappa,\kappa\oplus e_i\notin\bigcup A$ and
$A'=\rho_{i=0}((A\setminus \{\{\alpha ,\beta \}\})\cup \{\{\alpha ,\kappa \}\})\in\dimp_{n-1}$,
then $A\in\Gamma_n$.
\end{claim}
\begin{proof}[Proof of Claim~\ref{lemma-jeden-vpravo:claim-B}]
Since $A_1\in\dimp_{n-1}\subseteq\conp_{n-1}$ by our assumptions, there exists a connector $\{P_{\zeta ,\zeta'}\mid \{\zeta ,\zeta'\}\in A'\}$ of $A'$. Hence there are $\{\zeta,\zeta'\}\in A^{\prime\prime}$ and a subpath $(\nu ,\rho_{i=0}(\beta ),\nu')$ of $P_{\zeta ,\zeta'}$ such that $\nu$ is
closer to $\zeta$ than $\nu'$ in $P_{\zeta ,\zeta'}$.  Then
$$A^{\prime\prime}=(A'\setminus \{\{\zeta ,\zeta'\}\})\cup \{\{\zeta
,\nu \},\{\rho_{i=0}(\beta ),\rho_{i=0}(\beta )\},\{\zeta',\nu'\}
\}$$
is connectable because $\{P_{\xi ,\xi'}\mid \{\xi ,\xi'\}\in A'\}$ is a connector of $
A'$.
Let
$$A^{\prime\prime\prime}=\{\{\rho_{i=1}(\alpha_1),\rho_{i=1}(\beta_
1)\},\{\rho_{i=0}(\beta ),\rho_{i=0}(\kappa )\},\{\nu ,\nu'\}\}.$$
Since $\beta\notin X$, $A'''$ is a pair-set and from $\chi (\nu )=\chi (\nu')\ne\chi (\rho_{i=0}(\beta ))=\chi (\rho_{
i=0}(\kappa ))$, $\chi (\alpha_1)\ne\chi (\beta_1)$  it follows, by Proposition~\ref{old-results-prop}\;\eqref{old-results-prop-part7}, that $A^{\prime\prime\prime}\in\Gamma_{n-1}$, we can conclude that $\iota_{i,0}(A^{\prime\prime},A^{\prime\prime\prime})\overset {*}{\implies}A$. By Lemma~\ref{lemma-simple} we  conclude that $A\in\Gamma_n$ and the proof of  Claim~\ref{lemma-jeden-vpravo:claim-B} is complete.
\end{proof}
\begin{claim}
\label{lemma-jeden-vpravo:claim-C}
If there exists $\gamma\in\enc(\rho_{i=0}(A))$ then $A\in\Gamma_n$.
\end{claim}
\begin{proof}[Proof of Claim~\ref{lemma-jeden-vpravo:claim-C}]
First observe that if $\eta\in\enc(\rho_{i=0}(A))$ then for $\nu=\iota_{i=0}(\eta)$ we have
$\nu\notin\bigcup A$, for otherwise $\eta\in\bigcup(\rho_{i=0}(A))$. Since $A\in\dimp_n$ implies $\enc(A)=\emptyset$, we conclude that $\nu\oplus e_i\notin\{\alpha_1,\beta_1\}$. Hence
$\rho_{i=1}(\alpha_1),\rho_{i=1}(\beta_1)\notin\enc(\rho_{i=0}(A))$. Since
$\enco(\rho_{i=0}(X))=\{\rho_{i=1}(\{\alpha_1,\beta_1\})$, it follows that 
\begin{enumerate}[\upshape(a)]
\item 
\label{lemma-jeden-vpravo:claim-C:1}
for $\eta=\rho_{i=0}(\nu)\in\enc(\rho_{i=0}(A))$ there are at most two $\kappa\in X$ with $\{\nu,\kappa\}\in E(Q_n)$.
\end{enumerate}
By Observation~\ref{obser-encom}, if $\eta=\rho_{i=0}(\nu),\eta'=\rho_{i=0}(\nu')\in\enc(\rho_{i=0}(A))$ are distinct, then $\chi(\eta)\ne\chi(\eta')$.
Since $|\rho_{i=0}(A)|=n-1$, we conclude that 
\begin{enumerate}[\upshape(a)]
\setcounter{enumi}{1}
\item 
\label{lemma-jeden-vpravo:claim-C:2}
for such $\nu,\nu'$ and every $\{\alpha,\beta\}\in A$ with $\alpha(i)=\beta(i)=0$ we have either $\{\alpha,\nu\},\{\beta,\nu'\}\in E(Q_n)$ or $\{\alpha,\nu'\},\{\beta,\nu\}\in E(Q_n)$.
\end{enumerate}
Assume that there exist two distinct $\gamma,\gamma'\in\enc(\rho_{i=0}(A))$. If
$\eta\in\enco(\rho_{i=0}(\bigcup A)\cup\{\gamma\})$
then either $\{\eta,\gamma\}\notin E(Q_n)$ and $\eta\in\enc(\rho_{i=0}(A))$ or
$\{\eta,\gamma\}\in E(Q_n)$ and $\eta\in\rho_{i=0}(\bigcup A)$. Thus
$\enc(\rho_{i=0}(\bigcup A)\cup\{\gamma\})\subseteq\enc(\rho_{i=0}(A))$. Therefore there exists
$\{\alpha,\beta\}\in A_1$ such that $\{\alpha,\gamma\},\{\beta,\gamma'\}\in E(Q_n)$ (by \eqref{lemma-jeden-vpravo:claim-C:2}) and $\beta\notin X$ (by \eqref{lemma-jeden-vpravo:claim-C:1}).
Then $$A'=\rho_{i=0}(A\setminus\{\{\alpha,\beta\}\})\cup\{\{\rho_{i=0}(\alpha),\gamma\}\}$$
is an odd pair-set such that $|A'|=n-1$, $A'$ contains two edge pairs, namely
$\{\rho_{i=0}(\alpha_0),\rho_{i=0}(\beta_0)\}$ and $\{\rho_{i=0}(\alpha),\gamma\}$. Moreover,
$\enc(A')\subseteq\enc(\rho_{i=0}(A)\cup\{\gamma\})=\{\gamma'\}$,  but $\gamma'$ is not encompassed by $A'$.
Consequently, $A$ is connectable by Claim~\ref{lemma-jeden-vpravo:claim-B}. 

Thus we can assume that
$\{\gamma\}=\enc(\rho_{i=0}(A))$. If there exists $\{\alpha,\beta\}\in A_1$ such that
$\{\alpha,\gamma\}\in E(Q_n)$ and $\beta\notin X$, then analogously as in the foregoing case we obtain
that $A'$ is diminishable and, by Claim~\ref{lemma-jeden-vpravo:claim-B}, $A$ is connectable.  Observe
that for every $\{\alpha,\beta\}\in A_1$ we have
$\{\{\alpha,\gamma\},\{\beta,\gamma\}\}\cap E(Q_n)\ne\emptyset$. Thus we can assume that for every
$\{\alpha,\beta\}\in A_1$ we have $\{\alpha,\gamma\}\in E(Q_n)$ and $\beta\in X$. Moreover, by \eqref{lemma-jeden-vpravo:claim-C:1}, at least $n-4>1$ of these pairs satisfy $\alpha\notin X$.
Since for every $\eta\in\bigcup A$ with $\eta(i)=0$ and $\chi(\eta)\ne\chi(\beta)$ we have
$\{\eta,\gamma\}\in E(Q_n)$ and since there are at most two vertices $\kappa\in V_n$ such that
$\{\kappa,\gamma\},\{\kappa,\beta\}\in E(Q_n)$, we conclude that there exists $\kappa\in V_n$ such that
$\{\beta,\kappa\}\in E(Q_n)$, $\kappa(i)=0$ and $\kappa,\kappa\oplus e_i\notin\bigcup A$. Set
$$A'=\rho_{i=0}(A\setminus\{\{\alpha,\beta\}\})\cup\{\{\rho_{i=0}(\beta),\rho_{i=0}(\kappa)\}\}.$$
Then $\enc(A')\subseteq\enc(\rho_{i=0}(\bigcup A\cup\{\kappa\}))=\{\gamma\}$, but $\gamma$ is not
encompassed by $A'$ and thus $A'$ is diminishable. Hence $A$ is connectable by Claim~\ref{lemma-jeden-vpravo:claim-B} and the proof of Claim~\ref{lemma-jeden-vpravo:claim-C} is complete.
\end{proof}
\begin{claim}
\label{lemma-jeden-vpravo:claim-D}
If there exists $\{\alpha ,\beta \}\in A$ with $
\beta\notin X$ and 
$\alpha\in X\setminus \{\alpha_1\oplus e_i,\beta_1\oplus e_i\}$, then $A\in\Gamma_n$. 
\end{claim}
\begin{proof}
Without loss of generality we can assume that 
$\{\alpha ,\alpha_1\oplus e_i\}\in E(Q_n)$. Then for every $\nu\in 
V_n$ with 
\begin{align}\label{lemma-jeden-vpravo:claim-D:*}\tag{a}
\text{$\nu (i)=0$,  $\{\alpha ,\nu \}\in E(Q_n)$ and $\nu\ne\alpha_1\oplus e_i$}
\end{align}
we have $
\nu\oplus e_i\notin\bigcup A$. If some $\nu\notin\bigcup A$ and $\enc(\rho_{i=0}((A\setminus \{\alpha 
,\beta \})\cup \{\{\alpha ,\nu \}\}))=\emptyset$ then, by Claim~\ref{lemma-jeden-vpravo:claim-B}, 
$A\in\Gamma_n$. Thus we can restrict to the case that 
\begin{align}\label{lemma-jeden-vpravo:claim-D:**}
\tag{b}
\text{every $\nu\in V_n\setminus\bigcup A$ with \eqref{lemma-jeden-vpravo:claim-D:*} satisfies $\enc(\rho_{i=0}((A \setminus \{\alpha ,\beta \})\cup \{\{\alpha ,\nu \}\}))=\emptyset$.} 
\end{align}
Assume that there are two distinct $\kappa ,\kappa'$ such 
that $\kappa\in\enc(\rho_{i=0}((A\setminus \{\alpha ,\beta \})\cup 
\{\{\alpha ,\nu \}\}))$, 
$\kappa'\in\enc(\rho_{i=0}((A\setminus \{\alpha ,\beta \})\cup \{
\{\alpha ,\nu'\}\}))$ for some $\nu ,\nu'$ satisfying \eqref{lemma-jeden-vpravo:claim-D:*}
$\nu (i)=\nu'(i)=0$ and $\{\alpha ,\nu \},\{\alpha ,\nu'\}\in E(Q_n)$. Recall that by Claim~\ref{lemma-jeden-vpravo:claim-C} we can assume that $\enc(\rho_{i=0}(A))=\emptyset$. It follows that $\{\kappa,\rho_{i=0}(\nu)\},\{\kappa',\rho_{i=0}(\nu')\}\in E(Q_{n-1})$ and therefore $\chi(\kappa)=\chi(\kappa')$.
Then $n-2$ neighbors of both $\kappa$ and $\kappa'$ belong to $\rho_{i=0}(\bigcup A)$, vertices in these neighborhoods 
have the same parity and the intersection of these neighboorhoods is at most doubleton. 
It follows that $\rho_{i=0}(\bigcup A)$ contains at least $2(n-2)-2= 2n-6$ vertices of the same parity. Since $2n-6>n-1$ for $n\ge6$, this is a contradiction with our assumption that $\rho_{i=0}(A)$ consists of vertices of $n-1$ odd pairs.
Thus there is at most one $\nu\in V_n\setminus\bigcup A$ with 
$\nu (i)=0$, $\{\alpha ,\nu \}\in E(Q_n)$, \eqref{lemma-jeden-vpravo:claim-D:*} such that $\enc(\rho_{i=0}((A\setminus \{\alpha ,\beta \})\cup \{\{\alpha ,\nu \}\}))$ is nonepty and actually contains only a single vertex $\kappa$. 
Then \eqref{lemma-jeden-vpravo:claim-D:**} implies that at least $n-3$ neighbors of $\rho_{i=0}(\alpha)$ belong  to $\rho_{i=0}(\bigcup A)$. As $n-2$ neighbors of $\kappa$ also belong  to $\rho_{i=0}(\bigcup A)$, vertices in both neighborhoods share the same parity and their intersection is a doubleton including $\rho_{i=0}(\nu)\not\in\rho_{i=0}(\bigcup A)$, it follows that  $\rho_{i=0}(\bigcup A)$ contains at least $(n-3)+(n-2)-1= 2n-6$ vertices of the same parity. Since $2n-6>n-1$ for $n\ge6$, this leads to a contradiction as above.

It remains to deal with the case that 
\begin{align}\tag{c}
\label{lemma-jeden-vpravo:claim-D:***}
\text{each $\nu\in V_n$ with \eqref{lemma-jeden-vpravo:claim-D:*} falls into $\bigcup A$. }
\end{align}
Then every neighbor $\beta'$ of $\alpha$ with $\beta'(i)=0$ and $\beta'\ne\alpha_1\oplus e_i$ falls into $\bigcup A$. 
Hence there are at least $n-2$ pairs $\{\alpha',\beta'\}\in A$ such that $\{\alpha,\beta'\}\in E(Q_n)$, and $n-3$ of them satisfy $\beta'\not\in X$. Select one such pair $\{\alpha',\beta'\}$ and note that it satisfies the assumptions of this claim. If \eqref{lemma-jeden-vpravo:claim-D:***} holds for $\alpha'$, then --- similarly as above --- the neighborhood of $\rho_{i=0}(\alpha)$ and $\rho_{i=0}(\alpha')$ includes at least $2n-5$ vertices of $\rho_{i=0}(\bigcup A)$, which leads to a contradiction as $2n-5>n-1$ for $n\ge5$. Hence condition \eqref{lemma-jeden-vpravo:claim-D:***} and --- by the previous part of the proof --- therefore also \eqref{lemma-jeden-vpravo:claim-D:**} fails for $\{\alpha',\beta'\}$. It follows that $\{\alpha',\beta'\}$ satisfies the assumptions of Claim~\ref{lemma-jeden-vpravo:claim-B} and therefore $A\in\conp_n$.
The proof of Claim~\ref{lemma-jeden-vpravo:claim-D} is complete. 
\end{proof}
By Claims~\ref{lemma-jeden-vpravo:claim-A} and \ref{lemma-jeden-vpravo:claim-D} we can assume that each pair $\{\alpha,\beta\}\in A\setminus\{\{\alpha_1,\beta_1\}\}$ satisfies
\begin{enumerate}[\upshape(1)]
\item\label{lemma-jeden-vpravo:type-1} either $\alpha,\beta\in X$,
\item\label{lemma-jeden-vpravo:type-2} or $\alpha\in\{\alpha_1\oplus e_i,\beta_1\oplus e_i\}$ while $\beta\not\in X$.
\end{enumerate}
Note that there are at most two pairs of type \eqref{lemma-jeden-vpravo:type-2}.  First assume that there exists such a pair $\{\alpha_1\oplus e_i,\beta\}$ of type \eqref{lemma-jeden-vpravo:type-2}. Since $\beta\not\in X$ implies that $\{\alpha_1\oplus e_i,\beta\}\not\in E(Q_n)$, the remaining vertices of $\rho_{i=0}(\bigcup A)$ may occupy at most $n-2$ out of $n-1$ neighbors of $\rho_{i=0}(\alpha_1\oplus e_i)$. Hence there  is $\kappa\in V_n\setminus\bigcup A$ such that $\kappa(i)=0$ and $\{\alpha_1\oplus e_i,\kappa\}\in E(Q_n)$. Note that then  $\kappa\oplus e_i\notin\bigcup A$ and
$A'=\rho_{i=0}((A\setminus \{\{\alpha_1\oplus e_i ,\beta \}\})\cup \{\{\alpha_1\oplus e_i ,\kappa \}\})$ contains two edge-pairs.  Moreover, as $\bigcup A'\subseteq X$ or $\bigcup A'\subseteq X\cup\{\beta'\}$ --- the latter occurring if there is another pair  $\{\beta_1\oplus e_i,\beta'\}$ of type \eqref{lemma-jeden-vpravo:type-2} --- it follows that $\enc(A')=\emptyset$. Indeed, the only vertex that could be possibly encompassed by $A'$ is $\rho_{i=0}(\beta_1\oplus e_i)$, but as $\chi(\beta_1\oplus e_i)=\chi(\kappa)$ and $(\beta_1\oplus e_i)\oplus e_i=\beta_1\in\bigcup A$, this would imply that $\beta_1\oplus e_i\in\enc(A)$, contrary to our assumption that $A\in\dimp_n$. Hence $A'\in\dimp_{n-1}$ which means that the assumptions of Claim~\ref{lemma-jeden-vpravo:claim-B} are satisfied and therefore $A\in\conp_n$.

It remains to deal with the case that no pairs of type \eqref{lemma-jeden-vpravo:type-2} exist. Then $X=\bigcup A\setminus\{\alpha_1,\beta_1\}$. If $\{\alpha_0,\beta_0\}=\{\alpha_1\oplus e_i,\beta_1\oplus e_i\}$ then  $A\in\conp_n$ by Claim~\ref{lemma-jeden-vpravo:claim-A}, so we can without loss of generality assume that $\{\alpha_0,\beta_0\}\cap\{\alpha_1\oplus e_i,\beta_1\oplus e_i\}=\alpha_1\oplus e_i$. But then $\{\beta_1\oplus e_i,\beta\}$ is another edge-pair in $A$.  Moreover, as $\rho_{i=0}(\bigcup A)=\rho_{i=0}(X)$ implies that $\enc(\rho_{i=0}(\bigcup A))=\emptyset$, we can conclude that
$\rho_{i=0}(A)\in\dimp_{n-1}\subseteq\Gamma_{n-1}$. By Proposition~\ref{old-results-prop}\,\eqref{old-results-prop-part1}, 
$\rho_{i=1}(A)\in\Gamma_{n-1}$ and therefore $A\in\Gamma_n$.  This completes the proof of 
Lemma~\ref{lemma-jeden-vpravo}.
\end{proof}
\begin{lemma}
\label{lemma46}
Let $A\in\dimp_n$ for $n\ge 6$ such that $|A|=n$ and $\sigma_i(A)=(n-2,1)$
for some separating $i\in[n]$.  If $\Upsilon_{n-1}\subseteq\Gamma_{n-1}$ then
$A\in\Gamma_n$.
\end{lemma}
\begin{proof}Assume that $A\in\Upsilon_n$ satisfies the assumptions of the lemma. Then
there exist edge pairs $\{\alpha_0,\beta_0\},$ $\{\alpha_1,\beta_1\}\in A$ with $\alpha_0(i)=\beta_0(i)=0$, $\alpha_1(i)=\beta_1(i)=1$ and a pair $\{\alpha_2,\beta_2\}\in A$ with $
\alpha_2(i)\ne\beta_2(i)$, while
$\alpha (i)=\beta (i)=0$ for every $\{\alpha ,\beta \}\in A\setminus
\{\{\alpha_1,\beta_1\},\{\alpha_2,\beta_2\}\}$.  Assume
that $\alpha_2(i)=0$ and $\beta_2(i)=1$. The proof is divided into several cases.
\setcounter{claim}{0}
\begin{claim}
\label{lemma46:claim-A}
If $\{\alpha_2,\alpha_1\oplus e_i\},\{\alpha_2,\beta_1\oplus e_i\}
\notin E(Q_n)$ then $A\in\Gamma_n$.
\end{claim}
\begin{proof}[Proof of Claim~\ref{lemma46:claim-A}]
From
$\sigma_i(A)=(n-2,1)$ it follows that $|\rho_{i=0}(A)|=n-2$ and hence $
\rho_{i=0}(A)$ is
connectable because $\rho_{i=0}(A)$ is odd. Let $\{P_{\kappa ,\kappa'}
\mid \{\kappa ,\kappa'\}\in\rho_{i=0}(A)\}$
be a connector of $\rho_{i=0}(A)$.  Then there is a~$\{\eta,\eta'\}\in\rho_{i=0}(A)$ and
a subpath $(\zeta ,\rho_{i=0}(\alpha_2),\zeta')$ of the path $P_{
\eta ,\eta'}$ such that $\zeta$ is
closer to $\eta$ than $\zeta'$ in the path $P_{\eta,\eta'}$. Set
$$A^{\prime\prime}=(\rho_{i=0}(A)\setminus \{\{\eta ,\eta'\}\}
)\cup \{\{\eta ,\zeta \},\{\rho_{i=0}(\alpha_2),\rho_{i=0}(\alpha_
2)\},\{\zeta',\eta \}\}$$
then $A^{\prime\prime}$ is connectable because $\{P_{\kappa ,\kappa'}\mid
\{\kappa ,\kappa'\}\in\rho_{i=0}(A)\}$ is a
connector of $\rho_{i=0}(A)$. From $\{\alpha_2,\alpha_1\oplus e_i\},\{\alpha_2,\beta_1\oplus e_i\}\notin E(Q_n)$ it
follows that $\{\zeta ,\zeta',\rho_{i=0}(\alpha_2\}\cap \{\rho_{i=1}(\alpha_1),\rho_{
i=1}(\beta_1)\}=\emptyset$. Then
$$A^{\prime\prime\prime}=\{\{\rho_{i=1}(\alpha_1),\rho_{i=1}(\beta_
1)\},\{\zeta ,\zeta'\},\{\rho_{i=0}(\alpha_2),\rho_{i=1}(\beta_2)
\}\}$$
is a pair-set.  Since $\chi (\zeta )=\chi (\zeta')\ne\chi (\rho_{
i=0}(\alpha_2))=\chi (\rho_{i=1}(\beta_2))$
we infer, by Proposition~\ref{old-results-prop}\;\eqref{old-results-prop-part7},
that $A^{\prime\prime\prime}$ is connectable and from
$\iota_{i,0}(A^{\prime\prime},A^{\prime\prime\prime})\overset {*}{\implies}$A
it follows that $A$ is connectable.
\end{proof}
Thus we can assume that either $\{\alpha_2,\alpha_1\oplus e_i\}\in
E(Q_n)$ or
$\{\alpha_2,\beta_1\oplus e_i\}\in E(Q_n)$.
\begin{claim}
\label{lemma46:claim-B}
If $\rho_{i=0}(\alpha_2)\notin\enco(\bigcup\rho_{i=0}(A)\cup \{\rho_{i=1}(\{\alpha_1,\beta_1\})\})$ then $A$ is connectable.
\end{claim}
\begin{proof}[Proof of Claim~\ref{lemma46:claim-B}]
If there exists a vertex $\gamma\in V_n$ such that 
\begin{enumerate}[\upshape(a)]
\item\label{lemma46:claim-B:a} 
$\{\gamma,\alpha_2\}\in E(Q_n)$, $\gamma(i)=0$, $\gamma,\gamma\oplus e_i\notin\bigcup A$, and 
\item\label{lemma46:claim-B:b} 
$\enc(\rho_{i=0}(\bigcup A\cup\{\gamma\}))=\emptyset$,
\end{enumerate} 
we set
$A'=\rho_{i=0}(A)\cup \{\{\rho_{i=0}(\alpha_2),\gamma \}\}$. Then $A'$ is an odd pair set with at least two edge-pairs, $\enc(A')=\emptyset$ and $|A'|=n-1$, which means that $A'\in\dimp_{n-1}\subseteq\conp_{n-1}$.
Since $\gamma\oplus e_i\notin \{\alpha_1,\beta_1\}$ and $
\chi (\rho_{i=0}(\gamma) )\ne\chi (\rho_{i=0}(\alpha_2))=\chi (\rho_{i=1}(\beta_
2))$, we conclude, by Proposition~\ref{old-results-prop}\,\eqref{old-results-prop-part3}, that
$$A^{\prime\prime\prime}=\{\{\rho_{i=1}(\alpha_1),\rho_{i=1}(\beta_
1)\},\{\gamma ,\rho_{i=1}(\beta_2)\}\}$$
is connectable and $\iota_{i,0}(A^{\prime\prime},A^{\prime\prime\prime}
)\overset {*}{\implies}A$ implies that $A$ is connectable.

Thus our first aim is to find a $\gamma$ satisfying both \eqref{lemma46:claim-B:a} and \eqref{lemma46:claim-B:b}. Since $$\rho_{i=0}(\alpha_2)\notin\enco(\bigcup\rho_{i=0}(A)\cup \{\rho_{i=1}(\{\alpha_1,\beta_1\})\}$$
there exists $\gamma\in V_n$ satisfying \eqref{lemma46:claim-B:a}. 
First assume that there are two distinct vertices $\gamma,\gamma'\in V_n$ satisfying \eqref{lemma46:claim-B:a}.
Since $\chi(\gamma)=\chi(\gamma')$, we have
$$\enc(\rho_{i=0}(\bigcup A\cup\{\gamma\}))\cap\enc(\rho_{i=0}(\bigcup A\cup\{\gamma'\}))=\enc(\rho_{i=0}(\bigcup A)).$$
Thus if $\enc(\rho_{i=0}(A))=\emptyset$ then, by Observation~\ref{obser-encom}, at least one of
$\gamma,\gamma'$ satisfies both \eqref{lemma46:claim-B:a} and \eqref{lemma46:claim-B:b} which means that $A$ is connectable.

Next assume that there is a vertex $\gamma'\in\enc(\rho_{i=0}(\bigcup A))$. Then
$\gamma'\notin\rho_{i=1}\{\alpha_1,\beta_1\}$, for otherwise $\gamma=\iota_{i=0}(\gamma')\in\enc(A)$,
contrary to our assumption that $A$ is diminishable. Since  $|\rho_{i=0}(A)|=n-2$, we conclude that
$\{\alpha_2,\gamma\}\in E(Q_n)$. Then $\gamma,\gamma\oplus e_i\notin\bigcup A$ which means that $\gamma$ satisfies \eqref{lemma46:claim-B:a}. To prove that $\gamma$ satisfies \eqref{lemma46:claim-B:b}, consider $\gamma''\in V_{n-1}$ such that $\gamma''\in\enco(\rho_{i=0}(\bigcup A)\cup\{\gamma'\})$.
Then, by Observation~\ref{obser-encom}, $\chi(\gamma'')\ne\chi(\gamma')$ which means that 
either $\{\gamma'',\gamma'\}\in E(Q_{n-1})$ or $\gamma''\in\enco(\rho_{i=0}(A))$
because $\{\gamma'',\rho_{i=0}(\alpha_2)\}\notin E(Q_{n-1})$. In the first case
$\gamma''\in\rho_{i=0}(\bigcup A)$ because $\gamma'\in\enc(\rho_{i=0}(A))$ while the second case is impossible because
$|\rho_{i=1}(A)|=n-2$. Thus $\gamma$ satisfies \eqref{lemma46:claim-B:b} and therefore $A$ is connectable.

The last case remaining to settle Claim~\ref{lemma46:claim-B} is that $\enc(\rho_{i=0}(A))=\emptyset$ and there exists exactly one $\gamma\in V_n$ satisfying \eqref{lemma46:claim-B:a}. 
Since $\{\eta\in\bigcup A\mid\eta(i)=1\}=\{\alpha_1,\beta_1,\beta_2\}$ and exactly one of these three vertices shares its parity with $\chi(\gamma\oplus e_i)$, we conclude that $|\{\eta\in\bigcup A\mid\eta(i)=0,\,\{\eta,\alpha_2\}\in E(Q_n)\}|\in\{n-2,n-3\}$. Hence for every
$\gamma'\in V_n\setminus\{\alpha_2\}$ 
we have
\[|
\{\eta\in\bigcup A\,\cup\,\{\gamma\}\mid\eta(i)=0,\,\{\eta,\gamma'\}\in E(Q_n)\}|\le
\begin{cases}
4 & \text{if $\chi(\gamma')=\chi(\alpha_2)$}\\
n-2 & \text{if $\chi(\gamma')\ne\chi(\alpha_2)$}.
\end{cases}
\] 
Since $\max\{4,n-2\}<n-1$, it follows that $\enc(\rho_{i=0}(\bigcup A\cup\{\gamma\}))=\emptyset$. Therefore $\gamma$ satisfies both \eqref{lemma46:claim-B:a} and \eqref{lemma46:claim-B:b} 
and hence $A$ is connectable. The proof of Claim~\ref{lemma46:claim-B} is complete.
\end{proof}
It remains to settle the case that 
$\rho_{i=0}(\alpha_2)\in\enc(\bigcup\rho_{i=0}(A)\cup\rho_{i=1}(\{
\alpha_1,\beta_1\})$ and 
$\{\{\alpha_2,\alpha_1\oplus e_i\},\{\alpha_2,\beta_1\oplus e_i\}
\}\cap E(Q_n)\ne\emptyset$. We can without loss of 
generality assume that $\{\alpha_2,\alpha_1\oplus e_i\}\in E(Q_n)$. Since $
|\rho_{i=0}(A)|=n-2$, we conclude that every $\nu\in V_n\setminus \{\alpha_1\oplus e_i\}$ with $
\{\alpha_2,\nu \}\in E(Q_n)$ 
belongs to $\bigcup A$ and $\alpha_1\oplus e_i\notin\bigcup A$. Since $A\in\dimp_n$, it follows that $\alpha_1\oplus e_i$ is not encompassed 
by $\bigcup A$. 
First consider that $\alpha_2=\beta_2\oplus e_i$. Then 
$A'=\rho_{i=0}(A)\cup \{\{\rho_{i=0}(\alpha_2),\rho_{i=1}(\alpha_1)\}\}$ is an odd pair-set with two 
edge-pairs. Since $\rho_{i=0}(\alpha_2)\in\enc(\bigcup\rho_{i=0}(A)\cup\rho_{i=1}(\{\alpha_1,\beta_1\})$, we conclude that $\enc(A')=\enc(\rho_{i=0}(A))$ which is empty, as $n-2$ odd pairs of $\rho_{i=0}(A)$ is not enough to encompass a vertex of $V_{n-1}$. It follows that $A'$ is 
diminishable and hence connectable. Let $\{P_{\kappa ,\tau}\mid \{\kappa ,\tau \}\in A'\}$ be a 
connector of $A'$. Then $P_{\rho_{i=0}(\alpha_2),\rho_{i=1}(\alpha_1)}=(\rho_{i=0}(\alpha_2),\rho_{i=1}(\alpha_1))$ and there exists $\{\nu ,\eta \}\in A'$ such that $(\zeta ,\rho_{i=1}(\beta_1),\xi )$ is a subpath of $P_{\nu ,\eta}$ 
such that $\zeta$ is closer to $\nu$ than $\xi$ in $P_{\nu ,\eta}$. Then 
\begin{align*} A^{\prime\prime}=(\rho_{i=0}(A)\setminus \{\{\nu ,\eta \}\})\cup& 
\{\{\rho_{i=1}(\alpha_1),\rho_{i=1}(\alpha_1)\},\{\rho_{i=1}(\beta_
1),\rho_{i=1}(\beta_1)\},\\&\{\rho_{i=0}(\alpha_2),\rho_{i=0}(\alpha_
2)\},\{\nu ,\zeta \},\{\xi ,\eta \}\}
\end{align*}
is connectable. Let 
$$A^{\prime\prime\prime}=\{\{\rho_{i=1}(\alpha_1),\rho_{i=1}(\alpha_
1)\},\{\rho_{i=1}(\beta_1),\rho_{i=1}(\beta_1)\},\{\rho_{i=0}(\alpha_
2),\rho_{i=0}(\alpha_2)\},\{\zeta ,\xi \}\}.$$
Since $\chi (\rho_{i=0}(\alpha_2))\ne\chi (\rho_{i=1}(\alpha_1))$ and $
\chi (\zeta )=\chi (\xi )\ne\chi (\rho_{i=1}(\beta_1))$, we 
conclude that $A^{\prime\prime\prime}$ is even and balanced. 
Since $A^{\prime\prime\prime}$ is connectable by Proposition~\ref{old-results-prop}\,\eqref{old-results-prop-part8}, $\iota_{i,0}(A^{\prime
\prime},A^{\prime\prime\prime})\overset{*}{\implies}A$ implies that $A$ is connectable as well. 

Secondly assume that $\alpha_2\ne\beta_2\oplus e_i$. Then the intersection 
of the neighborhood of $\beta_2\oplus e_i$ with $\bigcup A_1$ where $A_1=A\setminus\{\alpha_1,\beta_1,\beta_2\}$ contains at most two 
vertices because $\chi (\alpha_2)=\chi (\beta_2\oplus e_i)$ and all vertices of $\bigcup A_1$ of parity $-\chi(\alpha_2)$ are neighbors of $\alpha_2$. Thus there exists $
\gamma\in V_n$ 
such that $\gamma (i)=0$, $\{\gamma ,\beta_2\oplus e_i\}\in E(Q_n
)$, $\gamma ,\gamma\oplus e_i\notin\bigcup A$ and 
$\enc(\rho_{i=0}(\bigcup A\cup \{\gamma \}))=\emptyset$. Choose $
\{\alpha ,\beta \}\in A_1\setminus\{\{\alpha_0,\beta_0\}\}$ such that $\{\alpha ,\alpha_2\}\in E(Q_n
)$ and there is $\nu\in V_n$ such that 
$\{\alpha ,\nu \}\in E(Q_n)$, $\nu ,\nu\oplus e_i\notin\bigcup A$ and $\enc
(\rho_{i=0}(\bigcup A\cup \{\nu \}))=\emptyset$. Let 
$$A'=(\rho_{i=0}(A\setminus \{\{\alpha ,\beta \}\}))\cup \{\{\alpha 
,\nu \},\{\alpha_2,\gamma \}\}.$$
Then $A'$ is diminishable and hence it is connectable. Let 
$\{P_{\eta ,\tau}\mid \{\eta ,\tau \}\in A'\}$ be a connector of $
A'$. Then there exists 
$\{\lambda ,\omega \}\in A'$ such that $(\zeta ,\rho_{i=0}(\beta 
),\xi )$ is a subpath of $P_{\lambda ,\omega}$ and $\zeta$ is 
closer to $\lambda$ than $\xi$ in $P_{\lambda ,\omega}$. Then 
$$A^{\prime\prime}=(A'\setminus \{\{\lambda ,\omega \}\})\cup \{\{
\lambda ,\zeta \},\{\xi ,\omega \},\{\rho_{i=0}(\beta ),\rho_{i=0}
(\beta )\}\}$$
is connectable. Let 
$$A^{\prime\prime\prime}=\{\{\rho_{i=1}(\alpha_1),\rho_{i=1}(\beta_
1)\},\{\rho_{i=0}(\gamma ),\rho_{i=1}(\beta_2)\},\{\rho_{i=0}(\beta 
),\rho_{i=0}(\nu )\},\{\zeta ,\xi \}\}.$$
Then $$\chi (\zeta )=\chi (\xi )\ne\chi (\rho_{i=0}(\beta ))=\chi (\rho_{
i=0}(\nu ))$$ 
and $\{\rho_{i=1}(\alpha_1),\rho_{i=1}(\beta_1)\}$ and $\{\rho_{
i=0}(\gamma ),\rho_{i=1}(\beta_2)\}$ are edge-pairs. 
Note that this means that there is $i'\in[n-1]$ such that  $\rho_{i=1}(\alpha_1)(i')=\rho_{i=1}(\beta_1)(i')\ne\rho_{i=0}(\gamma)(i')=\rho_{i=1}(\beta_2)(i')$.
Hence, by Lemma~\ref{2-even2-edgepairs}\,\eqref{2-even2-edgepairs:1}, $A^{\prime\prime\prime}$ is 
connectable and from $\iota_{i,0}(A^{\prime\prime},A^{\prime\prime
\prime})\overset{*}\implies A$ it follows that $A\in\Gamma_
n$.
\end{proof}
The last lemma tackles diminishable pair-sets with a bad coordinate. 
\begin{lemma}
\label{lemma47}
Let $A\in\Upsilon_n$ for $n\ge 6$
such that  $i\in [n]$ is bad for $A$. If $\Upsilon_{n-1}\subseteq\Gamma_{n-1}$ then $A\in\Gamma_n$.
\end{lemma}
\begin{proof}
Since $i\in [n]$ is bad for a diminishable pair-set $A$
we conclude that $|A|=n$ and there exist
$\{\alpha_0,\beta_0\},\{\alpha_1,\beta_1\},\{\alpha_2,\beta_2\},\{\alpha_3,\beta_3\}\in A$ such that
$\alpha_0(i)=\beta_0(i)=\alpha_2(i)=\alpha_3(i)=0$, $\alpha_1(i)=
\beta_1(i)=\beta_2(i)=\beta_3(i)=1$,
$\chi (\alpha_2)\ne\chi (\alpha_3)$, $\alpha (i)=\beta (i)=0$ for all
$\{\alpha ,\beta \}\in A\setminus \{\{\alpha_j,\beta_j\}\mid j\in[4]\}$. Furthermore,
\begin{enumerate}[\upshape(a)]
\item\label{lemma47:a}
for every $\kappa\in V_n$ such
that $\kappa (i)=0$ and $\{\kappa ,\alpha_j\}\in E(Q_n)$ for some $
j\in \{2,3\}$ we have that
$\{\kappa ,\kappa\oplus e_i\}\cap\bigcup A\ne\emptyset$. 
\end{enumerate}
Since $
|\{\kappa\in\bigcup A\mid\kappa (i)=0,\,\chi (\kappa )=l\}|=n-2$,
for every $l\in \{-1,1\}$ and $\{\kappa\in\bigcup A\mid\kappa (i)
=1\}=\{\alpha_1,\beta_1,\beta_2,\beta_3\}$, we
obtain that for both $j\in \{2,3\}$,
\begin{align}\tag{b}
\label{lemma47:b}
\begin{split}
&|\{\kappa\in V_n\mid\kappa (i)=0,\,\{\kappa ,\alpha_j\}\in E(Q_n),\,\kappa\in\bigcup A\}|\in \{n-3,n-2\},\\
&|\{\kappa\in V_n\mid\kappa (i)=0,\,\{\kappa ,\alpha_j\}\in E(Q_n),\,\kappa\oplus e_i\in\bigcup A\}|\in \{1,2\}.
\end{split} 
\end{align}
Moreover, if
$|\{\kappa\in\bigcup A\mid\kappa (i)=0,\,\{\kappa ,\alpha_j\}\in E(Q_n)\}|=n-2$ then
$\{\alpha_2,\alpha_3\}\in E(Q_n)$. Therefore for every $\kappa\in V_n$, 
\begin{align}\tag{c}
\label{lemma47:c}
\begin{split}
&\text{if $\kappa (i)=0$  and $\kappa\ne
\alpha_2,\alpha_3$
then $|\{\gamma\in \bigcup A\mid \{\gamma ,\kappa \}\in E(Q_n),\,\gamma
(i)=0\}|\le 3$},\\
&\text{if $\kappa (i)=1$ then
$|\{\gamma\in \bigcup A\mid \{\gamma ,\kappa \}\in E(Q_n),\,\gamma (i)=
1\}|\le 2$}.
\end{split} 
\end{align}
These facts are exploited later in the
proof to show that some vertices are not encompassed by a~suitable set.

Let $A'=\rho_{i=0}\big((A\setminus \{\{\alpha_2,\beta_2\},\{\alpha_
3,\beta_3\}\})\cup \{\alpha_2,\alpha_3\}\big)$.  Then
$A'$ is an odd pair-set with $|A'|=n-2$ which means that $A'\in\dimp_{n-1}\subseteq\conp_{n-1}$. 
Let $\{P_{\kappa ,\kappa'}\mid \{\kappa ,\kappa'\}\in A'\}$ be a connector
of $A'$.  For simplicity let us denote $\alpha_j'=\rho_{i=\alpha_j(i)}(\alpha_j)$ and $\beta'_j=\rho_{i=\beta_j(i)}(\beta_j)$ for $j\in[4]$.  Observe that if $\{\alpha'_j,\beta'_{j'}\}\in E(Q_{n-1})$ for
some $j,j'\in \{2,3\}$ then $j\ne j'$. 
\setcounter{claim}{0}\begin{claim}
\label{lemma47:claim-A}
If $\{\alpha'_2,\alpha_1'\},\{\alpha'_3,\beta_1'\}$ are edges of
$P_{\alpha'_2,\alpha'_3}$ then $A\in\Gamma_n$. Analogously if
$\{\alpha'_2,\beta_1'\}, \{\alpha'_3,\alpha_1'\}$ are edges of
$P_{\alpha'_2,\alpha'_3}$ then $A\in\Gamma_n$.
\end{claim}

\begin{proof}[Proof of Claim~\ref{lemma47:claim-A}]
Since $\{P_{\kappa,\kappa'}\mid\{\kappa,\kappa'\}\in A'$ is a connector of $A'$ we conclude that
$$A''=(A'\setminus\{\{\alpha'_2,\alpha'_3\}\})\cup\{\{\alpha'_2,\alpha'_2\},\{\alpha'_3,\alpha'_3\},\{\alpha_1',\beta_1'\}\}$$
is connectable. If $\{\alpha'_2,\alpha'_3\}\ne\{\beta'_2,\beta'_3\}$ then
$$A'''=\{\{\alpha_1',\alpha_1'\},\{\beta_1',\beta_1'\},\{\alpha'_2,\beta'_2\},\{\alpha'_3,\beta'_3\}\}$$
is a balanced pair-set because $\chi(\alpha_1)\ne\chi(\beta_1)$,
$\chi(\alpha'_2)=\chi(\beta'_2)\ne\chi(\alpha'_3)=\chi(\beta'_3)$. By Proposition~\ref{old-results-prop}\,\eqref{old-results-prop-part8}, $A'''$ is
connectable and from $\iota_{i,0}(A'',A''')\overset{*}{\implies}A$ it follows that $A$ is connectable as well.
If $\{\alpha'_2,\alpha'_3\}=\{\beta'_2,\beta'_3\}$ then $\alpha'_2=\beta'_2$ and $\alpha'_3=\beta'_3$.
Choose an edge $\{\xi,\xi'\}$ in a~path $P_{\kappa,\kappa'}$ for some
$\{\kappa,\kappa'\}\in A'\setminus\{\{\alpha'_2,\alpha'_3\}\}$ such that $\xi$ is closer to $\kappa$
than $\xi'$ in $P_{\kappa,\kappa'}$. Set
\[
\begin{split}
A''=&(A'\setminus\{\{\alpha'_2,\alpha'_3\}\})\cup\{\{\alpha'_2,\alpha'_2\},\{\alpha'_3,\alpha'_3\},\{\alpha_1',\beta_1'\},\{\kappa,\xi\},\{\xi',\kappa'\}\}\\
A'''=&\{\{\alpha_1',\alpha'_1\},\{\beta_1',\beta_1'\},\{\beta'_2,\beta'_2\},\{\beta'_3,\beta'_3\},\{\xi,\xi'\}\}
\end{split}
\]
Then $A''$ is connectable because $\{P_{\eta,\eta'}\mid\{\eta,\eta'\}\in A''\}$ is a connector of $A''$
while $A'''$ is a balanced pair-set with four degenerated pairs and therefore connectable by Proposition~\ref{old-results-prop}\;\eqref{old-results-prop-part10}. From $\iota_{i,0}(A'',A''')\overset{*}{\implies}A$ it follows that $A$ is
connectable and the proof of Claim~\ref{lemma47:claim-A} is complete.
\end{proof}
\begin{claim}
\label{lemma47:claim-B}
If $\{\alpha_j',\beta'_{5-j})\}$ is an edge of $P_{\alpha_2',\alpha'_3}$ and
$\alpha_j\oplus e_i\notin \{\alpha_1,\beta_1\}$ for some $j\in \{2,3\}$ then $A\in\Gamma_n$.
\end{claim}
\begin{proof}[Proof of Claim~\ref{lemma47:claim-B}]
Let $A^{\prime\prime}=(A'\setminus \{\{\alpha'_2,\alpha'_
3\}\})\cup \{\{\alpha'_j,\alpha'_j\},\{\alpha'_{5-j},\beta_{5-j}'
\}\}$,
then $A^{\prime\prime}$ is connectable because $\{P_{\kappa ,\kappa'}
\mid \{\kappa ,\kappa'\}\in A\}$ is a connector
of $A'$. Set
$$A^{\prime\prime\prime}=\{\{\alpha'_1,\beta'_1\}\{\alpha'_j,\beta'_
j\},\{\beta'_{5-j},\beta'_{5-j}\}\}.$$
Since $\alpha_j\oplus e_i\notin \{\alpha_1,\beta_1\}$ and
$\chi (\alpha'_j)=\chi (\beta'_j)\ne\chi (\beta'_{5-j})$, we conclude that $A'''$ is a balanced pair-set and,
by Proposition~\ref{old-results-prop}\;\eqref{old-results-prop-part7}, $A^{\prime\prime\prime}$ is
connectable. From
$\iota_{i,0}(A^{\prime\prime},A^{\prime\prime\prime})\overset {*}{\implies}A$ it follows that
$A\in\Gamma_n$ and the proof of Claim~\ref{lemma47:claim-B} is complete.
\end{proof}
Note that if assumptions of Claim~\ref{lemma47:claim-A} do not hold, then \eqref{lemma47:a} implies that
$\{\alpha'_j,\beta'_{5-j}\}$ for some $j\in\{2,3\}$ is an edge of the path
$P_{\alpha'_2,\alpha'_3}$. Then the edge $\{\alpha'_{5-j},\zeta\}$ of the path
$P_{\alpha'_2,\alpha'_3}$ satisfies $\zeta\in\{\alpha'_1,\beta'_1,\beta'_j\}$. 
If $\zeta\in\{\alpha'_1,\beta'_1\}$ then $\alpha_j\oplus e_i\notin\{\alpha_1,\beta_1\}$ and,
by Claim~\ref{lemma47:claim-B}, $A$ is connectable. If $\zeta=\beta'_j$ then, by
Claim~\ref{lemma47:claim-B}, $A$ is connectable whenever
$\{\alpha_2\oplus e_i,\alpha_3\oplus e_i\}\ne\{\alpha_1,\beta_1\}$. The following fact strengthens
this observation.

\begin{claim}
\label{lemma47:claim-C}
If the path $P_{\alpha'_2,\alpha'_3}$ is of length at least $5$ then $A$ is connectable.
\end{claim}

\begin{proof}[Proof of Claim~\ref{lemma47:claim-C}]
By Claims~\ref{lemma47:claim-A} and \ref{lemma47:claim-B} we can assume that $\{\alpha_1,\beta_1\}=\{\alpha_2\oplus e_i,\alpha_3\oplus e_i\}$.
Since the length of $P_{\alpha'_2,\alpha'_3}$ is at least $5$, there exists an
edge $\{\nu,\nu'\}$ of the path $P_{\alpha'_2,\alpha'_3}$ such that $\nu$ is closer to $\alpha'_2$
than $\nu'$ in $P_{\alpha'_2,\alpha'_3}$ and neither $\{\alpha'_2,\nu\}$ nor $\{\nu',\alpha'_3\}$ is an edge of $P_{\alpha'_2,\alpha'_3}$.
Hence
$$A''=(A'\setminus\{\{\alpha'_2,\alpha'_3\}\})\cup\{\{\alpha'_2,\nu\},\{\nu',\alpha'_3\}\}$$
is connectable. Set
$$A'''=\{\{\alpha'_1,\beta'_1\},\{\nu,\beta'_2\},\{\nu',\beta'_3\}\}.$$
From  $\{\alpha_1,\beta_1\}=\{\alpha_2\oplus e_i,\alpha_3\oplus e_i\}$ we conclude that $A'''$ is a pair-set. Then either $A'''$ is odd and therefore connectable by
Proposition~\ref{old-results-prop}\;\eqref{old-results-prop-part6}, or
$A'''$ is balanced with $\|A'''\|=1$ and therefore connectable by
Proposition~\ref{old-results-prop}\;\eqref{old-results-prop-part7}.
Then $\iota_{i,0}(A'',A''')\overset{*}{\implies}A$ completes the proof of Claim~\ref{lemma47:claim-C}.
\end{proof}
Thus we can assume that either
$P_{\alpha'_2,\alpha'_3}=(\alpha'_2,\beta'_3,\beta'_2,\alpha'_3)$ and $\{\alpha_2\oplus e_i,\alpha_3\oplus e_i\}=\{\alpha_1,\beta_1\}$ or
$P_{\alpha'_2,\alpha'_3}=(\alpha'_2,\alpha'_3)$. 
\begin{claim}
\label{lemma47:claim-D}
If $P_{\alpha'_2,\alpha'_3}=(\alpha'_2,\alpha'_3)$ and $\{\alpha'_2,\alpha'_3\}\cap\{\alpha'_1,\beta'_1\}=\emptyset$ then $A$ is connectable.
\end{claim}
\begin{proof}[Proof of Claim~\ref{lemma47:claim-D}]
Since $\{P_{\kappa,\kappa'}\mid\{\kappa,\kappa'\}\in A'\}$ is a connector of $A'$
and $P_{\alpha'_2,\alpha'_3}=(\alpha'_2,\alpha'_3)$,  we conclude that
$$A''=(A'\setminus\{\{\alpha'_2,\alpha'_3\}\})\cup\{\{\alpha'_2,\alpha'_2\},\{\alpha'_3,\alpha'_3\}\}$$
is connectable. By Proposition~\ref{old-results-prop}\:\eqref{old-results-prop-part7},
$$A'''=\{\{\alpha'_1,\beta'_1\},\{\alpha'_2,\beta'_2\},\{\alpha'_3,\beta'_3\}\}$$
is connectable because $\chi(\alpha'_2)=\chi(\beta'_2)\ne\chi(\alpha'_3)=\chi(\beta'_3)$. Then
$\iota_{i,0}(A'',A''')\overset{*}{\implies}A$ completes the proof.
\end{proof}
Thus it remains to settle the cases
\begin{itemize}
\item
$P_{\alpha'_2,\alpha'_3}=(\alpha'_2,\alpha'_3)$ and
$\{\alpha'_2,\alpha'_3\}\cap\{\alpha'_1,\beta'_1\}\ne\emptyset$;
\item
$P_{\alpha'_2,\alpha'_3}=(\alpha'_2,\beta'_3,\beta'_2,\alpha'_3)$ and
$\{\alpha'_2,\alpha'_3\}=\{\alpha'_1,\beta'_1\}$
\end{itemize}
because otherwise $A$ is connectable. 
We can assume that $\alpha'_3=\beta'_1$. Then
$\alpha_1\oplus e_i\ne\alpha_2\oplus e_j$ for all $j\in[n]$. Therefore if $\eta\in V_n$ such that
$\{\alpha_2,\eta\}\in E(Q_n)$ and $\eta\oplus e_i\in\{\alpha_1,\beta_1,\beta_2,\beta_3\}$, then
$\eta=\beta_3\oplus e_i$. Considering \eqref{lemma47:a} and  \eqref{lemma47:b}, it follows that
$$
\{\eta\in V_n\setminus\bigcup A\mid\{\eta,\alpha_2\}\in E(Q_n)\}=\{\beta_3\oplus e_i\}.
$$
Since $n\ge6$ we can choose
$\gamma,\gamma'\in V_{n-1}\setminus\rho_{i=0}(\bigcup A)$ such that
$\chi(\gamma)\ne\chi(\alpha'_2)$, $\{\gamma,\alpha'_2\}\notin E(Q_{n-1})$ and
$\{\gamma',\alpha'_3\}\in E(Q_{n-1})$. From \eqref{lemma47:c} we deduce that $\enc(\rho_{i=0}(\bigcup(A))\cup\{\gamma,\gamma'\})=\emptyset$.
Set $$A_1=(A'\setminus\{\{\alpha'_2,\alpha'_3\}\})\cup\{\{\alpha'_2,\gamma\},\{\alpha'_3,\gamma'\}\}.$$
Then $A_1$ is an odd pair-set, $\enc(A_1)=\enc(\rho_{i=0}(\bigcup(A))\cup\{\gamma,\gamma'\})=\emptyset$,
$|A_1|=n-1$ and $A_1$ contains two edge pairs $\{\alpha'_0,\beta'_0\}$,
$\{\alpha'_3,\gamma'\}$, thus $A_1\in\dimp_{n-1}\subseteq\conp_{n-1}$. 
Let $\{P_{\kappa,\kappa'}\mid\{\kappa,\kappa'\}\in A_1\}$ be a connector of $A_1$.
First assume that there exist $\{\eta,\eta'\}\in A_1\setminus\{\{\alpha'_3,\gamma'\}\}$ and a subpath $(\zeta,\xi,\xi',\zeta')$ of
$P_{\eta,\eta'}$ such that $\{\xi,\gamma'\}\in E(Q_{n-1})$,
$\zeta,\zeta',\xi'\notin\rho_{i=1}(\bigcup A)$ and $\zeta$ is closer to $\eta$ than $\zeta'$
in $P_{\eta,\eta'}$. Thus
$$A''=(A_1\setminus\{\{\alpha'_3,\gamma'\}\})\cup\{\{\alpha'_3,\xi'\},\{\eta,\zeta\},\{\zeta',\eta'\}\}$$
is connectable. Set
$$A'''=\{\{\alpha'_1,\beta'_1\},\{\zeta,\zeta'\},\{\gamma,\beta'_2\},\{\xi',\beta'_3\}\}.$$
Since $\chi(\zeta)\ne\chi(\zeta')$, $\chi(\gamma)\ne\chi(\alpha'_2)=\chi(\beta'_2)$ and
$\chi(\xi')\ne\chi(\xi)=\chi(\alpha'_3)=\chi(\beta'_3)$, we conclude that $A'''\in\oddp_{n-1}$ and from $n\ge6$ and $|A'''|=4$ it follows that $A'''\in\dimp_{n-1}\subseteq\conp_{n-1}$.

Otherwise it must be the case that the path $P_{\alpha'_3,\gamma'}$ from the connector of $A_1$ passes through at least four neighbors of $\gamma'$ in $V_{n-1}\setminus\bigcup A_1$. Then $P_{\alpha'_3,\gamma'}=(\alpha'_3,\zeta',\xi',\xi,\zeta,\dots,\gamma')$ and --- regarding \eqref{lemma47:a} and \eqref{lemma47:b} --- we have  $\{\zeta',\gamma'\}=\{\alpha'_1,\beta'_2\}$. 
Note that we can assume that $\zeta'\ne\beta'_3$, for otherwise the path $(\alpha'_3,\zeta',\xi',\xi,\zeta,\dots,\gamma')$ may be replaced with the path $(\alpha'_3,\gamma',\dots,\zeta,\xi,\xi',\zeta')$ which possesses the desired property.
Assuming that $\zeta'=\alpha'_1$ and $\gamma'=\beta'_2$ (the other case is analogic), observe that then
$$A''=(A_1\setminus\{\{\alpha'_3,\gamma'\}\})\cup\{\{\alpha'_3,\xi'\},\{\gamma',\xi\}\}$$
is connectable. Set
$$A'''=\{\{\alpha'_1,\beta'_1\},\{\xi,\alpha'_2\},\{\xi',\beta'_3\}\}.$$
Since $\chi(\xi)=\chi(\alpha'_2)\ne\chi(\xi')=\chi(\beta'_3)$, we conclude that $A'''$ is connectable by Proposition~\ref{old-results-prop}\,\eqref{old-results-prop-part7}.

In both cases, $\iota_{i,0}(A'',A''')\overset{*}{\implies}A$ implies that $A$ is connectable, which completes the proof of Lemma~\ref{lemma47}.
\end{proof}
At this point we are ready to gather the fruit of our efforts. The basis of our inductive construction is described as Theorem~\ref{dimen-5}. The general inductive step is formulated as part~\eqref{cor:main-section3:4} of Corollary~\ref{cor:main-section3},  while special cases \eqref{cor:main-section3:1}-\eqref{cor:main-section3:3} are settled by Lemmas~\ref{lem-4-1}, \ref{lem-4-2} and \ref{lem-4-3}, Lemmas~\ref{lemma-jeden-vpravo} and \ref{lemma46}, and Lemma~\ref{lemma47}, respectively. 
\begin{theorem}
\label{thm:main}
$\dimp_n\subseteq\conp_n$ for every $n\ge5$.
\end{theorem}
Putting together Theorem~\ref{thm:main} with Propositions~\ref{dimen-3} and \ref{dimen-4} that settle the case of dimension $n<5$, we can fulfill our main goal and characterize odd pair-sets that are connectable.
\begin{corollary}
Let $A\in\oddp_n^{n-1}$. Then $A$ is connectable if and only if either $n\ne4$
or $n=4$ and $A$ is not isomorphic to the pair-set $C_2$ on Fig.~\ref{fig:Fig2}.
\end{corollary}
\section*{Acknowledgements}
\noindent
This research was supported by the Czech Science Foundation under grant GA14-10799S.
 The first and third authors would like to dedicate this work to the memory of their teacher, colleague and friend Va\v{s}ek Koubek whose sudden decease made this paper his last work.

\end{document}